%% file: main.tex
\documentclass[
11pt, % The default document font size, options: 10pt, 11pt, 12pt
english, % ngerman for German
singlespacing, % Single line spacing, alternatives: onehalfspacing or doublespacing
headsepline, % Uncomment to get a line under the header
]{DoctoralThesis} % The class file specifying the document structure

\usepackage[utf8]{inputenc} % Required for inputting international characters
\usepackage[T1]{fontenc} % Output font encoding for international characters
\usepackage{booktabs}% http://ctan.org/pkg/booktabs
\usepackage{url}
\usepackage{graphicx}
\usepackage{titlesec}
\usepackage[labelsep=period]{caption}
\usepackage{enumitem}
\usepackage[most]{tcolorbox}
\usepackage{amssymb}
\usepackage{amsthm}
\usepackage[ruled,vlined]{algorithm2e}
\usepackage[backend=bibtex8,style= authoryear-comp,natbib=true,sorting=none,bibstyle=numeric,maxbibnames=99]{biblatex} % Use the bibtex backend with the authoryear citation style (which resembles APA)
\usepackage{subfigure}

\newtheorem{definition}{Definition}[chapter]
\newtheorem{theorem}{Theorem}[chapter]
\newtheorem{example}{Example}[chapter]
\newtheorem{proposition}{Proposition}[chapter]
\newtheorem{assumption}{Assumption}[chapter]
\newlist{abbrv}{itemize}{1}
\setlist[abbrv,1]{label=,labelwidth=1in,align=parleft,itemsep=0.1\baselineskip,leftmargin=!}

\usepackage{float}
\usepackage{optidef}
\titleformat{\paragraph}
{\normalfont\normalsize\bfseries}{\theparagraph}{1em}{}
\titlespacing*{\paragraph}
{0pt}{3.25ex plus 1ex minus .2ex}{1.5ex plus .2ex}

\usepackage[autostyle=true]{csquotes} % Required to generate language-dependent quotes in the bibliography

\usepackage{setspace}
\addtocontents{toc}{\protect\setcounter{tocdepth}{2}}

\thesistitle{Machine Learning Approaches to Automated Mechanism Design for Public Project Problem} % Your thesis title, this is used in the title and abstract, print it elsewhere with \ttitle
\supervisor{\textsc{Dr. Mingyu Guo, Dr. Wei Zhang}} % Your supervisor's name, this is used in the title page, print it elsewhere with \supname
\examiner{ } % Your examiner's name, this is not currently used anywhere in the template, print it elsewhere with \examname
\degree{Doctor of Philosophy} % Your degree name, this is used in the title page and abstract, print it elsewhere with \degreename
\author{\textsc{Guanhua Wang}} % Your name, this is used in the title page and abstract, print it elsewhere with \authorname
\addresses{ } % Your address, this is not currently used anywhere in the template, print it elsewhere with \addressname

\subject{Computer Science} % Your subject area, this is not currently used anywhere in the template, print it elsewhere with \subjectname, such as Computer Science
 % Keywords for your thesis, this is not currently used anywhere in the template, print it elsewhere with \keywordnames
\university{\href{http://www.adelaide.edu.au}{University of Adelaide}} % Your university's name and URL, this is used in the title page and abstract, print it elsewhere with \univname
\department{\href{https://cs.adelaide.edu.au}{School of Computer Science}} % Your department's name and URL, this is used in the title page and abstract, print it elsewhere with \deptname, CS is the department name
\group{\href{https://cs.adelaide.edu.au/research/cseducation/}{Computer Science Education Research}} % Your research group's name and URL, this is used in the title page, print it elsewhere with \groupname
\faculty{\href{http://faculty.university.com}{Faculty Name}} % Your faculty's name and URL, this is used in the title page and abstract, print it elsewhere with \facname

\AtBeginDocument{
\hypersetup{pdftitle=\ttitle} % Set the PDF's title to your title
\hypersetup{pdfauthor=\authorname} % Set the PDF's author to your name
\hypersetup{pdfkeywords=\keywordnames} % Set the PDF's keywords to your keywords
}

\begin{document}
\frontmatter % Use roman page numbering style (i, ii, iii, iv...) for the pre-content pages

\pagestyle{plain} % Default to the plain heading style until the thesis style is called for the body content

%----------------------------------------------------------------------------------------
%	TITLE PAGE
%----------------------------------------------------------------------------------------

\begin{titlepage}
\begin{center}

\includegraphics{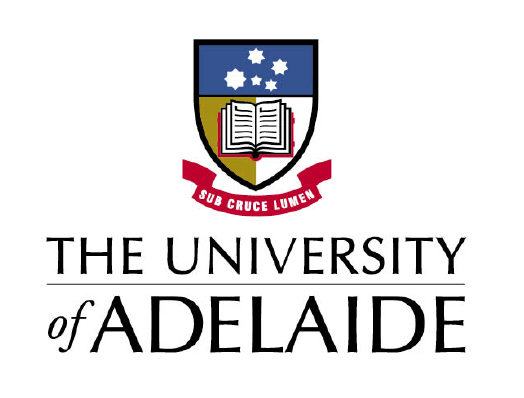}\\
\vspace{0.8cm}{\huge \bfseries \ttitle\par}\vspace{0.8cm}

{Guanhua Wang}
\vfill

\large A thesis submitted for the degree of\\ DOCTOR OF PHILOSOPHY\\The University of Adelaide\\[0.3cm] % University requirement text

\vfill

{\large \today}\\[4cm] % Date

\vfill
\end{center}
\end{titlepage}

%----------------------------------------------------------------------------------------
%	LIST OF CONTENTS/FIGURES/TABLES PAGES
%----------------------------------------------------------------------------------------

\tableofcontents % Prints the main table of contents
\listoffigures % Prints the list of figures
\listoftables % Prints the list of tables
\include{Appendices/Abstract}

\include{Appendices/Abbreviations}
\include{Appendices/Declaration}
\include{Appendices/Acknowledge}
\include{Appendices/Publications}

%----------------------------------------------------------------------------------------
%	THESIS CONTENT - CHAPTERS
%----------------------------------------------------------------------------------------

\mainmatter % Begin numeric (1,2,3...) page numbering

\pagestyle{thesis} % Return the page headers back to the "thesis" style

% Include the chapters of the thesis as separate files from the Chapters folder
% Uncomment the lines as you write the chapters

\include{Chapters/Chapter1}
\include{Chapters/Chapter2}
\include{Chapters/Chapter3}
\include{Chapters/Chapter4}
\include{Chapters/Chapter5}
\include{Chapters/Chapter6}

\include{Chapters/Chapter7}
%----------------------------------------------------------------------------------------
%	THESIS CONTENT - APPENDICES
%----------------------------------------------------------------------------------------

\appendix % Cue to tell LaTeX that the following "chapters" are Appendices

% Include the appendices of the thesis as separate files from the Appendices folder
% Uncomment the lines as you write the Appendices

%\include{Appendices/AppendixTemplate}

%----------------------------------------------------------------------------------------
%	BIBLIOGRAPHY
%----------------------------------------------------------------------------------------

\printbibliography[heading=bibintoc]

%----------------------------------------------------------------------------------------

\end{document}

%% file: Appendices/Abstract.tex
%----------------------------------------------------------------------------------------
%	ABSTRACT PAGE
%----------------------------------------------------------------------------------------

\begin{abstract}
\addchaptertocentry{\abstractname}
Mechanism design is a central research branch in microeconomics.
An effective mechanism can significantly improve performance and efficiency of social decisions under desired objectives, such as to maximize social welfare or to maximize revenue for agents.

However, mechanism design is challenging for many common models including the public project problem model which we study in this thesis. A typical public project problem is a group of agents crowdfunding a public project (e.g., building a bridge). The mechanism will decide the payment and allocation for each agent (e.g., how much the agent pays, and whether the agent can use it) according to their valuations. The mechanism can be applied to various economic scenarios, including those related to cyber security. There are different constraints and optimized objectives for different public project scenarios (sub-problems), making it unrealistic to design a universal mechanism that fits all scenarios, and designing mechanisms for different settings manually is a taxing job. Therefore, we explore automated mechanism design (AMD) (\cite{sandholm2003automated}) of public project problems under different constraints.

In this thesis, we focus on the public project problem, which includes many sub-problems (excludable/non-excludable, divisible/indivisible, binary/non-binary). We study the classical public project model and extend this model to other related areas such as the zero-day exploit markets. For different sub-problems of the public project problem, we adopt different novel machine learning techniques to design optimal or near-optimal mechanisms via automated mechanism design.
\end{abstract}
\newpage
We evaluate our mechanisms by theoretical analysis or experimentally comparing our mechanisms against existing mechanisms. The experiments and theoretical results show that our mechanisms are better than state-of-the-art automated or manual mechanisms.

%% file: Appendices/Abbreviations.tex
\chapter{List of Abbreviations}
\chaptermark{List of Abbreviations}
\begin{abbrv}
\item[AMD]
Automated mechanism design
\item[GAN]
Generative adversarial network
\item[MLP]
Multi-layer perceptions
\item[AMA]
Affine maximizer auctions (\cite{Likhodedov2005:Approximating})
\item[LP]
Linear programming
\item[IR]
Individual rationality
\item[SP]
Strategy-proofness
\item[SF]
Straight-forwardness (Assumption \ref{SF})
\item[DP]
Dynamic programming
\item[PORF]
Price-oriented rationing-free (\cite{yokoo2002price})
\item[DSIC]
Dominant-strategy incentive compatible
\item[CEC]
Conservative equal cost mechanism
\item[CDF]
Cumulative distribution function
\item[PDF]
Probability density function
\item[SCS]
Serial cost sharing mechanism
\item[NN]
Neural network
\item[TGA]
Truthful genetic algorithm (In Chapter \ref{Public Project with Minimum Expected Release Delay})
\item[ATGA]
Approximately truthful genetic algorithm (In Chapter \ref{Public Project with Minimum Expected Release Delay})
\item[VCG]
Vickrey-Clarke-Groves mechanism
\item[UB]
Upper bound

\end{abbrv}

%% file: Appendices/Declaration.tex
%----------------------------------------------------------------------------------------
%	DECLARATION PAGE
%----------------------------------------------------------------------------------------

\begin{declaration}
\addchaptertocentry{\authorshipname} % Add the declaration to the table of contents
%source: https://www.adelaide.edu.au/graduatecentre/current-students/your-thesis-examination/preparation#thesis-declaration --> https://www.adelaide.edu.au/graduatecentre/system/files/media/documents/2020-01/examples-of-thesis-declarations-for-submissions.pdf --> the statement below is the one "For a thesis that contains publications"

I certify that this work contains no material which has been accepted for the award of any other degree or diploma in my name, in any university or other tertiary institution and, to the best of my knowledge and belief, contains no material previously published or written by another person, except where due reference has been made in the text. In addition, I certify that no part of this work will, in the future, be used in a submission in my name, for any other degree or diploma in any university or other tertiary institution without the prior approval of the University of Adelaide and where applicable, any partner institution responsible for the joint-award of this degree.

I acknowledge that copyright of published works contained within this thesis resides with the copyright holder(s) of those works.

I also give permission for the digital version of my thesis to be made available on the web, via the University’s digital research repository, the Library Search and also through web search engines, unless permission has been granted by the University to restrict access for a period of time.

\vspace{40mm}
\begin{figure}[h!]
\begin{flushright}
  \includegraphics[width=0.4\textwidth]{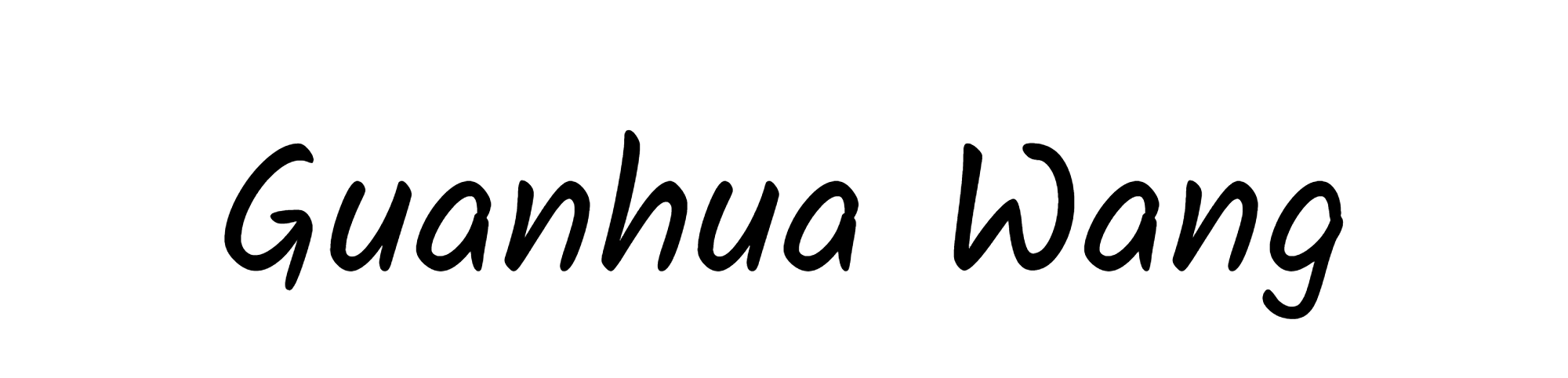}
\end{flushright}
\end{figure}
\begin{flushright}
Guanhua Wang
\end{flushright}
\begin{flushright}December 2021\end{flushright}
 
\end{declaration}

\cleardoublepage

%% file: Appendices/Acknowledge.tex
%----------------------------------------------------------------------------------------
%	ACKNOWLEDGEMENTS
%----------------------------------------------------------------------------------------

\begin{acknowledgements}
\addchaptertocentry{\acknowledgementname} 
Firstly, I want to thank my supervisor Dr. Mingyu Guo. I am honorable working with him in these three years. I very much appreciate his help. When I stock in no research inspiration, he can always figure it out and find a suitable way for me to solve the problems. Mingyu is a humble scholar. His research foresight and research inspiration is lifelong treasure for me.\\

I would thank Dr. Wei Zhang, and professor Muhammad Ali Babar. They provide some suggestions and comments that helped me to significantly improve the content of this thesis and papers.\\

And from my Ph.D. role, I also thank my collaborators, lab meta, and my friends for the help with the papers. I would like to mention Ba-Dung Le, Runqi Guo, Yuko Sakurai, Wuli Zuo, Xie Yue, Zhigang Lu for many useful discussions. My study area is not a hot area in this university. Thanks, you give me many useful discussions.  Good luck to you in the future. And I will thank Wuli Zuo for checking and correcting my grammar mistakes.\\

I also want to thank my parents and my grandparents for financially supporting me during my Ph.D. study. And because of Covid-19, it is a hard time to connect you only on the internet. Hopefully, we can reunion soon after Covid-19.\\

Finally, I will thank the university and its staff. They provide a good environment for my research.

\end{acknowledgements}

%% file: Appendices/Publications.tex
\chapter{Publications}
This thesis is based on the following research papers that have been published in peer-reviewed conferences or journals:

\begin{itemize}

\item \textbf{Guanhua Wang}, Runqi Guo, Yuko Sakurai, Muhammad Ali Babar and Mingyu Guo. Mechanism Design for Public Projects via Neural Networks. In Proceedings of the 20th International Conference on Autonomous Agents and Multiagent Systems (AAMAS 2021, Conference). (CORE Rank: A*)
\item \textbf{Guanhua Wang}, Mingyu Guo. Public Project with Minimum Expected Release Delay. In the 18th Pacific Rim International Conference on Artificial Intelligence (PRICAI, Conference) 2021. (CORE Rank: B)
\item \textbf{Guanhua Wang}, Wuli Zuo, Mingyu Guo. Redistribution in Public Project Problems via Neural Networks. In the 20th IEEE/WIC/ACM International Conference on Web Intelligence and Intelligent Agent Technology (WI-IAT, Conference), Melbourne, Australia, 2021. (CORE Rank: B)
\item Mingyu Guo, \textbf{Guanhua Wang}, Hideaki Hata, M. Ali Babar. Revenue-Maximizing Markets for Zero-Day Exploits. In Autonomous Agents and Multi-Agent Systems (AAMAS, Journal) 35, 36, 2021. (CORE Rank: A)
\end{itemize}

\newpage
Other research paper that has been published in peer-reviewed conferences but not used in this thesis:\\

\begin{itemize}
\item Ba-Dung Le; \textbf{Guanhua Wang}; Mehwish Nasim; Muhammad Ali Babar, "Gathering Cyber Threat Intelligence from Twitter Using Novelty Classification," 2019 International Conference on Cyberworlds (CW, Conference), 2019, pp. 316-323, (CORE Rank: B)

\end{itemize}

%% file: Chapters/Chapter1.tex
% Chapter Template

\chapter{Introduction} % Main chapter title

\label{introduction} % Change X to a consecutive number; for referencing this chapter elsewhere, use \ref{ChapterX}
Mechanism design is a fundamental and important research area in economics. Since 2000, the Nobel Memorial Prize in Economic Sciences has been twice awarded for mechanism design. Nobel Prize 2007 in Economics was awarded to Leonid Hurwicz, Eric S. Maskin, and Roger B. Myerson for laying the foundations of mechanism design theory (\cite{drobietz2021cardiovascular}). Then the Prize 2020 was awarded to Paul Milgrom and Robert Wilson for improving auction theory (\cite{janssen2020reflections}). Mechanism design takes an objective-first approach to design economic mechanisms or incentives, toward desired objectives, in strategic settings, where players are individual rational (\cite{lee2016incentive}).\\

%edition_begin
Mechanism design theory has applications in a number of computer science sub areas. Mechanism design studies how
to make social decisions when we need to take individual preferences into consideration. Many problems in the computer science domain are related to mechanism design. For example, the popularity of e-commerce leads to a long list of new Internet-based markets, such as Advertisement auctions and blockchains. Besides new markets, many sub areas of computer science involve mechanism design problems, such as multi-agent coordination in multiagent systems, distributed computing, and almost all domains involving competition (i.e., cyber security) and collaboration (i.e., federated machine learning and multiagent reinforcement learning).
Furthermore, latest development in computational techniques also leads to new development of economic results. The main topic of this thesis is along this line (i.e., using machine learning to develop new mechanisms). In this thesis, there is also discussion on applying mechanism design theory to cyber security.

This thesis focuses on the public project problem.
To illustrate the public project model, we present the following example.
Suppose there are $n$ households (often referred to as $n$ agents) living in a community. They plan to build a public project (a project shared by all households, e.g., a swimming pool).
We assume that the project costs $1$ dollar (without loss of generality).
These agents plan to crowdfund this swimming pool. Only those
who pay will have access.  Every agent has a different valuation for the
swimming pool, which is private information. We assume that agent $i$'s valuation for the swimming pool
is $v_i$.
$v_i$ is the maximum amount agent $i$ is willing to pay to contribute to the pool/access the pool.
A mechanism (social decision rule) maps the agents' preferences (i.e., the $v_i$) to the social decision (i.e., who can access the pool and how much each agent pays).
One example design goal is to design a crowdfunding
rule that maximizes the number of agents that can
access the pool (i.e., in expectation with respect to the prior distributions of the $v_i$).\\

One example mechanism is simply to ask every agent to pay an equal share ($\frac{1}{n}$). If any agent disagrees, then we simply cancel the project. This way, no agents can benefit by lying (disagree with the proposed payment when the agent can afford it, or agree with the proposed payment when the agent cannot afford it) and the total payment collected is exactly $1$, which meets the project cost (or $0$ when the project is cancelled). Of course, this rule is not ideal because as long as there is one agent whose value is lower than $\frac{1}{n}$, then the project is not built. Given the goal of maximizing the number of agents who can consume the project, we want a rule that charges different agents different amounts (i.e., only agents who have high values for the pool need to pay more -- we need to identify agents who can afford to pay more and force them to pay more).\\

The above task isn't actually easy, because we are designing a function (the crowdfunding social decision rule) that maps an input vector of dimension $n$ (the $v_i$) to an output vector that describes every agent's allocation and payment (output dimension is $2n$).
Furthermore, not all functions are feasible. It is a difficult task in itself to describe what feasible functions look like. We may add a constraint that is, for example, ``no manipulation'' (i.e., if an agent can afford to pay a high payment, then we don't want this agent to pretend that her value is low by reporting a low value -- recall that the agents' preferences are private information). It is difficult to mathematically derive the set of all functions that prevent manipulation and this is just for one mechanism design constraint. Later on in this thesis, we will formalize commonly used constraints in mechanism design.\\

In what follows, we present a timeline of mechanism design research.\\

Mechanism design theory originated from Hurwicz's pioneering
work in 1960 (\cite{hurwicz1960optimality}). For the general problem, it discusses whether and how to design an economic mechanism (such as an auction) for any given economic or social goal under decision-making conditions (such as  incentive compatibility). Major mechanisms are under the incentive compatibility (IC) constraint to align the personal interests of economic participants with the mechanism design goals.\\

In 1967, Harsanyi (\cite{harsanyi1967games}) defined the mechanism as a game of information involving the agents. The agents receive secret "messages" containing information relevant to payoffs. For example, a message may contain information about their preferences or the quality of a good for sale. This information is called the agent's "type". Agents then report a type to buy the goods. The reported types can be strategic false valuations or the true valuations. After reporting, the agents
are allocated according to the mechanism and are asked to pay accordingly
(Figure \ref{fig:mechanism1}). Mechanisms align the personal interests of agents with the mechanism design goal by encouraging agents to reveal their true valuations.\\

\begin{figure}[H]
\centering
    \includegraphics[width=0.8\textwidth]{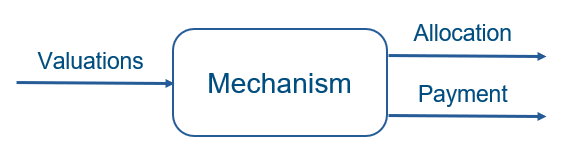}
        \caption[A typical mechanism payoff structure]{A typical mechanism payoff structure.}
\label{fig:mechanism1}
\end{figure}

%edition_end
In 1981, Myerson et al. (\cite{myerson1981optimal}) obtained the optimal auction mechanism for single-item auctions. Since then, intense research has been carried out to solve optimal auctions in different settings. However, after four decades, for multi-agent or multi-item revenue-maximizing auctions, the optimal auctions remain unknown (\cite{curry2022differentiable}). When it comes to the specific context of public project problems, designing a multi-agent mechanism meets more difficulties due to lack of characterisations of feasible mechanisms, as well as the variety of design objectives and the problem's high-dimension (when the number of agents is large) (\cite{maghsoudlou2017multi,carroll2019robustness}). For example, when we crowdfund a public project (i.e., a piece of software), the project could be non-excludable (the software is open source, therefore can be used by everyone, including the non-paying agents) or excludable (the software can only be used by agents who pay), indivisible (the software is a one-time sale for all agents) or divisible (the software is a subscription-based and an agent may enjoy ``a fraction of the software'' -- i.e.,, subscribe half of the time).\\

It is difficult to manually design public project mechanisms considering the aforementioned difficulties. We try to automatically design optimal or near-optimal mechanisms for the public project problem via machine learning methods. We propose several neural network training techniques for designing public project mechanisms, which also have the potential to be generalized to other mechanism design models. Besides neural networks, we also develop automated mechanism design techniques based on evolutionary computation. For instance, for an example application setting on zero-day exploit markets, we focus on the Affine Maximizer Auctions (AMA) (\cite{lavi2003:towards,Likhodedov2005:Approximating,curry2022differentiable}) and use the Fourier series to characterize the AMA mechanisms. We then adjust the AMA mechanism parameters via evolutionary computation.\\

%In the meanwhile in Chapter \ref{Revenue-Maximizing Markets for Zero-Day Exploits},

%briefly compare them with the classical linear programming (LP) method.

%----------------------------------------------------------------------------------------
%	SECTION Background
%----------------------------------------------------------------------------------------

\newpage
\section{Public Project Mechanism Design Research Challenges and Motivation} \label{introduction_motivation}

% In the real world, there are many problems that are related to public project problems (\cite{myerson1981optimal,abdulkadirouglu2003school,krishna1998efficient}). For example, a group of people crowdfund a road, free/paid PC games, or construction of public infrastructure, etc. They could be regarded as public project problems. Agents want to get this kind of project with low payment. However, people have not found a universal mechanism that could get optimal or near-optimal results in all scenarios of public project problems due to a number of reasons.\\

\subsection{Unknown Feasible/Optimal Mechanism Characterizations}
The first difficulty lies in mechanism design itself. Mechanisms collect the valuation from each agent, and design payment and allocation schemes (Figure \ref{fig:mechanism1}).
It is obvious that, in the design of the mechanism, the input side is the valuations from the agents, which the mechanism designer can never assign or design. To achieve the desired objectives, the mechanism designer needs to get the true valuations from the agents.
The only way we encourage the agents to tell the truth is by controlling the two output functions (allocation and payment functions). The designed allocation and payment functions must ensure that each agent can get the highest utility whenever she reports the true valuation.\\

Since Myerson solved the single-item optimal auction design problem in his seminal work (\cite{myerson1981optimal}), intense research has been taken to solve problems under different settings. However, the problem is not completely resolved (\cite{sandholm2003automated}). For public project problems, for some variants of the model, we do not know the exact characterizations of the truthful mechanisms. Even when we know the exact characterization of the truthful mechanism, the characterization is difficult to work with. For example, the characterization may contain an exponential number of parameters, and there are no known techniques for efficient optimisation that works with the characterization. This is different from Myerson's result for single-item auctions, where the proposed optimisation technique allows us to focus on very restricted forms of mechanisms.\\

\subsection{Neural Network Mechanism Design Challenges for Public Project Problems}
One of the characteristics of the public project problem is that the mechanism needs to decide whether or not to build the project (\cite{sen2007theory}). When we apply neural networks to design public project mechanisms, the two decisions (build vs not build) often lead to entirely different gradients for the loss function. In stochastic gradient descent, we can only use a small batch of type profiles, which means that often the number of "build" cases plays a more important role than the quality of the mechanism when it comes to calculating the gradient. This causes poor training results and slow training speed.
\\

\subsection{High Variety of Public Project Problems}
In many mechanism design problems, the typical objectives are to maximize social welfare or to maximize revenue (\cite{zou2009double,manelli2007multidimensional,amanatidis2019budget,amanatidis2017budget}). In this thesis, we deal with public project problems with a few more objectives besides maximizing for welfare and revenue. Our objectives studied also include minimizing the total and maximum release delay (for public projects where the mechanism designer can set a release time) and maximizing the number of consumers for the public project. Existing manual mechanisms designed by the economists perform poorly toward our proposed objectives.\\

The performance of mechanisms is often impacted by the number of agents
(\cite{zhao2014multi}). For example, if there are a large number of agents
joining the public project problem, then the project has a high probability to
be built. We could even entirely ignore the ``not build'' possibility when
designing mechanisms. This observation allows us to design asymptotically optimal
mechanisms for a few variants of the model.
The challenging design cases are sometimes only about small numbers of agents.
This plays well with our approach of using
neural networks and evolutionary computation to automatically design
mechanisms.\\

%For example, in some software, there are different release date for the paid user and free user.

\newpage
\section{Research Summary}

In this thesis, we focus on automatically designing mechanisms for the public project problem (\cite{mezzetti2004mechanism}). In what follows we describe our mechanism design constraints, the design objectives including the details of various public project models that are covered in this thesis, our mechanism optimisation techniques, and finally our contribution outlines.\\

\subsection{Mechanism Constraints: Individual Rationality (IR) and Strategy-Proofness (SP)}
The mechanism designer expects the agents to report their valuations truthfully so that the overall system objective can be achieved. However, an individual agent wants to maximize her own utility (i.e., receive better allocation and pay less). Therefore, an agent may have the incentive to lie about her true valuation if doing so increases her utility (\cite{matsushima2007mechanism}). An agent also may refuse to join a mechanism if she deems that joining the mechanism leads to a worse utility compared to her utility when staying out. When designing mechanisms, we enforce two constraints: Strategy-Proofness (SP) (\cite{ma1994strategy}) and Individual Rationality (IR).

\begin{itemize}
\item {\em Strategy-proofness (SP)}: For any agent $i$, her utility is maximized by reporting her valuation truthfully.\\
\item {\em Individual rationality (IR)}: For any agent $i$,  her utility is nonnegative
    when she reports truthfully.
\end{itemize}
%($v_i(t)$)

As mentioned earlier, for many mechanism design settings, coming up with a characterization of all strategy-proof and individual rational mechanisms is often a difficult task, which then makes mechanism design under these two constraints difficult.\\

\subsection{Public Project Problem}

A public project problem is a social decision model where a group of agents have to decide whether or not to build a public project (e.g., a swimming pool, public library, open source software, etc., \cite{kaiser2007athenian}). \\

Public project problems have many variants. In this thesis we consider all the following:

\begin{enumerate}
\item The project may be non-excludable or excludable (\cite{samuelson1954pure,ott2006excludable})
\item The project may be indivisible or divisible (\cite{lipton2004approximately})
\item We assume that the project is binary and non-rivalrous (\cite{samuelson1954pure,wadhwa2020failure})
\end{enumerate}
$$$$
\textbf{Terminology explanations:}
\begin{itemize}
	\item Non-excludable: it is impossible to exclude any individuals from consuming the good.
	(\cite{samuelson1954pure,Ohseto2000:Characterizations,ott2006excludable})

	\item Excludable: the mechanism designer can specify which agents can consume the project, for example, it could be that only those who pay for the project have access to its benefits. (\cite{ott2006excludable})

	\item Non-rivalrous: when one agent consumes the project, others are not prevented from using it.
	(\cite{samuelson1954pure,wadhwa2020failure})

	\item Binary: a public project is built or not built.
    A non-binary public project refers to a project with different provisional levels. For example, the agents may face building a better swimming pool, a poor-quality swimming pool, or not build at all.
        (\cite{lipton2004approximately})

	\item Indivisible: a public project can not be divided. For example, the allocation for an agent is either 1 (this agent uses the entire project) or 0 (this agent can not use the project at all). To be more specific, an agent only has two options, consumes the entire project or does not consume the project at all.
	(\cite{lipton2004approximately})

	\item Divisible: a public project can be divided. For example, one agent may consume or use part of the project.  The allocation to an agent is in the range from [0,1].\\
\end{itemize}

%edition_begin
Table \ref{table Sub-problems} highlights the differences in terms of model settings among the different chapters. In each chapter, we study a different public project model variant, and we use different machine learning methods to solve these sub-problems according to their unique mathematical structures.

\begin{table}[H]
\begin{tabular}{|c|c|c|c|c|}
\hline
          & Binary                    & Non-rivalrous             & Divisible\textbackslash{}Indivisible & Excludable\textbackslash{}Non-excludable \\ \hline
Chapter 3 & \checkmark & \checkmark & Indivisible                          & Study Both                               \\ \hline
Chapter 4 & \checkmark & \checkmark & Study Both                           & Excludable                           \\ \hline
Chapter 5 & \checkmark & \checkmark & Indivisible                                    & Non-excludable                                        \\ \hline
Chapter 6 [*] & / & / & Divisible                                    & /                                        \\ \hline
\end{tabular}
$$$$

\begin{tabular}{|c|c|c|c|}
\hline
          & SP                                                                                            & IR &  \multicolumn{1}{c|}{Objectives (goals)}                                                                                                             \\ \hline
Chapter 3 & \checkmark                                                                                             & \checkmark  & \begin{tabular}[c]{@{}c@{}} \quad 1. Maximize Expected Number of Consumers\\ \quad 2. Maximize Expected Agents’ Welfare \quad \quad \quad  \end{tabular} \\ \hline
Chapter 4 & \begin{tabular}[c]{@{}c@{}}Study Both SP\\ and Almost SP\end{tabular} & \checkmark  & \begin{tabular}[c]{@{}c@{}}  1. Minimize Expected Total Delay \quad \quad \quad \quad \\ 
2. Minimize Expected Max Delay \quad \quad \quad \quad   \end{tabular} \\ \hline
Chapter 5 & \checkmark                                                                                             &\begin{tabular}[c]{@{}c@{}} IR in\\ Expectation \end{tabular}     & \begin{tabular}[c]{@{}c@{}}  1. Worst-Case Efficiency Ratio \quad \quad \quad \quad \quad \quad  \\ 2. Expected Welfare
\quad \quad \quad \quad \quad \quad \quad \quad \quad \quad 

\end{tabular}                 \\ \hline
Chapter 6 & \checkmark                                                                                            & \checkmark  &  1. Revenue-Maximizing  \quad \quad \quad \quad  \quad \quad \quad \quad  \quad                                                                                \\ \hline
\end{tabular}
\caption{Public Project Model Variants for Each Chapter}
[*] In Chapter 6, the general problem is not a public project problem. However, from the perspective of the defenders, they face a model that {\bf highly resembles} a divisible public project model.
\label{table Sub-problems}
\end{table}

\newpage
\section{Contributions and Outline} \label{introduction_contributions}

In this thesis, we discuss various different methods for automatically designing mechanisms. Our aim is to improve the performance of public project mechanisms, and provide optimal or near-optimal public project mechanisms for real-world applications. We adopt different neural network and evolutionary computation frameworks to automatically design mechanisms for public project problems, while considering different objectives, constraints and model variety (Table \ref{table Sub-problems} and Table \ref{Machine Learning Methods}).\\

A commonly used approach in this thesis is that for a small number of agents, we use automated mechanism design to derive mechanisms for our objectives (e.g., Table \ref{table Sub-problems}). For a large number of agents, We try to theoretically analyse the automated designed mechanisms and calculate the theoretical performance upper bound. We compare our mechanisms to the theoretical upper bound and/or existing state-of-the-art mechanisms to demonstrate the effectiveness of our proposed mechanisms.\\

In Table~\ref{Machine Learning Methods}, we highlight the mechanism families considered (i.e., we would focus on these families of mechanisms, and then optimize within them) and also our machine learning techniques used.

\begin{table}[H]
\begin{tabular}{|c|c|c|}
\hline
          & Mechanism Families                       & Machine Learning Approach                                                             \\ \hline
Chapter 3 & Largest Unanimous Mechanisms & Neural Networks\\ \hline
Chapter 4 & Sequential Unanimous Mechanisms & Evolutionary Computation                                                              \\ \hline
Chapter 5 & VGC Redistribution Mechanisms    & Neural Networks                                                                             \\ \hline
Chapter 6 & AMA mechanisms                            & \begin{tabular}[c]{@{}c@{}}Evolutionary Computation\\+ Neural Network\end{tabular} \\ \hline
\end{tabular}

\caption{Mechanism Families and Machine Learning Techniques used in Each Chapter}
\label{Machine Learning Methods}
\end{table}

Finally, we conclude this chapter with a detailed summary of our contributions
in future chapters.

\begin{itemize}

\item In Chapter \ref{Mechanism Design for Public Projects via Neural Networks}, we study both the \textbf{non-excludable} and the \textbf{excludable} versions of the \emph{binary} \emph{non-rivalrous} \emph{indivisible} public project problem.
The existing methods (e.g., serial cost-sharing mechanism, \cite{Moulin1994:Serial}) are not optimal in general.
We identify a sufficient condition on the prior distribution for the conservative equal costs mechanism to be the optimal strategy-proof and individually rational mechanism. We prove that, in some scenarios, the serial cost-sharing mechanism is optimal. For other scenarios, we find that the serial cost-sharing mechanism is far from optimality. For these scenarios, we design better-performing mechanisms via neural networks.

For non-excludable public project problems, we improve existing mechanisms by using dynamic programming (DP). For excludable problems, we design a better-performing price-oriented rationing-free (PORF) mechanisms via neural networks. We propose three neural network training techniques. The experiments show that our mechanisms are better than previous results and are closer to the theoretical upper bound.

Contributions:

\begin{enumerate}

    \item We are the first to use the price-oriented rationing-free (PORF) mechanisms to
    assist the designing of iterative mechanisms via neural networks. PORF
        mechanisms move the mechanism's complex (\emph{e.g.}, iterative)
        decision-making process off the neural network to a separate simulation
        process. \\

    \item We feed prior distribution into the loss function. By feeding the prior distribution's {\em analytical form} into the cost function, we can provide high-quality gradients. \\

\item We learn toward existing manual mechanisms as initialization. We use supervision to state-of-the-art manual mechanisms as a systematic way for initialization. It is also effective in fixing constraint violations for heuristic-based mechanisms.\\

\end{enumerate}

\item In Chapter \ref{Public Project with Minimum Expected Release Delay}, we study a divisible public project model, where an agent can be allocated only part of the project.
    Specifically, we study settings where the mechanism can set different project release times for different agents. For an agent, the higher she pays, the earlier she can use the project. There are two objectives for this problem: to minimize the expected maximum release delay and to minimize the expected total release delay. The existing mechanisms (without delayed release) can be trivially applied to our model by interpreting them as mechanisms whose release times are either 0s or 1s. Nevertheless, as expected, existing mechanisms do not perform well as they were not designed for our objectives.

Under our automated mechanism design approach, we first find regularities by analysing cases involving a small number of agents and then generalise the rule to a larger number of agents.
For small numbers of agents, we propose the sequential unanimous mechanisms by extending the existing largest unanimous mechanisms. We then use evolutionary computation to optimize within the sequential unanimous mechanism family.
The experiments show that our mechanisms are better than the existing mechanisms.
Then we summarize the patterns of these sequential unanimous mechanisms and apply them to a larger number of scale. We end up with the single deadline mechanisms. We theoretically prove that the single deadline mechanisms are asymptotically optimal, regardless of the prior distributions.

Contributions:

\begin{enumerate}

\item For a small number of agents, we propose the sequential unanimous mechanism family and apply automated mechanism design via evolutionary computation.\\

\item For a large number of agents, we propose a novel single deadline mechanism, which is asymptotically optimal.\\

\end{enumerate}

\item  In Chapter \ref{Redistribution in Public Project Problems via Neural Networks}, we focus on the VCG and the VCG redistribution (which is based on VCG) mechanisms for the public project problem.
We design mechanisms via neural networks with two welfare-maximizing objectives: optimal in the worst case and optimal in expectation. We design generative adversarial networks and multi-layer perceptions (GAN + MLP) to find the optimal worst-case performance of VCG redistribution mechanisms for public project problems. GAN is used to generate type profiles with poor worst-case performances, which are used for training.
The experiments show that our mechanism is better than existing approaches.
We use multi-layer perceptions (MLP) to find the optimal-in-expectation VCG redistribution mechanism. Our innovation is a new way to construct the cost function for training, by including the prior distribution of the agent valuations as weights for a training batch. The experiments show that our mechanisms' performances are very close to the theoretical upper bound.

Contributions:

\begin{enumerate}

\item We are the first to utilise a GAN network to generate worst-case type profiles for training toward worst-case optimal mechanisms.\\

\item We feed prior distribution into loss function to provide quality gradients for the optimal-in-expectation objective.\\

\item We discuss dimension reduction for handling large agent count. By reducing the input dimension, the neural network converges faster and still retains good performance.\\

\end{enumerate}

\item In Chapter \ref{Revenue-Maximizing Markets for Zero-Day Exploits}, we
    focus on a specific scenario called zero-day exploits market.
    In such a market, one zero-day exploit (i.e., an
    exploit that allows cyber attackers to hack into iOS systems) is sold to
    multiple offender and defenders. In this model, for the defensive side, as
    long as any defender gains access to the exploit, the exploit is assumed to
    be immediately fixed, which benefits all defenders. The defensive side of
    the our model is very similar to a non-excludable and divisible public project problem. Nevertheless, the overall model is not a public project problem.

In order to maximize revenue in zero-day exploit markets, we adopt
computational feasible technical automated mechanism design approaches
(\cite{Guo2010:Computationally,Guo2017:Optimizing}). One commonly used
mechanism family for revenue maximization is the Affine Maximizer Auctions (AMA). For this
particular model, we observe that an AMA mechanism can be characterized by a single variable
function that can be visualized as a curve.  We propose two numerical solution
techniques, one is based on neural networks and the other one is based on
evolutionary computation.  We use neural networks to automatically design the
optimal curve for the Affine Maximizer Auctions (AMA) mechanisms. We also use
evolutionary computation (based on Fourier series) to optimize for a good AMA
curve.
The experiments show that our mechanisms based on neural networks
and evolutionary computation are near-optimal and get better results compared
to the state-of-the-art method, which was based on iterative linear programming
(LP).

Contributions:

\begin{enumerate}

\item We are the first to use a series of novel techniques to train the AMA mechanisms, including Fourier-series-based evolutionary computation and neural networks. \\

\end{enumerate}

\item In Chapter \ref{Conclusion}, we summarize this thesis and discuss future
    directions.\\

\end{itemize}

%% file: Chapters/Chapter2.tex
% Chapter Template

\chapter{Literature Review and Definitions} % Main chapter title

\label{Background and Related Work} % Change X to a consecutive number; for referencing this chapter elsewhere, use \ref{ChapterX}

\section{Mechanism Design}

Mechanism design is the art of designing social decision rules so that the agents are motivated to report their true valuations and a suitable (according to a given objective) outcome is chosen. The objectives are varied. It could be social welfare, fairness, or revenue maximization (\cite{cole2014sample,morgenstern2015pseudo}).

\begin{figure}[H]
\centering
    \includegraphics[width=0.7\textwidth]{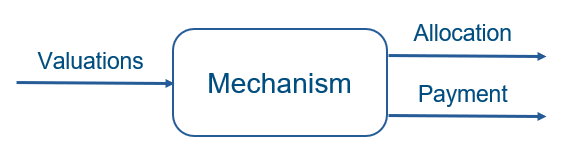}
        \caption[A typical mechanism structure]{A typical mechanism structure.}
        \label{fig:m}
\end{figure}

Generally speaking, mechanisms specify payments and allocations (as outlined in Figure \ref{fig:m}). For a certain agent, the higher she pays, the better allocation she gets. For different agents, they may gain different allocations even if they pay the same.\\

For example, for the public project problem, a mechanism would specify whether the project is built or not, who can consume the project (if the project is excludable), and how much each agent needs to pay.
% Mechanism usually has some common properties (\cite{roughgarden2010algorithmic,morgenstern2016learning}).\\

\subsection{Mechanism Properties}

Self-interested agents may lie about their valuations if doing so increases their own utilities. Therefore, it is necessary to design mechanisms with desired properties. In this thesis, we focus on designing mechanisms that are strategy-proof and individually rational (\cite{yokoo2004effect,dash2004trust,azevedo2019strategy}).\\

\subsubsection{Individual Rationality}

Individual rationality means that it is unacceptable for an agent to receive less utility than that she would have received if not joining in the mechanism. Under individually rational mechanisms, each individual weakly prefers to join in the mechanism rather than not participating (\cite{anand1995foundations,dash2004trust}).\\

\subsubsection{Strategy-proofness}

Strategy-proofness (SP) is also called truthfulness or
dominant-strategy incentive-compatibility (DSIC). It means that no agents can benefit by
misreporting their preferences
(\cite{mookherjee1992dominant,azevedo2019strategy}). Roughgarden
(\cite{roughgarden2010algorithmic}) summarized the commonly used definitions
of incentive compatibility. The weaker notion is Bayesian-Nash incentive
compatibility and the stronger notion is dominant-strategy incentive
compatibility (also often called strategy-proofness). Typical examples of SP
mechanisms are the VCG mechanisms.\\

% Roughgarden (\cite{roughgarden2010algorithmic}) showed that a mechanism is incentive-compatible (IC) if every agent can achieve the best outcome to themselves just by acting their true preferences. Roughgarden et al. identify the incentive-compatible into two degrees:
% \begin{enumerate}
% \item  The weaker degree is Bayesian-Nash incentive-compatibility (BNIC): All the others act truthfully, then it is best or at least not worse for you to be truthful.
% \item  The stronger degree is Dominant-strategy incentive-compatibility (DSIC, also called strategy-proofness): regardless of what the others tell the truth or not, it is best or at least not worse for you to be truthful. \\
% \end{enumerate}

% Telling the truth could make every member better off than misreporting the valuations.

% In this thesis, we study strategy-proof mechanisms. For each agent $i$, the utility of telling the truth will be greater or equal to the utility of telling the lie.  \\
%$$u(v_i) \geq u(v_i'), \forall v_i, v_i' \in [0,1]$$
%\\

In 1981, Myerson (\cite{myerson1981optimal}) solved the optimal auction design problem when there is a single item for sale.  However, after four decades, for multi-item revenue-maximizing auctions, the optimal solutions remain unknown (\cite{curry2022differentiable}). Researchers have developed some elegant partial characterization results (\cite{manelli2007multidimensional,pavlov2011optimal,yao2017dominant}) or have developed impressive algorithmic optimal or near-optimal mechanisms (\cite{cai2012algorithmic,hart2017approximate}) for some specific settings. For the public project problem, there are limited characterization results and also limited existing manual mechanisms.
\\

\subsection{Manual Mechanism Design}

Mechanism design has traditionally been based on manual effort and human-expert experience (\cite{sandholm2003automated}). The designer manually designs a certain rule set and then proves that the proposed mechanism is optimal or near-optimal. Often, the designer studies the mechanism design problem from a mathematical structure perspective (\cite{krishna2009auction}). Indeed, many impressive and significant mechanisms were designed this way. Armstrong obtained the revenue optimal mechanisms for selling two items to one buyer (\cite{armstrong1996multiproduct}). Pavlov et al. (\cite{pavlov2011optimal}) derived optimal results for two items with symmetric uniform distributions.  In this thesis, we also conduct some level of manual mechanism design for certain public projects, and then theoretically prove by mathematical analysis that the derived mechanism is optimal in certain situations or under certain technical assumptions.

\subsection{Automated Mechanism Design}

In 2002, Conitzer et al. (\cite{Conitzer2002:Complexity}) first proposed the idea of automated mechanism design (AMD).
The designer first needs to design a general framework and a given objective for the problem. Then the designer leaves the optimisation process to the computer, and the computer will automatically calculate the suitable parameters and decide on the details of the mechanisms.\\

    Since the proposal of AMD, many researchers used AMD for mechanism design.
    In 2006, Jurca and Faltings (\cite{jurca2006minimum}) applied Automated
    Mechanism Design to compute the minimum payments for a reputation system.
    In 2007, Constantin and Parkes (\cite{constantin2007revenue}) studied
    dynamic single-item auctions using automated mechanism design for
    interdependent value agents. The authors used mixed-integer programming to
    handle a small number of agents. In 2010, Bhattacharya et al.
    (\cite{bhattacharya2010budget}) used AMD to study revenue-maximizing
    multi-item auctions for budget-constrained agents. Likhodedov
    et al.  (\cite{Likhodedov2004:Methods,Likhodedov2005:Approximating})
    applied automated mechanism design to maximize revenue for combinatorial
    auctions. The authors focused on the family of Affine Maximizer Auctions
    (AMA), and AMD was used to automatically adjust the mechanism parameters.

\subsection{Classic Mechanisms}

In this section, we will introduce a few well-known/state-of-the-art mechanisms that are relevant to this thesis.

\subsubsection{Cost Sharing Mechanisms}

Cost sharing mechanisms are mechanisms that can be applied to binary indivisible public project models. Cost sharing mechanisms are \textbf{strategy-proof} and \textbf{individual rational} (\cite{moulin2005price}).
% In this thesis, the marginal cost of the public project is constant. We use the cost sharing mechanism as the baseline to evaluate our mechanisms. The cost sharing mechanism is a kind of the largest unanimous mechanism.
We summarize a few existing cost sharing mechanisms as follows:

\begin{enumerate}
    \item  \paragraph{\textit {Largest Unanimous Mechanism} (\cite{Ohseto2000:Characterizations})}

For every nonempty coalition of agents $S = \{S_1,S_2,\ldots,S_k\}$, there is a constant cost share vector $C_S=(c_{S_1},c_{S_2},\ldots,c_{S_k})$ with $c_{S_i}\ge 0$ and $\sum_{1\le i\le k} c_{S_i}=1$. $c_{S_i}$ is agent $S_i$'s cost share under coalition $S$. If agent $i$ belongs to two coalitions $S$ and $T$ with $S\subsetneq T$, then $i$'s cost share under $S$ must be greater than or equal to her cost share under $T$. Agents in $S$ unanimously approve the cost share vector $C_S$ if and only if $v_{S_i}\ge c_{S_i}$ for all $i \in S$. The mechanism picks the largest coalition $S^*$ satisfying that $C_{S^*}$ is unanimously approved. If $S^*$ does not exist, then the decision is not to build. If $S^*$ exists, then it is always unique, in which case the decision is to build. Only agents in $S^*$ are consumers and they pay according to $C_{S^*}$.

\item  \paragraph{\textit {Serial Cost Sharing Mechanism (SCS)} (\cite{Moulin1994:Serial,Guo2018:Cost})}

The serial cost sharing mechanism is a specific largest unanimous mechanism. For every nonempty subset of agents $S$ with $|S|=k$, the cost share vector is $(\frac{1}{k},\frac{1}{k},\ldots,\frac{1}{k})$. The mechanism picks the largest coalition $S^*$ where the agents are willing to pay equal shares.

\begin{itemize}
\item If $S^*$ is empty, then the project will not be built

% , no agent can use it (release date $t_i$ = 0) or needs to pay ( $p_i = 0$).

\item If $S^*$ is not empty, then every agent in $S^*$ each pays an equal share and only those who pay can consume the project.\\
% the project immediately (release date $t_i$ = 1). The other agents use the project at time 1 (release date $t_i$ = 0) and each pays 0 ($p_i = 0$).
\end{itemize}

The serial cost sharing mechanism is elegant and we also prove that in many situations, they are in fact optimal/near-optimal.

% \begin{enumerate}
% \item It is simple and easy to implement.
% \item It is not affected by manipulations in which other agents secretly transfer or collude with each other.
% \item Each agent is paid at least its stand-alone cost - the cost he would have paid by himself without the existence of other agents.\\
% \end{enumerate}

% The above mentioned mechanisms are not optimal in general.
\end{enumerate}
% $$$$

\subsubsection{VCG and VCG Redistribution Mechanisms}

\begin{enumerate}
\item  \paragraph{\textit {VCG Mechanism}}
The Vickrey-Clarke-Groves (VCG) mechanism is efficient and strategy-proof but not budget-balanced (\cite{Vickrey1961:Counterspeculation,clarke1971multipart,groves1973incentives}).
\\

Under the VCG mechanism, each agent $i$ reports her private type $\theta_i$. The outcome that maximizes the agents' total valuations is chosen. Every agent is required to make a VCG payment $t(\theta_{-i})$, which is determined by the other agents' types. An agent’s VCG payment is usually described as the extent to which the agent's existence harms the other agents, in terms of the total valuation of other agents, which is called the externality in economics. The total VCG payment may be quite large, leading to decreased welfare for the agents. In particular, in the context of the public project problem, where the goal is often to maximize the social welfare (the agents' total utility considering payments), having large VCG payments are undesirable.\\

% /*
% \begin{tcolorbox}
%     \begin{center}
%         VCG Mechanisms
%     \end{center}
% \begin{itemize}
%     \item Allocation rule:
%       $$ x(b)  = argmax_\omega \sum_i b_i(\omega) ...(1)$$

%     \item Agent $i$'s payment equals:
%       $$ p_i(b)= max \sum_{j\neq i}{b_j(\omega)}-\sum_{j \neq i} b_j(\omega*)  ...(2)$$
%       where $b$ represents ther externality, $\omega \in \Omega$, and $\Omega$ is an finite outcome set. $\omega* = argmax_\omega \sum_i b_i(\omega) $ is the outcome chosen in (1).
%     \end{itemize}
% \end{tcolorbox}
% */

\item \paragraph{\textit {VCG Redistribution Mechanisms}}

VCG redistribution mechanisms are also efficient, strategy-proof and not budget-balanced (\cite{guo2008undominated,Guo2011:VCG}). They are proposed to address the social welfare loss due to large VCG payments.

Under a VCG redistribution mechanism, we first run VCG, and then redistribute the VCG payments back to the agents as much as possible. The amount that every agent receives (or pays additionally) is called the redistribution payment, and is characterized by a redistribution payment function $h$. For agent $i$, its redistribution payment $h(\theta_{-i})$ depends on other agents' types $\theta_{-i}$. In this way, the VCG redistribution mechanism remains to be strategy-proof (\cite{Naroditskiy2012:Redistribution}). \\

The VCG redistribution mechanism may bring the agents more benefit due to the redistribution payment, so the total social welfare could be increased compared to the VCG mechanism.
\end{enumerate}

% \subsection{One Application: In Zero-day Exploit Markets}

% A zero-day exploit or attack happens once a software/hardware vulnerability or flaw occurs. Software vulnerabilities that have not been disclosed to the public and are not known to software vendors. \\

% Information of new vulnerabilities gives cyber attackers free passes to attack their targets, while the vulnerabilities remain undetected. Software vendors can buy exploits via in-house or community-run bug bounty reward programs. Some agents such as national security agencies and companies may also buy exploits. For example, ZeroDium is a zero-day acquisition firm, who buys high-risk vulnerabilities with premium rewards, then resells them mostly to government clients~ (\cite{Fisher2015:VUPEN}).
% Another similar company is Vupen, which offers a subscription service to its clients, providing vulnerability data and exploits for zero-days and other bugs~ (\cite{Fisher2015:VUPEN}). The zero-day exploit market can be regarded as a public project problem (\cite{Egelman2013:Markets}).

\subsubsection{Affine Maximizer Auctions (AMA)}

Myerson's (\cite{myerson1981optimal}) solved for the optimal single-item auction. For combinatorial auctions, Myerson's technique does not generalize beyond single-parameter settings. Revenue maximizing mechanism design remains an open problem. Many revenue-boosting techniques were proposed by researchers. The Affine Maximizer Auctions (AMA) mechanisms are a family of strategy-proof mechanisms that are characterized by a set of parameters (\cite{Likhodedov2005:Approximating}).
By focusing on the AMA mechanisms, the mechanism designer can focus on tuning the mechanism parameters in order to achieve better revenue. This way, mechanism design is no longer a functional optimisation process, but rather a value optimisation process, which is often much easier.

% The Affine Maximizer Auctions (AMA) mechanisms are a family of mechanisms which are summarised by these techniques (\cite{Likhodedov2005:Approximating}).   They provide a good execution mechanism with a set of parameters to maximize revenue.  Based on AMA mechanism, the original zero-day market problem is transformed into a value optimization problem that only needs to adjust parameters. \\
% We need to study mechanism models for certain problem instances and choose suitable machine learning methods to optimize the parameters of this model.
The family of AMA mechanisms for the public project model was formally defined by Guo et al. (\cite{Guo2016:Revenue}) as follows:

\begin{tcolorbox}
    \begin{center}
        AMA Mechanisms
    \end{center}
\begin{itemize}
    \item Given a type profile $\theta$, the outcome picked is the following:
        \[o^*=\arg\max_{o\in O}\left(\sum_{i=1}^nu_iv_i(\theta_i,o)+a_o\right)\]

    \item Agent $i$'s payment equals:
        \[\frac{\max_{o\in O}\left(\sum_{j\neq i}u_jv_j(\theta_j,o)+a_o\right)
        - \sum_{j\neq i}u_jv_j(\theta_j,o^*)-a_{o^*}}{u_i}\]
    \end{itemize}

\end{tcolorbox}

Here, $O$ represents the outcome space, $\Theta_i$ represents agent $i$'s type space, and $v_i(\theta_i,o)$ represents agent $i$'s valuation for
outcome $o\in O$ when her type is $\theta_i\in \Theta_i$.\\
%Under the model described in \ref{Model Description6.1}, the outcome space is $[0,1]$. To more specific, an outcome $o \in [0,1]$ represents when the exploit is ended/useless (revealed to the defenders).

% \subsubsection{Price-oriented Rationing-free Mechanism (PORF)}
% In 2002, Yokoo et al. (\cite{yokoo2002price,Yokoo2003:Characterization}) proposed a Price-oriented Rationing-free Mechanism (PORF) mechanism. The characteristics of a PORF are as follows: For each agent, his price is determined based on the declared evaluation values of other agents ($-i$), while it is independent of its own declaration. Then, each agent can choose the option that maximizes utility independently of the allocations of other agents. \\

\section{Machine Learning Methods for Automated Mechanism Design}

Machine learning is a useful tool for finding an acceptable or near-optimal results for optimization problems. Typical machine learning methods include evolutionary computation algorithm, particle swarm optimization algorithm, artificial neural network algorithm, etc. (\cite{mohri2018foundations})\\

%edition_begin
In 2003, Sandholm (\cite{sandholm2003automated}) reported that search algorithms is a new approach to solve the automated mechanism design problem. In 2005, Balcan et al. (\cite{balcan2005mechanism}) obtained a unified approach for a variety of revenue-maximizing mechanism design problems. They used sample-complexity techniques in machine learning theory to simplify the design of revenue-maximizing incentive-compatible mechanisms to standard algorithmic questions. Various novel machine learning methods have been proposed for AMD (\cite{narasimhan2016automated}). In the current "Era of Neural Networks" (\cite{sharma2017era}), many researchers used machine learning methods to find optimal mechanisms (\cite{alaei2012bayesian,cai2012algorithmic}) or near-optimal mechanisms (\cite{hart2017approximate,hartline2009simple,yao2014n}). \\

%edition_end

In this thesis, we mainly use evolutionary computation and artificial neural networks to find optimal/near-optimal mechanisms for public project problems. \\

\subsection{Evolutionary Computation}
Evolutionary computation techniques can produce highly optimized solutions by mutation and crossover, which generates new genotypes to find good solutions for a given problem.  Evolutionary computation is widely used for a long range of computational tasks (\cite{Neumann2019:Evolutionary,Long2020:Evolutionary,Do2021:Analysis}).
Researchers can apply evolutionary computation to mechanism design by treating mechanism design as an engineering problem and bring in engineering design principles (\cite{phelps2010evolutionary}). We categorize these approaches under the banner of evolutionary mechanism design.
Andrews (\cite{andrews1994genetic}) is the first to apply evolutionary computation to the double-auction design problem with a view to automating the mechanism design process.
Cliff (\cite{cliff1998evolving}) used evolutionary
search to explore the parameter space of the zip strategy. After that, more researchers used evolutionary computations to study mechanism design problems (\cite{conitzer2007incremental,cliff2002evolution,phelps2002co,phelps2003applying}).\\

Evolutionary computation (in the context of mechanism design) usually involves the following steps:
\begin{enumerate}
\item Initialization: An initial batch of mechanisms is created when evolutionary computation starts.
\item Genetic operators (such as crossover and mutation): Small random changes are introduced to the existing mechanisms in the population (\cite{fox1991genetic}).
\item Selection: As weaker mechanisms are stochastically removed, the mechanism population becomes refined.
\end{enumerate}

% The main advantage of this method is that it could use the AMD objectives to search for optimal results in the large feasible regions. The disadvantage is that the method is not accurate enough compared to other methods. It may get local optimal or the inaccuracy global minimal for public project problems.\\

\subsection{Artificial Neural Network}

A neural network is often formed by thousands of neurons, which are grouped
into different layers to pass and process data from the input layer to the
output layer. The internal layers are called hidden layers and usually, they are
fully connected with their neighbour hidden layers. This pattern forms the
structure of the network. Activation functions are often applied to every layer
(\cite{sharma2017era}).\\

%edition_begin

Recent development of automated mechanism design capitalizes on deep learning. Dütting et al. (\cite{dutting2019optimal}) in 2019 first used the neural network (called \textbf{RegretNet}) to solve automatic mechanism design problems for learning approximately strategy-proof auctions for multi-bidder multi-item auctions. They also proposed another neural network (called \textbf{RochetNet}) for a single bidder, which is perfectly strategy-proof ensured by network construction. Feng et al. (\cite{feng2018deep}) and Golowich et al. (\cite{golowich2018deep}) extended \textbf{RegretNet} to deal with different constraints and objectives.
Sakurai et al. (\cite{Sakurai2019:Deep}) adopted the idea of \textbf{RegretNet} to design false-name-proof mechanisms, where false-name-proof violations are added to the cost function, which is then minimized to remove false-name manipulation. Curry et al. (\cite{curry2020certifying}) modified \textbf{RegretNet} to be able to verify strategy-proofness of the auction mechanism learned by neural network.  Peri et al. (\cite{peri2021preferencenet}) developed \textbf{PreferenceNet} to encode human preference (e.g. fairness) constraints into \textbf{RegretNet}.\\

Later on, many other researchers explored the use of neural networks in automating mechanism design (\cite{Manisha2018:Learning,shen2018automated,rahme2020auction,duan2022context,rahme2020permutation,zhan2020incentive,bichler2021learning,brero2020reinforcement}).\\

For example, Shen et al. (\cite{shen2018automated}) used neural networks
(called \textbf{MenuNet}) to automatically design revenue optimal mechanisms
for multi-item revenue optimization settings (selling two items to one
buyer). They theoretically proved that the mechanism framework is
indeed optimal. Rahme et al. (\cite{rahme2020auction}) proposed \textbf{ALGNet}
to solve two-player game through parameterizing the misreporter. Then Rahme et
al. (\cite{rahme2020permutation}) proposed a permutation-equivariant
architecture called EquivariantNet for symmetric auctions. Duan et al. (\cite{duan2022context}) proposed
\textbf{CITransNet} which focused on permutation equivariant auctions without
symmetric constraints. Zhan et al. (\cite{zhan2020incentive}) used a
reinforcement learning-based (DRL-based) solution that can automatically learn
the best pricing strategy. Bichler et al. (\cite{bichler2021learning}) used neural network to
study Bayesian Nash equilibrium strategies. They
developed a neural pseudogradient ascent (NPGA) methods to learn the equilibrium bid
functions.\\

%edition_end
All evidences suggest that, for many settings, it is possible for the neural network to find optimal or near-optimal truthful mechanisms with proper loss functions and network structures.\\

For example, \textbf{RegretNet} (\cite{dutting2019optimal}) used loss function to
ensure approximate strategy-proofness. Whenever the training samples cause
truthfulness violations, RegretNet reports a large loss so loss function can force the neural network based mechanism to learn toward strategy-proofness.\\

Another way to achieve truthfulness is via adopting specific ``truthful'' neural network structures.
Both \textbf{RochetNet} (\cite{dutting2019optimal}) and \textbf{MenuNet} (\cite{shen2018automated}) were restricted to a single agent, but enforced strategy-proofness at the architectural level (\cite{curry2022differentiable}). \textbf{ALGNet} and \textbf{CITransNet} focused on a small number of agents (such as two bidders) for permutation or symmetric equivariant settings. In this thesis, we propose a \textbf{PORF network} which is a strategy-proof neural network enforced by network structure.\\

As an example truthful neural network, Figure \ref{Manisha} shows how Manisha et al. used the network structure to ensure the strategy-proof constraint (\cite{Manisha2018:Learning}). The input for network is $\theta_{-i}$ (n-1 dimensions), and the output (redistribution payment) for agent $i$ is $h(\theta_{-i})$ which is independent from agent $i$'s type $\theta_{i}$. Since an agent cannot impact her redistribution function under this network, we know for sure that no matter how the parameters are trained, the overall mechanism is always truthful.

\begin{figure}[H]
\centering
\includegraphics[width=0.5\textwidth]{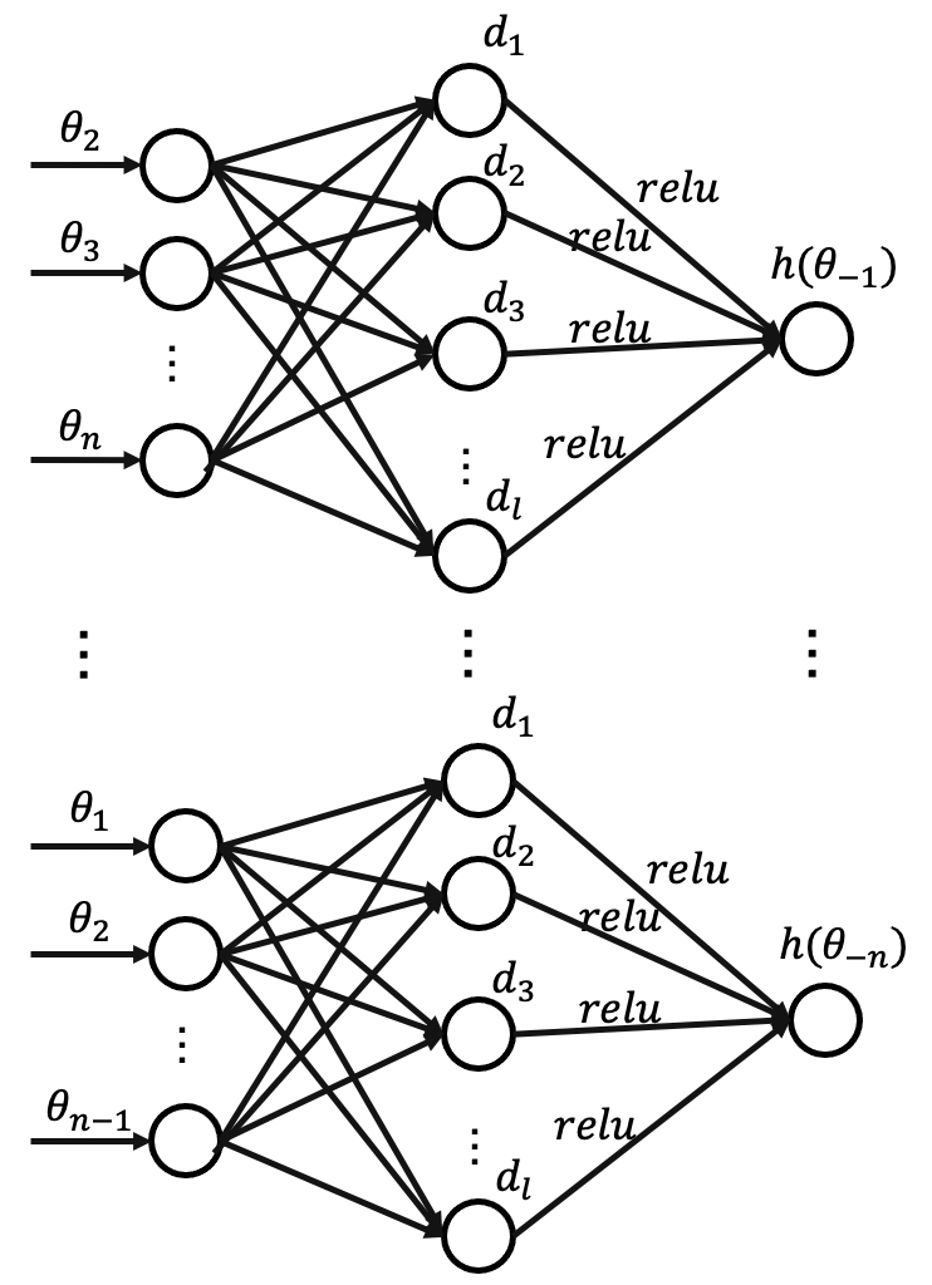}
\caption[Manisha's network model]{Manisha's nonlinear network model (\cite{Manisha2018:Learning})}\label{Manisha}
\end{figure}

% For the public project problem, one of the key points is whether to build or not to build, and it is difficult for the neural network to get the gradients of the binary options.
% Well-designed networks are capable of dealing with the binary constraints. Specifically designed neural networks are needed in order to deal with the binary decision making involved in public project mechanisms. \\

In this thesis, we proposed a few novel neural networks for automated mechanism design,
such as networks based on price-oriented rationing-free (PORF) interpretation of strategy-proof iterative (largest unanimous) mechanisms (\cite{yokoo2002price,Yokoo2003:Characterization}),
and GAN-assisted networks for worst-case mechanism design, where generative adversarial network (GAN) is used to generate worst-case type profiles (\cite{goodfellow2014generative,goodfellow2020generative}).\\

\subsubsection{Unsupervised Neural Network}

Traditionally, neural networks are supervised with labeled training data for recognition tasks. For example, image/face recognition started off as pure supervised but has become hybrid by adopting unsupervised pre-training (\cite{martins2007unsupervised}). Typical unsupervised neural network methods include autoencoders, deep belief networks, generative adversarial networks (GAN), and self-organizing maps, etc. (\cite{mohri2018foundations})\\

\subsubsection{Generative Adversarial Network (GAN)}
In 2014, Goodfellow et al. (\cite{goodfellow2014generative}) first proposed a creative neural network structure to find the ``worst cases'' for a neural network (discriminator network).
Then it immediately became one of the most important neural network frameworks (\cite{yu2017seqgan,yi2019generative}).
Generative adversarial networks (GAN) usually include two different neural networks: a generator network and a discriminator network.\\

The generator network generates candidates' data while the discriminator
network evaluates them. Specific labeled training data are not necessary for
GAN. The generator is trained based on whether it succeeds in fooling the
discriminator. The discriminator network is trained with samples from the
generator and sometimes from the initial training data set. The generator and
the discriminator are trained alternatively in each epoch.

%% file: Chapters/Chapter3.tex
\chapter{Mechanism Design for Public Projects via Neural Networks} % Main chapter title

\label{Mechanism Design for Public Projects via Neural Networks} % Change X to a consecutive number; for referencing this chapter elsewhere, use \ref{ChapterX}

\section{Introduction}\label{sec:intro}
Many multiagent system applications (\emph{e.g.},
crowdfunding) are related to the public project problem. The public project
problem is a classic economic model that has been studied extensively in both
economics and computer
science~ (\cite{Mas-Colell1995:Microeconomic,Moore2006:General,Moulin1988:Axioms}).
Under this model, a group of agents decide whether or not to fund a
\emph{nonrivalrous} public project --- when one agent consumes the project, it does not prevent others from using it.
We study both the \textbf{non-excludable} and the \textbf{excludable} versions
of the \emph{binary} \emph{non-rivalrous} \emph{indivisible} public project problem. The binary decision is either to
build or not. If the decision is not to build, then no agents can consume the
project.  For the \emph{nonexcludable} version, once a project is built, all
agents can consume it, including those who do not pay.  For example, if the
public project is an open source software project, then once the project is
built, everyone can consume it.  For the \emph{excludable} version, the mechanism has the capability to exclude agents from the built project. For example, if the public project is a swimming pool, then we could impose the restriction that only the paying agents have access to it.\\

Our aim is to design mechanisms that maximize \emph{expected} performances.  We consider two design objectives. One is to maximize the \textbf{expected number
of consumers} (expected number of agents who are allowed to consume the
project).\footnote{For the nonexcludable version, this is simply
to maximize the probability of building.} The other objective is
to maximize the agents' \textbf{expected agents' welfare}.
We argue that maximizing the expected number of consumers is \emph{more fair}
in some applications. When maximizing the number of consumers, agents with lower valuations are treated as important as high-value agents.\\

With slight technical adjustments, we adopt the existing characterization
results from Ohseto~ (\cite{Ohseto2000:Characterizations}) for
\emph{strategy-proof} and \emph{individually rational} mechanisms for both the
nonexcludable and the excludable public project problems.  
Under various conditions, we show that existing mechanisms or mechanisms derived via classic mechanism design approaches are optimal
or near optimal. When existing mechanism design approaches do not
suffice, we propose a neural network based approach, which successfully
identifies better performing mechanisms.  Mechanism design via deep
learning/neural networks has been an emerging topic~ (\cite{Golowich2018:Deep,
Duetting2019:Optimal,Shen2019:Automated,Manisha2018:Learning}).  Duetting
\emph{et.al.}~ (\cite{Duetting2019:Optimal}) proposed a general approach for
revenue maximization via deep learning. The high-level idea is to manually
construct often complex network structures for representing mechanisms for
different auction types. The cost function is the negate of the revenue. By minimizing the cost function via gradient descent, the network parameters are
adjusted, which leads to better performing mechanisms.  The mechanism design
constraints (such as strategy-proofness) are enforced by adding a penalty term
to the cost function. The penalty is calculated by sampling the type profiles
and adding together the constraint violations.  Due to this setup, the final
mechanism is only approximately strategy-proof. The authors demonstrated that
this technique scales better than the classic mixed integer programming based
automated mechanism design approach~ (\cite{Conitzer2002:Complexity}).  Shen
\emph{et.al.}~ (\cite{Shen2019:Automated}) proposed another neural network based
mechanism design technique, involving a seller's network and a buyer's network.
The seller's network provides a menu of options to the buyers.  The buyer's
network picks the utility-maximizing menu option. An exponential-sized
hard-coded buyer's network is used (\emph{e.g.}, for every discretized type
profile, the utility-maximizing option is pre-calculated and stored in the
network).  The authors mostly focused on settings with only one buyer.\\

Our approach is different from previous approaches, and it involves three
  technical innovations, which can be applied to neural network-based mechanism
  design in general.
  $$$$
  
\vspace{.1in}\noindent \emph{Calculating mechanism decisions off the network by
interpreting mechanisms as price-oriented rationing-free (PORF)
mechanisms~ (\cite{Yokoo2003:Characterization}) :} A mechanism often involves
binary decisions.
A common way to model binary decisions on neural networks is by using the
\emph{sigmoid} function (or similar activation functions).  A mechanism may
involve a complex decision process, which makes it impractical to
model via \emph{static} neural networks.  For example, for our setting, a
mechanism involves \emph{iterative} decision making where the number of ``rounds'' depends on the agents' types. We could stack multiple
sigmoid functions to model this.  However, stacking sigmoid functions leads to
vanishing gradients and significant numerical errors. Instead, we rely on the
PORF interpretation: every agent faces a set of options (outcomes with prices)
determined by the other agents. We single out a randomly chosen agent $i$, and
draw a sample of \emph{the other agents' types $v_{-i}$}.  We use a separate program (off the
network) to calculate the options $i$ would face. For example, the separate
program can be any Python function, so it is trivial to handle complex and
iterative decision making. We no longer need to construct complex network
structures like the approach in~ (\cite{Duetting2019:Optimal}) or resort to
exponential-sized hard-coded buyer networks like the approach
in~ (\cite{Shen2019:Automated}).  After calculating $i$'s options, we link the
options together using terms that contain network parameters, which enables backpropagation.  One effective way to do
this is by making use of the prior distribution as discussed below.

\newpage
\vspace{.1in}\noindent \emph{Feeding prior distribution into the cost
function:} In conventional machine learning, we have access to a finite set of
samples, and the process of machine learning is essentially to infer the true
probability distribution of the samples. For existing neural network mechanism
design
approaches~ (\cite{Duetting2019:Optimal,Shen2019:Automated})
(as well as this chapter), it is assumed that the prior distribution is known.
After calculating agent $i$'s options, we make use of $i$'s distribution to
figure out the probabilities of all the options, and then derive the expected
objective value from $i$'s perspective. We assume that the prior distribution is continuous. If we have the \emph{analytical form}
of the prior distribution, then the probabilities can
provide quality gradients for our training process. This is due to the fact that probabilities are
calculated based on neural network outputs. In summary, we
combine both samples and distribution in our cost function.
(In classic machine learning, the cost function only involves the samples.)
$$ $$

\vspace{.1in}\noindent
\emph{Supervision to manual mechanisms as initialization:} We start our
training by first conducting supervised learning. We teach the network to mimic
an existing manual mechanism, and then leave it to gradient descent. This is
essentially a systematic way to improve manual mechanisms.
In our experiments,
besides the \emph{serial cost sharing mechanism}, we also considered two
heuristic-based manual mechanisms as starting points. One heuristic is feasible
but not optimal, and the gradient descent process is able to improve its
performance. The second heuristic is not always feasible, and the gradient
descent process is able to fix the constraint violations. Supervision to manual
mechanisms is often better than random initializations.  For one thing, the
supervision step often pushes the performance to a state that is already
somewhat close to optimality.  It may take a long time for random
initializations to catch up. In computational expensive scenarios, it may never
catch up.  Secondly, supervision to a manual mechanism is a systematic way to
set good initialization point, instead of trials and errors.  It should be noted
that for many conventional deep learning application domains, such as computer
vision, well-performing manual algorithms do not exist. Fortunately, for
mechanism design, we often have simple and well-performing mechanisms to be
used as starting points.

%%%%%%%%%%%%%%%%%%%%%%%%%%%%%%%%%%%%%%%%%%%%%%%%%%%%%%%%%%%%%%%%%%%%%%%%
\newpage
\section{Model Description}
$$$$
$n$ agents need to decide whether or not to build a public project.  The
project is \emph{binary} (build or not build) and \emph{nonrivalrous} (the cost
of the project does not depend on how many agents are consuming it).  We
normalize the project cost to $1$.  Agent $i$'s type $v_i\in[0,1]$ represents
her private valuation for the public project. We assume that the $v_i$ are
drawn \emph{i.i.d.} from a known prior distribution. Let $F$ and $f$ be the CDF
and PDF, respectively. We assume that the distribution is continuous and $f$ is
differentiable.

\begin{itemize}

    \item For the nonexcludable public project model, agent $i$'s valuation is
        $v_i$ if the project is built, and $0$ otherwise.

    \item For the excludable public project model, the outcome space is
        $\{0,1\}^n$.  Under outcome $(a_1,a_2,\ldots,a_n)$, agent $i$ consumes
        the public project if and only if $a_i=1$. If for all $i$, $a_i=0$,
        then the project is not built.  As long as $a_i=1$ for some $i$, the
        project is built.\\

\end{itemize}

We use $p_i\ge 0$ to denote agent $i$'s payment. We require that $p_i=0$ for
all $i$ if the project is not built and $\sum p_i=1$ if the project is built.
An agent's payment is also referred to as her \emph{cost share} of the project.
An agent's utility is $v_i-p_i$ if she gets to consume the project, and $0$
otherwise.\\

We focus on \emph{strategy-proof} and \emph{individually rational} mechanisms.
We study two objectives. One is to maximize the expected number of consumers.
The other is to maximize the agents' welfare.

%%%%%%%%%%%%%%%%%%%%%%%%%%%%%%%%%%%%%%%%%%%%%%%%%%%%%%%%%%%%%%%%%%%%%%%%
\newpage
\section{Characterizations and Bounds}
$$ $$
We adopt a list of existing characterization results
from~ (\cite{Ohseto2000:Characterizations}), which characterizes strategy-proof
and individual rational mechanisms for both nonexcludable and excludable public
project problems.  A few technical adjustments are needed for the existing
characterizations to be valid for our problem.  The characterizations
in~ (\cite{Ohseto2000:Characterizations}) were not proved for quasi-linear
settings. However, we verify that the assumptions needed by the proofs are
valid for our model setting. One exception is that the characterizations
in~ (\cite{Ohseto2000:Characterizations}) assume that every agent's valuation is
strictly positive.  This does not cause issues for our objectives as we are
maximizing for expected performances and we are dealing with continuous
distributions.\footnote{Let $M$ be the optimal mechanism. If we restrict the
valuation space to $[\epsilon,1]$, then $M$ is Pareto dominated by an
unanimous/largest unanimous mechanism $M'$ for the nonexcludable/excludable
setting. The expected performance difference between $M$ and $M'$ vanishes as
$\epsilon$ approaches $0$. Unanimous/largest unanimous mechanisms are still
strategy-proof and individually rational when $\epsilon$ is set to exactly
$0$.} We are also safe to drop the \emph{citizen sovereign} assumption
mentioned in one of the characterizations\footnote{If a mechanism always
builds, then it is not individually rational in our setting.  If a mechanism
always does not build, then it is not optimal.}, but not the other two minor
technical assumptions called \emph{demand monotonicity} and \emph{access
independence}.\\\\

\subsection{Nonexcludable Mechanism Characterization}
$$$$

\begin{definition}[Unanimous mechanism~ (\cite{Ohseto2000:Characterizations})]
    There is a constant cost share vector $(c_1,c_2,\ldots,c_n)$ with $c_i\ge
    0$ and $\sum c_i=1$. The mechanism builds if and only if $v_i\ge c_i$ for
    all $i$. Agent $i$ pays exactly $c_i$ if the decision is to build.  The
unanimous mechanism is strategy-proof and individually rational.\\
\end{definition}

\newpage
\begin{theorem}[Nonexcludable mech. characterization~ (\cite{Ohseto2000:Characterizations})]
    For the nonexcludable public project model,
if a mechanism is strategy-proof, individually rational, and citizen
sovereign, then it is weakly Pareto dominated by an unanimous mechanism.

\noindent Citizen sovereign: Build and not build are both possible outcomes.\\
\end{theorem}

Mechanism $1$ weakly Pareto dominates Mechanism $2$ if every agent
weakly prefers Mechanism $1$ under every type profile.\\

\begin{example}[Conservative equal costs mechanism~ (\cite{Moulin1994:Serial})] An
example unanimous mechanism works as follows: we build the project if and only
if every agent agrees to pay $\frac{1}{n}$.\\  \end{example}

\subsection{Excludable Mechanism Characterization}
$$$$
\begin{definition}[Largest unanimous
    mechanism~ (\cite{Ohseto2000:Characterizations})] For every nonempty coalition
    of agents $S = \{S_1,S_2,\ldots,S_k\}$, there is a constant cost share
    vector $C_S=(c_{S_1},c_{S_2},\ldots,c_{S_k})$ with $c_{S_i}\ge 0$ and
    $\sum_{1\le i\le k} c_{S_i}=1$.  $c_{S_i}$ is agent $S_i$'s cost share
    under coalition $S$.
    If agent $i$ belongs to two coalitions $S$ and $T$ with $S\subsetneq T$,
    then $i$'s cost share under $S$ must be greater than or equal to her cost
    share under $T$.
    Agents in $S$ unanimously approve the cost share
    vector $C_S$ if and only if $v_{S_i}\ge c_{S_i}$ for all $i \in S$.
    The mechanism picks the largest coalition $S^*$ satisfying that $C_{S^*}$
    is unanimously approved.  If $S^*$ does not exist, then the decision is not
    to build.  If $S^*$ exists, then it is always unique, in which case the
    decision is to build. Only agents in $S^*$ are consumers and they pay according to $C_{S^*}$.
The largest unanimous mechanism is strategy-proof and individually rational.\\
\end{definition}

\newpage
The mechanism essentially specifies a constant cost share vector for every
  non-empty coalition.
The cost share vectors must satisfy that if we remove some agents from a
  coalition, then under the remaining coalition, every remaining agent's cost
  share must be equal or higher.
The largest unanimously approved coalition become the consumers/winners and
  they pay according to this coalition's cost share vector.
The project is not built if there are no unanimously approved coalitions.\\

Another way to interpret the mechanism is that the agents start with the grand coalition of all agents.
Given the current coalition, if some agents do not approve their cost shares, then they are forever removed.
The remaining agents form a smaller coalition, and they face increased cost shares.
We repeat the process until all remaining agents approve their shares, or when
  all agents are removed.\\

\begin{theorem}[Excludable mech.
    characterization~ (\cite{Ohseto2000:Characterizations})] For the excludable
    public project model, if a mechanism is strategy-proof, individually
    rational, and satisfies the following assumptions, then it is weakly Pareto
    dominated by a largest unanimous mechanism.

Demand monotonicity: Let $S$ be the set of consumers.  If for every agent $i$
    in $S$, $v_i$ stays the same or increases, then all agents in $S$ are still
    consumers.  If for every agent $i$ in $S$, $v_i$ stays the same or
    increases, and for every agent $i$ not in $S$, $v_i$ stays the same or
    decreases, then the set of consumers should still be $S$.

    Access independence: For all $v_{-i}$, there exist $v_i$ and $v_i'$ so that
agent $i$ is a consumer under type profile $(v_i,v_{-i})$ and is not a consumer
under type profile $(v_i',v_{-i})$. \\ \end{theorem}

\begin{example}[Serial cost sharing mechanism~ (\cite{Moulin1994:Serial})] Here is
    an example largest unanimous mechanism.  For every nonempty subset of
    agents $S$ with $|S|=k$, the cost share vector is
    $(\frac{1}{k},\frac{1}{k},\ldots,\frac{1}{k})$.  The mechanism picks the
    largest coalition where the agents are willing to pay equal shares.\\
\end{example}

%     Another interpretation of this mechanism is the follow iterative process:
%     Initially, all $n$ agents are offered a cost share of $\frac{1}{n}$. If all
%     agents agree, then we build the project and every agent pays $\frac{1}{n}$.
%     If some agents disagree, then they are forever removed.  The remaining $k$
%     agents are offered a cost share of $\frac{1}{k}$. We repeat this until all
%     remaining agents agree or all agents are removed. If all agents are
%     removed, then the decision is not to build.  An agent may face multiple
%     cost share offers throughout the process, but the cost share offers are
%     nondecreasing, so there is no room for strategic manipulation.

    % When an agent reports truthfully under a largest unanimous mechanism, her
    % utility is maximized because her cost share is completely determined by the
    % others. For example, if there are three agents in total, and agents $2$ and
    % $3$ report $\frac{1}{2}$ and $\frac{1}{3}$ respectively, then agent $1$'s
    % cost share offer is $\frac{1}{3}$. Agent $1$ is a consumer if and only if
    % she can afford $\frac{1}{3}$. If agents $2$ and $3$ report $\frac{1}{2}$
    % and $\frac{1}{4}$ respectively, then agent $1$'s cost share offer is
%$\frac{1}{2}$.

\newpage
Deb and Razzolini~ (\cite{Deb1999:Voluntary}) proved that if we further require an
\emph{equal treatment of equals} property (if two agents have the same type,
then they should be treated the same), then the only strategy-proof and
individually rational mechanism left is the serial cost sharing mechanism.  For
many distributions, we are able to outperform the serial cost sharing mechanism.
That is, equal treatment of equals (or requiring anonymity) does hurt
performances.\\

$$$$
\subsection{Nonexcludable Public Project Analysis}\label{sub:nonexcludable}

We start with an analysis on the nonexcludable public project. The results
presented in this section will lay the foundation for the more complex
excludable public project model coming up next.\\

Due to the characterization results, we focus on the family of unanimous
mechanisms. That is, we are solving for the optimal cost share vector
$(c_1,c_2,\ldots,c_n)$, satisfying that $c_i\ge 0$ and $\sum c_i=1$.\\

Recall that $f$ and $F$ are the PDF and CDF of the prior distribution.  The
\emph{reliability function} $\overline{F}$ is defined as $\overline{F}(x)=1-F(x)$.  We
define $w(c)$ to be the expected utility of an agent when her cost share is
$c$, conditional on that she accepts this cost share.
\[w(c)=\frac{\int_c^1 (x-c)f(x)dx}{\int_c^1f(x)dx}\]
One condition we will use is \emph{log-concavity}:
if $\log(f(x))$ is concave in $x$, then $f$ is log-concave.
We also introduce another condition called \emph{welfare-concavity}, which requires $w$ to be concave.\\%, which is equivalent to $f(x)+xf'(x)\ge 0$.

\begin{theorem}\label{thm:nonexcludable}
If $f$ is log-concave, then the conservative equal costs mechanism maximizes the expected
number of consumers.
If $f$ is log-concave and welfare-concave, then the conservative equal costs mechanism
maximizes the expected agents' welfare.\\
\end{theorem}

\begin{proof} Let $C=(c_1,c_2,\ldots,c_n)$ be the cost share vector. Maximizing
    the expected number of consumers is equivalent to maximizing the
    probability of $C$ getting unanimously accepted, which equals $\overline{F}(c_1)
    \overline{F}(c_2) \ldots \overline{F}(c_n)$.  Its log equals
    $\sum_{1\le i\le n}\log(\overline{F}(c_i))$.  When $f$ is log-concave, so is
    $\overline{F}$ according to~ (\cite{Bagnoli2005:Log}).  This means that when cost
    shares are equal, the above probability is maximized.\\

    The expected agents' welfare
    under the cost share vector $C$ equals $\sum w(c_i)$, conditional on all
    agents accepting their shares. This is maximized when shares are equal.
    Furthermore, when all shares are equal, the probability of unanimous
approval is also maximized.\\  \end{proof}

$f$ being log-concave is also called the \emph{decreasing reversed failure
rate} condition~ (\cite{Shao2016:Optimal}).  Bagnoli and
Bergstrom~ (\cite{Bagnoli2005:Log}) proved log-concavity for many common
distributions, including the distributions in Table~\ref{tb:logconcave} (for
all distribution parameters).  All distributions are restricted to $[0,1]$.
We also list some limited results for welfare-concavity.
We prove that the uniform distribution is welfare-concave, but for the other
distributions, the results are based on simulations.
Finally, we include the conditions for $f$ being nonincreasing, which will be used in the excludable public project model.\\

\begin{table}[ht]
\caption{Example Log-Concave Distributions}
\centering
\begin{tabular}{ l c r }\label{tb:logconcave}
    & Welfare-Concavity & Nonincreasing \\
    Uniform $U(0,1)$ & Yes & Yes \\
    \hline
    %Normal & $\sigma^2+\mu\ge 1$ & $\mu\le 0$ \\
    Normal & No ($\mu=0.5,\sigma=0.1$) & $\mu\le 0$ \\
    \hline
  % Exponential & $\lambda\le 1$ & Yes \\
    Exponential & Yes ($\lambda=1$) & Yes \\
    \hline
  %Logistic & $1-(e^{-\frac{1-\mu}{s}})^2\le s(1+e^{-\frac{1-\mu}{s}})^2$ & $\mu\le 0$ \\
    Logistic & No ($\mu=0.5,\sigma=0.1$) & $\mu\le 0$ \\
\end{tabular}
\end{table}
$$$$

Even when optimal, the conservative equal costs mechanism performs poorly.  We
take the uniform $U(0,1)$ distribution as an example. Every agent's cost share
is $\frac{1}{n}$.  The probability of acceptance for one agent is
$\frac{n-1}{n}$, which approaches $1$ asymptotically. However, we need
unanimous acceptance, which happens with much lower probability.  For the
uniform distribution, asymptotically, the probability of unanimous acceptance
is only $\frac{1}{e}\approx 0.368$. In general, we have the following bound:\\

\begin{theorem}
If $f$ is Lipschitz continuous, then when $n$ goes to infinity, the probability of unanimous
    acceptance under the conservative equal costs mechanism is $e^{-f(0)}$.\\
\end{theorem}

Without log-concavity, the conservative equal costs mechanism is not
necessarily optimal. We present the following dynamic program (DP) for calculating
the optimal unanimous mechanism. We only present the formation for welfare
maximization.\footnote{Maximizing the expected number of consumers can be viewed as a
special case where every agent's utility is $1$ if the project is built} We assume that there is an ordering of the agents based on their identities.  We define
$B(k,u,m)$ as the maximum expected agents' welfare under the following conditions:

\begin{itemize}
    \item The first $n-k$ agents have already approved their cost shares, and their total
        cost share is $1-m$. That is, the remaining $k$ agents need to come up with $m$.\\
    \item The first $n-k$ agents' total expected utility is $u$.\\
\end{itemize}

The optimal agents' welfare is then $B(n,0,1)$. We recall that $\overline{F}(c)$ is the probability
that an agent accepts a cost share of $c$, we have\\
\[
    B(k,u,m)=\max_{0\le c\le m}\overline{F}(c)B(k-1,u+w(c), m-c)
\]\\
The base case is $B(1,u,m)=\overline{F}(m)(u+w(m))$.  In terms of implementation of
this DP, we have $0\le u\le n$ and $0\le m\le 1$. We
need to discretize these two intervals. If we pick a discretization size of
$\frac{1}{H}$, then the total number of DP subproblems is
about $H^2n^2$.\\

To compare the performance of the conservative equal costs mechanism and our DP
solution, we focus on distributions that are not log-concave (hence, uniform
and normal are not eligible).  We introduce the following non-log-concave
distribution family:

\begin{definition}[Two-Peak Distribution $(\mu_1,\sigma_1,\mu_2,\sigma_2,p)$]
    With probability $p$, the agent's valuation is drawn from the normal
    distribution $N(\mu_1,\sigma_1)$ (restricted to $[0,1]$).
    With probability $1-p$, the agent's valuation is drawn from $N(\mu_2,\sigma_2)$ (restricted to $[0,1]$).\\
\end{definition}
$$ $$
The motivation behind the two-peak distribution is that there may be two
categories of agents. One category is directly benefiting from the public
project, and the other is indirectly benefiting. For example, if the public
project is to build bike lanes, then cyclists are directly benefiting, and the
other road users are indirectly benefiting (\emph{e.g.}, less congestion for
them).  As another example, if the public project is to crowdfund a piece of
security information on a specific software product (\emph{e.g.}, PostgreSQL),
then agents who use PostgreSQL in production are directly benefiting and the
other agents are indirectly benefiting (\emph{e.g.}, every web user is pretty
much using some websites backed by PostgreSQL).  Therefore, it is natural to
assume the agents' valuations are drawn from two different
distributions. For simplicity, we do not consider three-peak, \emph{etc.}\\

For the two-peak distribution $(0.1,0.1,0.9,0.1,0.5)$, DP significantly
outperforms the conservative equal costs (CEC) mechanism.\\

\begin{center}
\begin{tabular}{ l c r }
    & E(no. of consumers) & E(welfare)\\
  n=3 CEC & 0.376 & 0.200 \\
    \hline
  n=3 DP & 0.766 & 0.306 \\
    \hline
  n=5 CEC & 0.373 & 0.199 \\
    \hline
  n=5 DP & 1.426 & 0.591 \\
\end{tabular}
\end{center}

\newpage
\subsection{Excludable Public Project}
Due to the characterization results, we focus on the family of largest
unanimous mechanisms.  We start by showing that the serial cost sharing
mechanism is optimal in some scenarios.\\

\begin{theorem}\label{thm:excludable}\,
$2$ agents case:
If $f$ is log-concave, then the serial cost sharing mechanism maximizes the expected
    number of consumers.
If $f$ is log-concave and welfare-concave, then the serial cost sharing mechanism
maximizes the expected agents' welfare.

$3$ agents case:
If $f$ is log-concave and nonincreasing, then the serial cost sharing mechanism maximizes the expected
    number of consumers.
If $f$ is log-concave, nonincreasing, and welfare-concave, then the serial cost sharing mechanism maximizes the agents' welfare.\\
\end{theorem}

For $2$ agents, the conditions are identical to the nonexcludable case.  For
$3$ agents, we also need $f$ to be nonincreasing.  Example distributions
satisfying these conditions were listed in Table~\ref{tb:logconcave}.\\

\begin{proof}
We only present the proof for welfare maximization when $n=3$, which is the most complex case.
    (For maximizing the number of consumers, all references to the $w$ function should be
    replaced by the constant $1$.)
The largest unanimous mechanism specifies constant cost
shares for every coalition of agents.
We use $c_{1\underline{2}3}$ to denote agent $2$'s cost share when the coalition is $\{1,2,3\}$.
Similarly, $c_{\underline{2}3}$ denotes agent $2$'s cost share when the coalition is $\{2,3\}$.
If the largest unanimous coalition has size $3$, then the expected agents' welfare gained due to this
    case is:
    \[
        \overline{F}(c_{\underline{1}23})
        \overline{F}(c_{1\underline{2}3})
        \overline{F}(c_{12\underline{3}})
        (
        w(c_{\underline{1}23})
        +w(c_{1\underline{2}3})
        +w(c_{12\underline{3}})
        )
    \]
    Given log-concavity of $\overline{F}$ (implied by the log-concavity of $f$) and welfare-concavity,
    and given that $c_{\underline{1}23}+c_{1\underline{2}3}+c_{12\underline{3}}=1$. We have that the above is maximized when all agents have equal shares.\\

    If the largest unanimous coalition has size $2$ and is $\{1,2\}$, then the expected agents' welfare gained due to this
    case is:
    \[
        \overline{F}(c_{\underline{1}2})
        \overline{F}(c_{1\underline{2}})
        F(c_{12\underline{3}})
        (
        w(c_{\underline{1}2})
        +w(c_{1\underline{2}})
        )
    \]\\
    $F(c_{12\underline{3}})$ is the probability that agent $3$ does not join in the coalition.
    The above is maximized when
    $c_{\underline{1}2}=c_{1\underline{2}}$, so it simplifies to
    $2\overline{F}(\frac{1}{2})^2 w(\frac{1}{2}) F(c_{12\underline{3}})$.
    The welfare gain from all size $2$ coalitions is:\\
    \[
        2\overline{F}(\frac{1}{2})^2
        w(\frac{1}{2})(
        F(c_{\underline{1}23})
        +F(c_{1\underline{2}3})
        +F(c_{12\underline{3}})
        )
    \]\\
    Since $f$ is nonincreasing, we have that $F$ is concave, the above is again maximized when all cost shares are equal.\\

    Finally, the probability of coalition size $1$ is $0$, which can be ignored in our analysis.
    Therefore, throughout the proof, all terms referenced are maximized when the cost shares are equal.
\end{proof}
$$$$

For $4$ agents and uniform distribution, we have a similar result.
\begin{theorem}\label{thm:uniform}
    Under the uniform distribution $U(0,1)$, when $n=4$, the serial cost sharing
    mechanism maximizes the expected number of consumers and the expected agents' welfare.
\end{theorem}
$$$$

For $n\ge 4$ and for general distributions, we propose a numerical method for
calculating the performance upper bound.  A largest unanimous mechanism can be
carried out by the following process: we make cost share offers to the agents
one by one based on an ordering of the agents. Whenever an agent disagrees, we
remove this agent and move on to a coalition with one less agent. We repeat
until all agents are removed or all agents have agreed. We introduce the
following mechanism based on a Markov process.  The initial state is
$\{(\underbrace{0,0,\ldots,0}_n),n\}$, which represents that initially, we only
know that the agents' valuations are at least $0$, and we have not made any
cost share offers to any agents yet (there are $n$ agents yet to be offered).
We make a cost share offer $c_1$ to agent $1$.  If agent $1$ accepts, then we
move on to state $\{(c_1,\underbrace{0,\ldots,0}_{n-1}),n-1\}$. If agent $1$
rejects, then we remove agent $1$ and move on to reduced-sized state
$\{(\underbrace{0,\ldots,0}_{n-1}),n-1\}$. In general, let us consider a state
with $t$ users $\{(l_1,l_2,\ldots,l_t),t\}$. The $i$-th agent's valuation lower
bound is $l_i$. Suppose we make offers $c_1,c_2,\ldots,c_{t-k}$ to the first
$t-k$ agents and they all accept, then we are in a state
$\{(\underbrace{c_1,\ldots,c_{t-k}}_{t-k},\underbrace{l_{t-k+1},\ldots,l_{t}}_k),k\}$.
The next offer is $c_{t-k+1}$. If the next agent accepts, then we move on to
$\{(\underbrace{c_1,\ldots,c_{t-k+1}}_{t-k+1},\underbrace{l_{t-k+2},\ldots,l_{t}}_{k-1}),k-1\}$.
If she disagrees (she is then the first agent to disagree), then we move on to
a reduced-sized state
$\{(\underbrace{c_1,\ldots,c_{t-k}}_{t-k},\underbrace{l_{t-k+2},\ldots,l_{t}}_{k-1}),t-1\}$.
Notice that whenever we move to a reduced-sized state, the number of agents yet
to be offered should be reset to the total number of agents in this state.
Whenever we are in a state with all agents offered
$\{(c_1,\ldots,c_t),0\}$, we have gained an objective value of
$t$ if the goal is to maximize the number of consumers.  If the goal is to
maximize welfare, then we have gained an objective value of $\sum_{1\le i\le t}
w(c_i)$.  Any largest unanimous mechanism can be represented via the above
Markov process.  So for deriving performance upper bounds, it suffices to focus
on this Markov process.\\

Starting from a state, we may end up with different objective values. A state
has an expected objective value, based on all the transition probabilities.  We
define $U(t,k,m,l)$ as the maximum expected objective value starting from a
state that satisfies:

\begin{itemize}
    \item There are $t$ agents in the state.

    \item There are $k$ agents yet to be offered.  The first $t-k$ agents
        (those who accepted the offers) have a total cost share of $1-m$. That
        is, the remaining $k$ agents are responsible for a total cost share of $m$.

    \item The $k$ agents yet to be offered have a total lower bound of $l$.\\
\end{itemize}

The upper bound we are looking for is then $U(n,n,1,0)$, which can be calculated
via the following DP process:\\
\[
    U(t,k,m,l) = \max_{\substack{0\le l^*\le l\\l^*\le c^*\le m}} \left( \frac{\overline{F}(c^*)}{\overline{F}(l^*)}U(t,k-1,m-c^*,l-l^*)\right.
\]
\[
    \left.+(1-\frac{\overline{F}(c^*)}{\overline{F}(l^*)})U(t-1,t-1,1,1-m+l-l^*)\right)
\]
$$$$
\\
In the above, there are $k$ agents yet to be offered. We maximize over the next agent's
possible lower bound $l^*$ and the cost share $c^*$. That is, we look for the
best possible lower bound situation and the corresponding optimal offer.
$\frac{\overline{F}(c^*)}{\overline{F}(l^*)}$ is the probability that the next agent
accepts the cost share, in which case, we have $k-1$ agents left. The remaining
agents need to come up with $m-c^*$, and their lower bounds sum up to $l-l^*$.
When the next agent does not accept the cost share, we transition to a new
state with $t-1$ agents in total. All agents are yet to be offered, so $t-1$
agents need to come up with $1$. The lower bounds sum up to $1-m+l-l^*$.\\

There are two base conditions.  When there is only one agent, she
has $0$ probability for accepting an offer of $1$, so $U(1,k,m,l) = 0$.
When there is only $1$ agent yet to be offered,
the only valid lower bound is $l$ and the only sensible offer is $m$. Therefore,\\
\[U(t,1,m,l) = \frac{\overline{F}(m)}{\overline{F}(l)}G(t)+(1-\frac{\overline{F}(m)}{\overline{F}(l)})U(t-1,t-1,1,1-m)\]
\newpage
Here, $G(t)$ is the maximum
objective value when the largest unanimous set has size $t$.  For maximizing
the number of consumers, $G(t)=t$.  For maximizing welfare,\\
\[G(t)= \max_{\substack{c_1,c_2,\ldots,c_t\\c_i\ge 0\\\sum c_i=1}}\sum_i w(c_i)\]
The above $G(t)$ can be calculated via a trivial DP.
\\

Now we compare the performances of the serial cost sharing mechanism against
the upper bounds.  All distributions used here are log-concave.  In every cell,
the first number is the objective value under serial cost sharing, and the
second is the upper bound.  We see that the serial cost sharing mechanism is
close to optimality in all these experiments.  We include both welfare-concave
and non-welfare-concave distributions (uniform and exponential with $\lambda=1$
are welfare-concave). For the two distributions not satisfying
welfare-concavity, the welfare performance is relatively worse.\\
\\
\begin{center}
\begin{tabular}{ l c r }
\Large
    & E(no. of consumers) & E(welfare)\\
    n=5 $U(0,1)$ & 3.559, 3.753 & 1.350, 1.417 \\
    \hline
    n=10 $U(0,1)$ & 8.915, 8.994& 3.938, 4.037 \\
    \hline
    n=5 $N(0.5,0.1)$ & 4.988, 4.993 & 1.492, 2.017 \\
    \hline
    n=10 $N(0.5,0.1)$ & 10.00, 10.00 & 3.983, 4.545 \\
    \hline
    n=5 Exponential $\lambda=1$ & 2.799, 3.038 & 0.889, 0.928 \\
    \hline
    n=10 Exponential $\lambda=1$ & 8.184, 8.476 & 3.081, 3.163 \\
    \hline
    n=5 Logistic$(0.5,0.1)$ & 4.744, 4.781 & 1.451, 1.910 \\
    \hline
    n=10 Logistic$(0.5,0.1)$ & 9.873, 9.886 & 3.957, 4.487 \\
\end{tabular}
\end{center}

\newpage
\begin{example} Here we provide an example to show that the serial cost sharing
    mechanism can be far away from optimality. We pick a simple Bernoulli
    distribution, where an agent's valuation is $0$ with $0.5$ probability and
    $1$ with $0.5$ probability.\footnote{This chapter assumes that the
    distribution is continuous, so technically we should be considering a
    smoothed version of the Bernoulli distribution. For the purpose of
    demonstrating an elegant example, we ignore this technicality.} Under the
    serial cost sharing mechanism, when there are $n$ agents, only half of the
    agents are consumers (those who report $1$s). So in expectation, the number
    of consumers is $\frac{n}{2}$.  Let us consider another simple mechanism.
    We assume that there is an ordering of the agents based on their identities
    (not based on their types). The mechanism asks the first agent to accept a
    cost share of $1$. If this agent disagrees, she is removed from the system.
    The mechanism then moves on to the next agent and asks the same, until an
    agent agrees. If an agent agrees, then all future agents can consume the
    project for free. The number of removed agents follows a geometric
    distribution with $0.5$ success probability. So in expectation, $2$
    agents are removed.  That is, the expected number of consumers is $n-2$.
\end{example}

%%%%%%%%%%%%%%%%%%%%%%%%%%%%%%%%%%%%%%%%%%%%%%%%%%%%%%%%%%%%%%%%%%%%%%%%
$$$$

\section{Mechanism Design vs Neural Networks}
For the rest of this chapter, we focus on the excludable public project model and distributions that are not log-concave.  Due to the characterization results, we only need to consider the largest unanimous mechanisms. We use neural networks and deep learning to solve for well-performing largest unanimous mechanisms. Our approach involves several technical innovations as discussed in Section~\ref{sec:intro}.\\

\subsection{Mechanism Design via Neural Networks} We start with an
overview of automated mechanism design (AMD) via neural networks.  Previous
papers on mechanism design via neural
networks~ (\cite{Golowich2018:Deep,Duetting2019:Optimal,Shen2019:Automated,Manisha2018:Learning})
all fall into this general category.

\begin{itemize}

    \item Use neural networks to represent the full (or a part of the) mechanism.
        Like mechanisms, neural networks are essentially functions with multi-dimensional inputs and outputs.

    \item Training is essentially to adjust the network parameters in order to move toward
        a better performing network/mechanism. Training is just parameter optimization.

    \item Training samples are not real data. Instead, the training
        type profiles are generated based on the known prior distribution. We
        can generate infinitely many fresh samples. We use these generated
        samples to build the cost function, which is often a combination of
        mechanism design objective and constraint penalties.
        The cost function must be differentiable with respect to the network parameters.

    \item The testing data are also type profiles generated based on the known
        prior distribution. Again, we can generate infinitely many fresh
        samples, so testing is based on completely fresh samples. We average
        over enough samples to calculate the mechanism's expected performance.\\

\end{itemize}

\subsection{Network Structure} A largest unanimous mechanism specifies constant
cost shares for every coalition of agents. The mechanism can be characterized
by a neural network with $n$ binary inputs and $n$ outputs. The $n$ binary
inputs present the coalition, and the $n$ outputs represent the constant cost
shares.  We use $\vec{b}$ to denote the input vector (tensor) and $\vec{c}$ to
denote the output vector. We use $NN$ to denote the neural network, so
$NN(\vec{b})=\vec{c}$.
There are several network constraints:

\begin{itemize}
    \item All cost shares are nonnegative: $\vec{c}\ge 0$.

    \item For input coordinates that are $1$s, the output coordinates should
        sum up to $1$.  For example, if $n=3$ and $\vec{b}=(1,0,1)$ (the
        coalition is $\{1,3\}$), then $\vec{c}_1+\vec{c}_3=1$ (agent $1$ and
        $3$ are to share the total cost).

    \item For input coordinates that are $0$s, the output coordinates are
        irrelevant. We set these output coordinates to $1$s, which makes it
        more convenient for the next constraint.

    \item Every output coordinate is nondecreasing in every input coordinate.
        This is to ensure that the agents' cost shares are nondecreasing when
        some other agents are removed. If an agent is removed, then her cost
        share offer is kept at $1$, which makes it trivially nondecreasing.\\
\end{itemize}

All constraints except for the last is easy to achieve.
We will simply use $OUT(\vec{b})$ as output instead of directly using $NN(\vec{b})$\footnote{This is done by appending additional calculation structures to the output layer.}:
\[OUT(\vec{b})=\text{softmax}(NN(\vec{b})-1000(1-\vec{b}))+(1-\vec{b})\]\\

Here, $1000$ is an arbitrary large constant.
For example, let $\vec{b}=(1,0,1)$ and $\vec{c}=NN(\vec{b})=(x,y,z)$. We have
\[OUT(\vec{b})=\text{softmax}((x,y,z)-1000(0,1,0))+(0,1,0)\]
\[=\text{softmax}((x,y-1000,z))+(0,1,0)\]
\[=(x',0,z')+(0,1,0)=(x',1,y')\]

In the above, $\text{softmax}((x,y-1000,z))$ becomes $(x',0,y')$ with $x',y'\ge
0$ and $x'+y'=1$ because the second coordinate is very small so it
(essentially) vanishes after softmax. Softmax always produces nonnegtive
outputs that sum up to $1$.  Finally, the $0$s in the output are flipped to
$1$s per our third constraint.

The last constraint is enforced using a penalty function.
For $\vec{b}$ and $\vec{b}'$, where $\vec{b}'$ is obtained from $\vec{b}$ by changing one $1$ to $0$,
we should have that $OUT(\vec{b})\le OUT(\vec{b}')$, which leads to this penalty:
\[\text{ReLU}(OUT(\vec{b})-OUT(\vec{b}'))\]

Another way to enforce the last constraint is to use the
\emph{monotonic networks} structure~ (\cite{Sill1998:Monotonic}).
This idea has been used
in~ (\cite{Golowich2018:Deep}) , where the authors also dealt with networks that
take binary inputs and must be monotone.  However, we do not use this approach
because it is incompatible with our design for achieving the other constraints.
There are two other reasons for not using the monotonic network structure. One
is that it has only two layers. Some argue that having a \emph{deep} model is
important for performance in deep learning~ (\cite{Zhou2017:Deep}).  The other is
that under our approach, we only need a fully connected network with ReLU
penalty, which is highly optimized in state-of-the-art deep learning toolsets (while
the monotonic network structure is not efficiently supported by existing toolsets).
In our experiments, we use a fully connected network with four layers ($100$
nodes each layer) to represent our mechanism.\\

$$$$

\subsection{Cost Function}
For presentation purposes, we focus on maximizing the expected number of
consumers.  Only slight adjustments are needed for welfare maximization.\\

Previous approaches of mechanism design via neural networks used \emph{static}
networks~ (\cite{Golowich2018:Deep,
Duetting2019:Optimal,Shen2019:Automated,Manisha2018:Learning}). Given a sample,
the mechanism simulation is done on the network.  Our largest unanimous
mechanism involves iterative decision making, and the number of rounds
is not fixed, as it depends on the users' inputs.\\

To model iterative decision making via a static network, we could adopt the
following process.
The initial offers are $OUT((1,1,\ldots,1))$.  The remaining agents
after the first round are then $S=\text{sigmoid}(v-OUT((1,1,\ldots,1)))$.
Here, $v$ is the type profile sample. The sigmoid function turns positive
values to (approximately) $1$s and negative values to (approximately)  $0$s.
The next round of offers are then $OUT(S)$. The remaining agents afterwards are
then $\text{sigmoid}(v-OUT(S))$.  We repeat this $n$ times because the largest
unanimous mechanism have at most $n$ rounds.  The final coalition is a
converged state, so even if the mechanism terminates before the $n$-th round,
having it repeat $n$ times does not change the result (except for additional
numerical errors).  Once we have the final coalition $S^f$, we include
$\sum_{x\in S^f}x$ (number of consumers) in the cost function.
However, this approach
performs \emph{abysmally}, due to the \emph{vanishing gradient problem} and
numerical errors caused by stacking $n$ sigmoid functions.\\

We would like to avoid stacking sigmoid to model iterative decision making or
get rid of sigmoid altogether. Sigmoid is heavily used in existing works on neural network mechanism design, but it is the culprit of significant numerical errors. We propose an alternative approach, where
decisions are simulated off the network using a separate program (\emph{e.g.},
any Python function). The advantage of this approach is that it is now trivial
to handle complex decision making.  However, experienced neural network
practitioners may immediately notice a pitfall.  Given a type profile sample
$v$ and the current network $NN$, if we simulate the mechanism off the network
to obtain the number of consumers $x$, and include $x$ in the cost function,
then training will fail completely. This is because $x$ is not a differentiable function of network parameters and cannot support backpropagation at all.\footnote{We use PyTorch in our experiments. An
overview of Automated Differentiation in PyTorch is available here~ (\cite{Paszke2017:Automatic}).}\\

One way to resolve this is to interpret the mechanisms as price-oriented
rationing-free (PORF) mechanisms~ (\cite{Yokoo2003:Characterization}).  That is,
if we single out one agent, then her options (outcomes combined with payments)
are completely determined by the other agents and she simply has to choose the
utility-maximizing option.  Under a largest unanimous mechanism, an agent faces
only two results: either she belongs to the largest unanimous coalition or not.
If an agent is a consumer, then her payment is a constant due to
strategy-proofness, and the constant payment is determined by the other agents.
Instead of sampling over complete type profiles, we sample over $v_{-i}$ with a
random $i$.  To better convey our idea, we consider a specific example.  Let
$i=1$ and $v_{-1}=(\cdot, \frac{1}{2},\frac{1}{2},\frac{1}{4}, 0)$.  We assume
that the current state of the neural network is exactly the serial cost sharing
mechanism.  Given a sample, we use a separate program to calculate the
following entries. In our experiments, we simply used Python simulation to
obtain these entries.

\begin{itemize}

    \item The objective value if $i$ is a consumer ($O_s$). Under the example,
        if $1$ is a consumer, then the decision must be $4$ agents each pays
        $\frac{1}{4}$. So the objective value is $O_s=4$.

    \item The objective value if $i$ is not a consumer ($O_f$). Under the
        example, if $1$ is not a consumer, then the decision must be $2$ agents
        each pay $\frac{1}{2}$. So the objective value is $O_f=2$.

    \item The binary vector that characterizes the coalition that decides $i$'s
offer ($\vec{O_b}$). Under the example, $\vec{O_b}=(1,1,1,1,0)$. \\
\end{itemize}

$O_s$, $O_f$, and $\vec{O_b}$ are constants without network parameters. We link them together
using terms with network parameters, which is then included in the cost function:
\begin{equation}\label{eq:single1}
    (1-F(OUT(\vec{O_b})_i))O_s + F(OUT(\vec{O_b})_i)O_f
\end{equation}
\\
$1-F(OUT(\vec{O_b})_i)$ is the probability that agent $i$ accepts her offer.
$F(OUT(\vec{O_b})_i)$ is then the probability that agent $i$ rejects her offer.
$OUT(\vec{O_b})_i$ carries gradients as it is generated by the network.  We use the
analytical form of $F$, so the above term carries gradients.\footnote{PyTorch
has built-in analytical CDFs of many common distributions.}\\

The above approach essentially feeds the prior distribution into the cost
function. We also experimented with two other approaches. One does not use the
prior distribution. It uses a full profile sample and uses one layer of sigmoid to select
between $O_s$ or $O_f$:
\begin{equation}\label{eq:sigmoid}
    \text{sigmoid}(v_i-OUT(\vec{O_b})_i)O_s + \text{sigmoid}(OUT(\vec{O_b})_i-v_i))O_f
\end{equation}

The other approach is to feed ``even more'' distribution information into the cost
function.  We single out two agents $i$ and $j$. Now there are $4$ options:
they both win or both lose, only $i$ wins, and only $j$ wins. We still use $F$
to connect these options together.\\

In Section~\ref{Mechanism Design for Public Projects via Neural Network Experiments}, in one experiment, we show that singling out
one agent works the best. In another experiment, we show that even if we do not
have the analytical form of $F$, using an analytical approximation also enables
successful training.

\newpage
\subsection{Supervision as Initialization} We introduce an additional
supervision step in the beginning of the training process as a systematic way
of initialization.  We first train the neural network to mimic an existing
manual mechanism, and then leave it to gradient descent.  We considered three
different manual mechanisms. One is the serial cost sharing mechanism.  The
other two are based on two different heuristics:\\

\begin{definition}[One Directional Dynamic Programming] We make offers to the agents one by one. Every agent faces only one offer. The offer is based on how many agents are left, the objective value cumulated so far by the previous agents, and how much money still needs to be raised.  If an agent rejects an offer, then she is removed from the system.  At the end of the algorithm, we check whether we have collected $1$.  If so, the project is built and all agents not removed are consumers.  This mechanism belongs to
the largest unanimous mechanism family. This mechanism is not optimal because
we cannot go back and increase an agent's offer. \end{definition}

\begin{definition}[Myopic Mechanism] For coalition size $k$, we treat it as a
    nonexcludable public project problem with $k$ agents.  The offers are
    calculated based on the dynamic program proposed at the end of
    Subsection~\ref{sub:nonexcludable}, which computes the optimal offers for
    the nonexcludable model.  This is called the myopic mechanism, because it
    does not care about the payoffs generated in future rounds.  This mechanism
    is not necessarily feasible, because the agents' offers are not necessarily
nondecreasing when some agents are removed.  \end{definition}

%%%%%%%%%%%%%%%%%%%%%%%%%%%%%%%%%%%%%%%%%%%%%%%%%%%%%%%%%%%%%%%%%%%%%%%%

\newpage
\section{Experiments}\label{Mechanism Design for Public Projects via Neural Network Experiments}
The experiments are conducted on a
machine with Intel i5-8300H CPU.\footnote{We experimented with both PyTorch and
Tensorflow (eager mode).  The PyTorch version runs significantly faster,
because we are dealing with dynamic graphs.} The largest experiment
with $10$ agents takes about $3$ hours. Smaller scale experiments take only
about $15$ minutes.\\

In our experiments, unless otherwise specified, the distribution considered is
two-peak $(0.15,0.1,0.85,0.1,0.5)$.  The x-axis shows the number of training
rounds. Each round involves $5$ batches of $128$ samples ($640$ samples each
round).  Unless otherwise specified, the y-axis shows the expected number of
\textbf{non}consumers (so lower values represent better performances).  Random
initializations are based on Xavier normal with bias $0.1$.\\

Figure~\ref{fig:1} (Left) shows the performance comparison of three different
ways for constructing the cost function: using one layer of sigmoid
(without using distribution) based on~\eqref{eq:sigmoid}, singling out one agent based on~\eqref{eq:single1}, and singling out two agents.
All trials start from random initializations.  In this experiment, singling out
one agent works the best. The sigmoid-based approach is capable of moving the
parameters, but its result is noticeably worse. Singling out two agents has
almost identical performance to singling out one agent, but it is slower in
terms of time per training step.\\

\begin{figure}[H]

\centering
\includegraphics[width=0.85\textwidth]{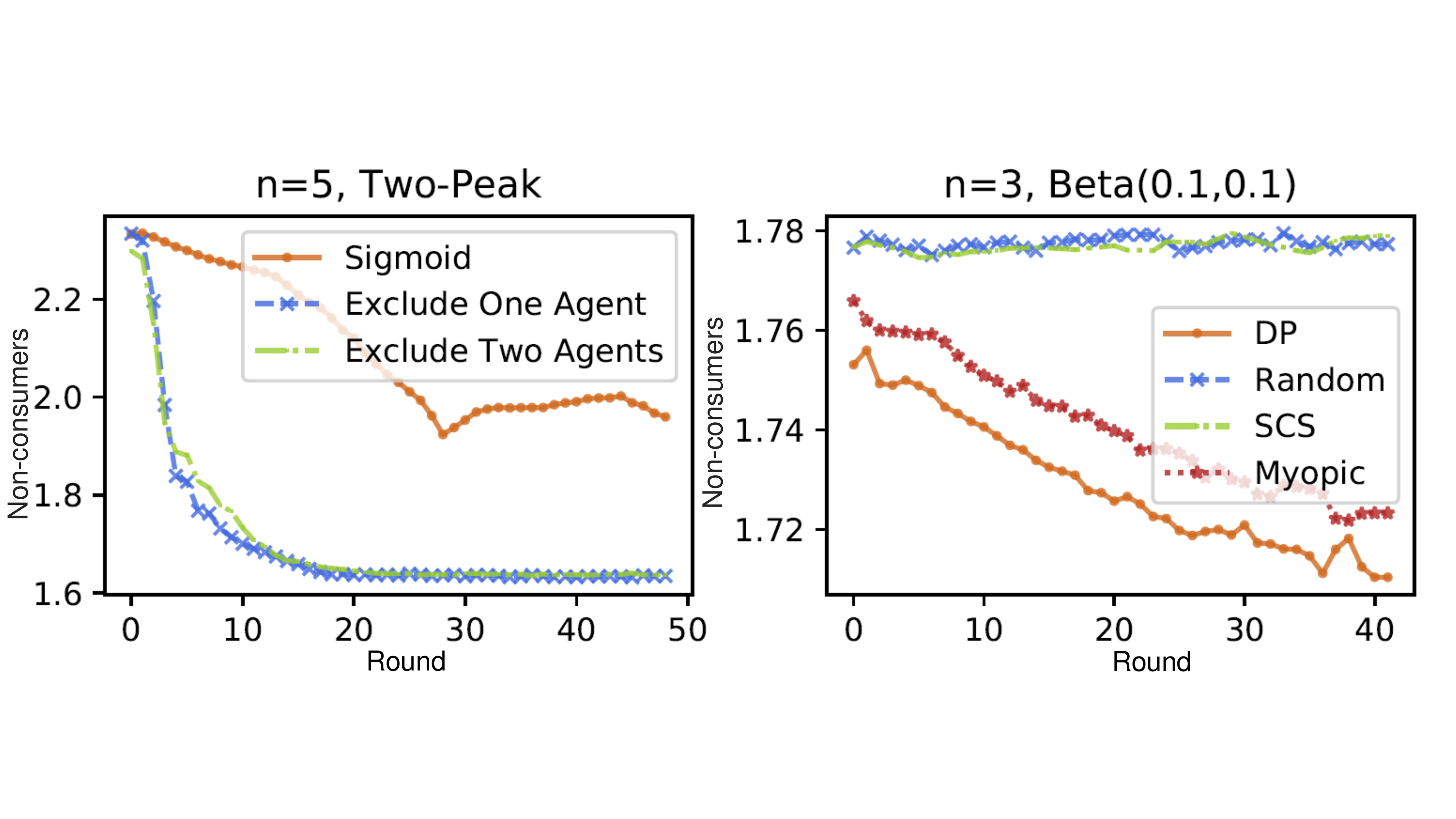}
    \caption[Experiment result: Effect of distribution info on training]{Effect of Distribution Info on Training (test set result during every training step)}
\label{fig:1}
\end{figure}

Figure~\ref{fig:1} (Right) considers the Beta $(0.1,0.1)$ distribution.  We use
Kumaraswamy $(0.1,0.354)$'s analytical CDF to approximate the CDF of Beta
$(0.1,0.1)$. The experiments show that if we start from random initializations
(Random) or start by supervision to serial cost sharing (SCS), then the cost
function gets stuck. Supervision to one directional dynamic programming (DP)
and Myoptic mechanism (Myopic) leads to better mechanisms. So in this example
scenario, approximating CDF is useful when analytical CDF is not available. It
also shows that supervision to manual mechanisms works better than random
initializations in this case.\\

\begin{figure}[H]

\centering
    \includegraphics[width=\textwidth]{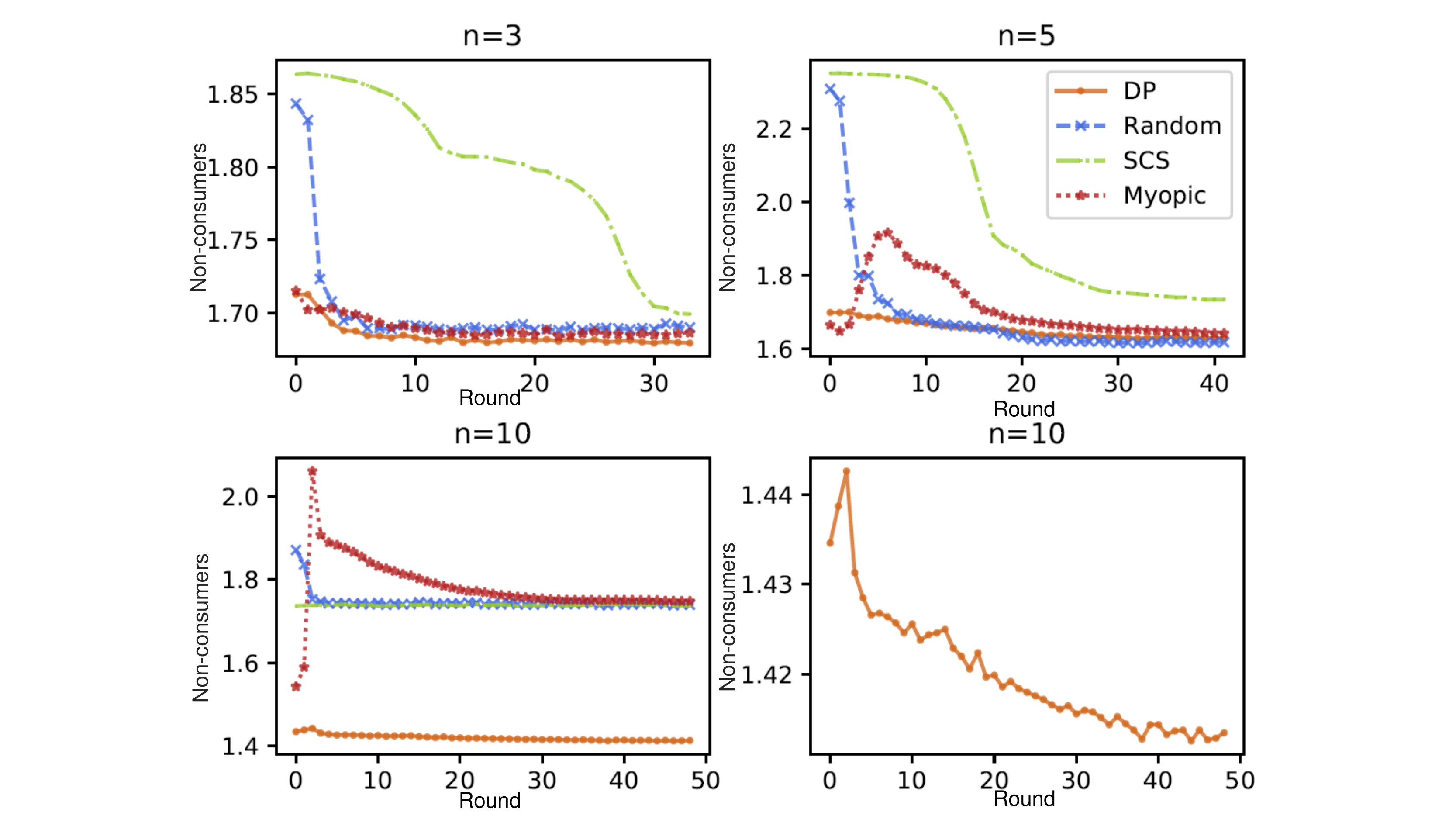}
        \caption[Experiment result: Supervision to different manual mechanisms]{Experiment result: Supervision to Different Manual Mechanisms (test set result during every training step)}
\label{fig:2}
\end{figure}

Figure~\ref{fig:2} (Top-Left $n=3$, Top-Right $n=5$, Bottom-Left $n=10$) shows
the performance comparison of supervision to different manual mechanisms.  For
$n=3$, supervision to DP performs the best. Random initializations is able to
catch up but not completely close the gap. For $n=5$, random initializations
caught up and actually became the best performing one. The Myopic curve first
increases and then decreases because it needs to first fix the constraint
violations. For $n=10$, supervision to DP significantly outperforms the others. Random
initializations closes the gap with regard to serial cost sharing, but it then
gets stuck. Even though it looks like the DP curve is flat, it is actually improving, albeit very slowly. A magnified version is shown in Figure~\ref{fig:2} (Bottom-Right).\\

\begin{figure}[H]

\centering
    \includegraphics[width=0.9\textwidth]{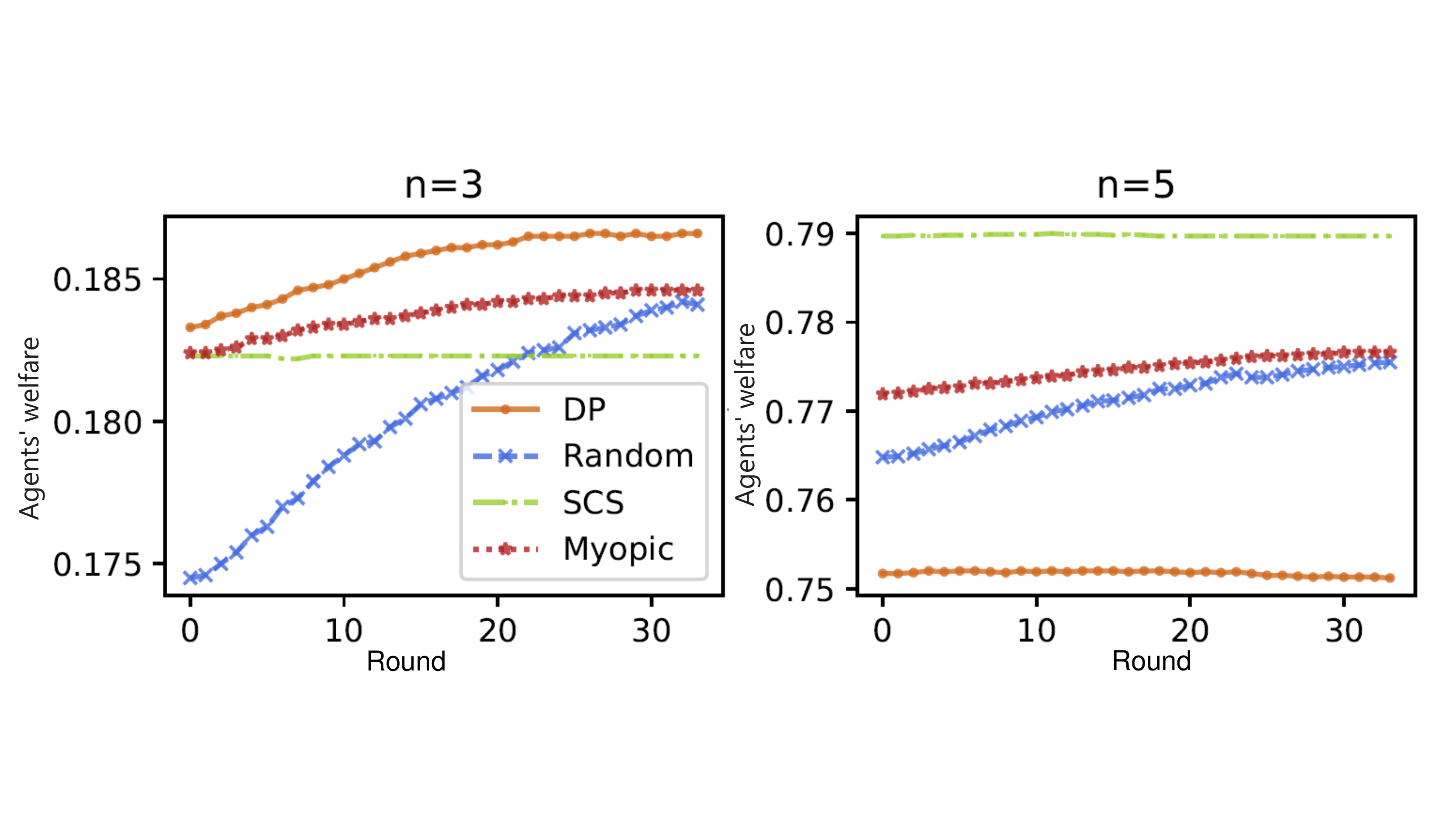}
    
        \caption[Experiment result: Maximizing agents' welfare]{Experiment result: Maximizing Agents' Welfare (test set result during every training step)}
\label{fig:3}
\end{figure}

Figure~\ref{fig:3} shows two experiments on maximizing expected agents' welfare (y-axis) under
two-peak $(0.2,0.1,0.6,0.1,0.5)$.  For $n=3$, supervision to DP leads to the
best result. For $n=5$, SCS is actually the best mechanism we can find (the
cost function barely moves).\\

It should be noted that all manual mechanisms \emph{before training} have very similar welfares:
$0.7517$ (DP), $0.7897$ (SCS), $0.7719$ (Myopic). Even random initialization
before training has a welfare of $0.7648$.  In this case, there is just little room
for improvement.

%%%%%%%%%%%%%%%%%%%%%%%%%%%%%%%%%%%%%%%%%%%%%%%%%%%%%%%%%%%%%%%%%%%%%%%%

\newpage
\section{Chapter Summary}
In this chapter, we studied optimal-in-expectation mechanism design for the public project model. We are the first to use neural networks to design iterative mechanisms. To avoid modeling iterative decision making via the sigmoid function, we simulate the different options an agent faces under an iterative mechanism and combine the options using distribution information to build the cost function for our optimal-in-expectation objective. We showed that under various conditions, existing mechanisms or mechanisms derived via classic mechanism design approaches are optimal. When classic mechanism design approaches do not suffice, we derived better mechanisms via neural networks by first training the neural networks to mimic manual mechanisms and then improving by gradient descent.

%% file: Chapters/Chapter4.tex
\chapter{Public Project with Minimum Expected Release Delay} % Main chapter title

\label{Public Project with Minimum Expected Release Delay} % Change X to a consecutive number; for referencing this chapter elsewhere, use \ref{ChapterX}

In this chapter, we study a public project model where the project can be released to different agents at different times. The mechanism designer uses ``delay in release'' to incentivize the agents to pay for the public project. The goal of the mechanism design is to minimize the maximum delay or the total delay.  \\

\section{Introduction}

The public project problem is a fundamental mechanism design model with many applications in multiagent systems.  The public project problem involves multiple agents, who need to decide whether or not to build a public project.
The project can be {\bf nonexcludable} ({\em i.e.}, if the project is built, then every agent gets to consume the project, including the non-paying agents/free riders) or
{\bf excludable} ({\em i.e.}, the setting makes it possible to exclude some
agents from consuming the project)~ (\cite{Ohseto2000:Characterizations}).\footnote{An
example nonexcludable public project is a public airport, and an example
excludable public project is a gated swimming pool.} A public project can be
{\bf indivisible/binary} or {\bf
divisible}~ (\cite{Moulin1994:Serial}).  A binary public project is either
built or not built ({\em i.e.}, there is only one level of provision).
In a divisible public project, there are multiple
levels of provision ({\em i.e.}, build a project with adjustable quality).\\

In this chapter, we study an excludable public project model that is
``divisible'' in a different sense. In the model, the level of provision is
binary ({\em i.e.}, the project is either built or not built), but an agent's
consumption is divisible. The mechanism specifies when an agent can start
consuming the project. High-paying agents can consume the project earlier, and
the free riders need to wait.  The waiting time is also called an agent's {\bf
delay}. The delay is there to incentivize payments.  The model was proposed by
Guo {\em et al.}~ (\cite{Guo2018:Cost}).  The authors studied the following
mechanism design scenario. A group of agents come together to crowd-fund a
piece of security information. No agent is able to afford the information by
herself.\footnote{Zero-day exploits are very
expensive~ (\cite{Greenberg2012:Shopping,Fisher2015:VUPEN}).} Based on the agents'
valuations on the information, the mechanism decides whether or not to
crowd-fund this piece of information ({\em i.e.}, purchase it from the security
consulting firm that is selling this piece of information).  If we are able to
raise enough payments to cover the cost of the security information, then {\em
ideally} we would like to share it to all agents, including the free riders, in
order to maximizes the {\em overall protection of the community}.  However, if all agents
receive the information regardless of their payments, then no agents are
incentivized to pay. To address this, the mechanism releases the information
only to high-paying agents in the beginning and the non-paying/low-paying
agents need to wait for a delayed release.  The mechanism design goal is to
minimize the delay as long as the delay is long enough to incentivize enough
payments to cover the cost of the information.  Guo {\em et
al.}~ (\cite{Guo2018:Cost}) proposed two design objectives. One is to minimize the
{\em max-delay} ({\em i.e.}, the maximum waiting time of the agents) and the
other is to minimize the {\em sum-delay} ({\em i.e.}, the total waiting time of
the agents). The authors focused on {\em worst-case mechanism design} and
proposed a mechanism that has a constant approximation ratio compared to the
optimal mechanism.  The authors also briefly touched upon {\em expected} delay.
The authors used simulation to show that compared to their worst-case
competitive mechanism, the {\em serial cost sharing mechanism} proposed by
Moulin~ (\cite{Moulin1994:Serial}) has much lower expected {\em max-delay} and
{\em sum-delay} under various distributions.\\

In this chapter, we focus on minimizing the expected {\em max-delay} and the
expected {\em sum-delay}. We propose a mechanism family called the {\bf single
deadline mechanisms}.  For both objectives, under minor technical assumptions,
we prove that there exists a single deadline mechanism that is {\em near
optimal} when the number of agents is large, {\em regardless of the prior
distribution}.  Furthermore, when the number of agents approaches infinity, the
optimal single deadline mechanism approaches optimality asymptotically.  For
small number of agents, the single deadline mechanism is not optimal. We extend
the single deadline mechanisms to multiple deadline mechanisms. We also propose
a genetic algorithm based automated mechanism design approach.
We use a sequence of offers to represent a mechanism and we evolve the sequences.
By simulating
mechanisms using multiple distributions, we show that our genetic algorithm
successfully identifies better performing mechanisms for small number of agents.\\

Ohseto~ (\cite{Ohseto2000:Characterizations}) characterized
all strategy-proof and individually rational mechanisms for the binary public
project model (both excludable and nonexcludable), under minor technical
assumptions.
% For the nonexcludable model, the only strategy-proof and
% individually rational mechanisms are the {\em unanimous mechanisms}. For the
% excludable model, the only strategy-proof and individually rational mechanisms
% are the {\em largest unanimous mechanisms}.
Deb and
Razzolini~ (\cite{Deb1999:Voluntary}) further showed that on top of Ohseto's
characterization, if we require {\em equal treatment of equals} ({\em i.e.}, if
two agents have the same type, then they should be treated the same), then the
only strategy-proof and individually rational mechanisms are the {\em conservative equal cost mechanism} (nonexcludable) and the
{\em serial cost sharing mechanism} (excludable), which were both proposed by
Moulin~ (\cite{Moulin1994:Serial}). It should be noted that Ohseto's
characterization involves {\em exponential} number of parameters, so
knowing the characterization does not mean it is easy to locate
good mechanisms.  Wang {\em et al.}~ (\cite{Wang2021:Mechanism})
proposed a neural network based approach for optimizing within Ohseto's characterization
family.\\
% the largest
% unanimous mechanism family.

The authors studied two objectives: maximizing the
 number of consumers and maximizing the social welfare.
 It should be noted that Ohseto's characterization does not apply to the model  in this chapter, as our model has an additional spin that is the release delay.
 In this chapter, we propose a family of mechanisms called the sequential
 unanimous mechanisms, which is motivated by Ohseto's characterization.
% Guo {\em et al.}~\cite{Guo2018:Cost} proposed a modified version of the binary
% excludable public project model.
% In the original binary excludable public project model, an agent either consumes
% the project or is excluded. In Guo {\em et al.}'s model, the mechanism sets
% a premium period. High-paying agents can consume the premium period, while the
% low-paying/no-paying agents need to wait.
% This paper focuses on this modified model.
% The formal model
% description will be presented in Section~\ref{sec:model}.
% Guo {\em et al.} proposed a mechanism called {\em group-based cost
% sharing with optimal deadline} (GCSOD). The authors showed that GCSOD is
% competitive against the optimal mechanism in the worst case. The authors
% also showed that the GCSOD mechanism does not perform well in expectation.
% The GCSOD mechanism's expected performance is worse than the serial cost
% sharing mechanism proposed by Moulin~\cite{Moulin1994:Serial}.
% Finally, the last part of this paper introduces an evolutionary computational
% automated mechanism design approach, which successfully identifies better
% mechanisms than our manually designed single deadline mechanism.
 We apply a genetic algorithm for tuning the sequential unanimous mechanisms.
Mechanism design via evolutionary computation~ (\cite{phelps2010evolutionary}) and
mechanism design via other computational means (such as linear
programming~ (\cite{Conitzer2002:Complexity}) and neural
networks~ (\cite{Duetting2019:Optimal,Shen2019:Automated,Wang2021:Mechanism}))
have long been shown to be effective for many design settings.\\

% \subsection{Serial Cost Sharing(SCS)}
% Introduction
% The typical traditional methods is serial cost sharing(SCS)\cite{Moulin1994:Serial,Guo2018:Cost}.
% For every nonempty subset of  agents $S$ with $|S|=k$, the cost share vector is    $(\frac{1}{k},\frac{1}{k},\ldots,\frac{1}{k})$.  The mechanism picks the
%     largest coalition where the agents are willing to pay equal shares. \\
% \begin{itemize}
% \item If S is empty, then the good will not build, every agent's can't use it ( release date $t_i$ = 0) and pays 0( $p_i = 0$).
% \item If S is not empty, then the highest  $k^*$ = max $k$  agents each pays 1/$k^*$ ( $p_i = 1/k^*$).and the other can use the goods immediately (release date $t_i$ = 1). The other agents use the goods at time 1 (release date $t_i$ = 0)and each pay 0( $p_i = 0$).
% \end{itemize}

% Moulin (1994)\cite{Moulin1994:Serial} considers truthful mechanisms for excludable public good problems.

$$$$
\section{Model Description}\label{sec:model}

There are $n$ agents who decide whether or not to build a public project. The
project is binary (build or not build) and nonrivalrous (the cost of the
project does not depend on how many agents are consuming it). We normalize the
project cost to $1$. Agent $i$'s type $v_i \in [0, 1]$ represents her private
valuation for the project. We use $\vec{v}=(v_1,v_2,\ldots,v_n)$ to denote the type profile. We assume that the $v_i$ are drawn {\em
i.i.d.} from a known prior distribution, where $f$ is the probability
density function.  For technical reasons, we assume
$f$ is {\em positive} and {\em Lipschitz continuous over $[0,1]$}.\\

We assume that the public project has value over a time period $[0,1]$.
For example, the project could be a piece of security information that is discovered
at time $0$ and the corresponding exploit expires at time $1$.  We
assume the setting allows the mechanism to specify each agent's release time for the project, so that some agents can consume the project earlier than the others.  Given a type
profile, a mechanism outcome consists of two vectors: $(t_1,t_2,\ldots,t_n)$
and $(p_1,p_2,\ldots,p_n)$. {\em I.e.}, agent $i$ starts consuming the project
at time $t_i \in [0,1]$ and pays $p_i\ge 0$.  $t_i=0$ means agent $i$ gets to consume
the public project right from the beginning and $t_i=1$ means agent $i$ does not
get to consume the public project.
We call $t_i$ agent $i$'s {\em release time}.
We assume the agents' valuations over the time
period is uniform. That is, agent $i$'s valuation equals $v_i(1-t_i)$, as she
enjoys the time interval $[t_i,1]$, which has length $1-t_i$. Agent $i$'s utility
is then $v_i(1-t_i)-p_i$.
We impose the following mechanism design constraints:

\newpage
\begin{itemize}

    \item Strategy-proofness: We use $t_i$ and $p_i$ to denote agent $i$'s release
        time and payment when she reports her true value $v_i$. We use $t_i'$ and $p_i'$ to denote agent $i$'s release
        time and payment when she reports a false value $v_i'$. We should have
        \[v_i(1-t_i)-p_i \ge v_i(1-t_i')-p_i'\]

    \item Individual rationality:\\ $$v_i(1-t_i)-p_i\ge 0$$\\
    \item Ex post budget balance:

        {\em If the project is not built}, then no agent
        can consume the project and no agent pays. That is, we must have $t_i=1$ and $p_i=0$ for all $i$.

        {\em If the project is built}, then the agents' total payment must cover
        exactly the project cost. That is, $\sum_i p_i=1$.\\

\end{itemize}
Our aim is to design mechanisms that minimize the following design objectives:

\begin{enumerate}
    \item Expected {\em Max-Delay}: $E_{v_i\sim f}\left(\max\{t_1,t_2,\ldots,t_n \}\right)$
    \item Expected {\em Sum-Delay}: $E_{v_i\sim f}\left(\sum_{i}t_i\right)$
\end{enumerate}

\newpage
\section{Single Deadline Mechanism}

We first describe the {\em serial cost sharing mechanism (SCS)} proposed by Moulin~ (\cite{Moulin1994:Serial}). Under SCS, an agent's release
time is either $0$ or $1$.\footnote{Because the concept of release time does not exist
in the classic binary excludable public project model.}\\

Let $\vec{v}$ be the type profile.
We first define the following functions:\\

    \[I(\vec{v})=\begin{cases}
        1 & \exists k \in \{1,2,\ldots,n\}, k \le |\{v_i|v_i\ge \frac{1}{k}\}| \\
        0 & \text{otherwise}
   \end{cases}
    \]
\\

$I(\vec{v})$ equals $1$ if and only if there exist at least $k$ values among $\vec{v}$
that are at least $\frac{1}{k}$, where $k$ is an integer from $1$ to $n$.\\

    \[K(\vec{v})=\begin{cases}
        \max\{k| k \le |\{v_i|v_i\ge \frac{1}{k}\}|,k\in \{1,2,\ldots,n\}\} & I(\vec{v})=1\\
        0 & I(\vec{v})=0
   \end{cases}
    \]
\\

Given $\vec{v}$, there could be multiple values for $k$, where there exist at
least $k$ values among $\vec{v}$ that are at least $\frac{1}{k}$. $K(\vec{v})$
is the largest value for $k$. If such a $k$ value does not exist, then
$K(\vec{v})$ is set to $0$.\\

\begin{definition}[Serial Cost Sharing Mechanism~ (\cite{Moulin1994:Serial})]
    Let $\vec{v}$ be the type profile.
    Let $k=K(\vec{v})$.

    \begin{itemize}

        \item
    If $k>0$, then agents with the highest $k$ values are the consumers.
            The consumers pay $\frac{1}{k}$. The non-consumers do not pay.

        \item
    If $k=0$, then there are no consumers and no agents pay.\\
    \end{itemize}
\end{definition}

Essentially, the serial cost sharing mechanism finds the largest $k$ where $k$
    agents are willing to equally split the cost. If such a $k$ exists, then we say {\em the cost share is successful} and these $k$ agents are {\em joining the cost share}.
    If such a $k$ does not exist, then we say {\em the cost share failed}.

Next we introduce a new mechanism family called the single deadline mechanisms.\\

\begin{definition}[Single Deadline Mechanisms]

A single deadline mechanism is characterized by one parameter $d \in [0,1]$.
    $d$ is called the mechanism's {\bf deadline}.
    We use $M(d)$ to denote the single deadline mechanism with deadline $d$.
    The time interval before the deadline $[0,d]$ is called the {\bf non-free}
    part. The time interval after the deadline $[d,1]$ is called the {\bf free}
    part.

    We run the serial cost sharing mechanism on the non-free part as follows.  For the
    non-free part, the agents' valuations are
    $d\vec{v}=(dv_1,dv_2,\ldots,dv_n)$.
    % $d\vec{v}$ is the product between a real value $d$ and a vector $\vec{v}$.
    Let $k=K(d\vec{v})$.
    Agents with the highest $k$ values get to consume the non-free part, and
    they each needs to pay $\frac{1}{k}$.

    The free part is allocated to the agents for free. However, we cannot
    give out the free part if the public project is not built.
    % That is, we give out the free part only if $I(d\vec{v})=1$.

    If we give out the free part if and only if $I(d\vec{v})=1$,
    then the mechanism is not strategy-proof, because the free parts change the
    agents' strategies.\footnote{For example, an agent may over-report to turn an
    unsuccessful cost share into a successful cost share, in order to claim the
    free part.}
    Instead, we give agent $i$ her free part if and only if $I(dv_{-i})=1$.
    That is, agent $i$ gets her free part if and only if the other agents
    can successfully cost share the non-free part without $i$.

    If an agent receives both the non-free part and the free part, then her release
    time is $0$. If an agent only receives the free part, then her release time
    is $d$. If an agent does not receive either part, then her release time is $1$.
    Lastly, if an agent only receives the non-free part, then her release time
    is $1-d$, because such an agent's consumption interval should have length $d$ ({\em i.e.}, $[1-d,1]$).\\
\end{definition}

\begin{proposition}
    The single deadline mechanisms are strategy-proof, individually rational, and
    ex post budget balanced.
\end{proposition}

\begin{proof}
    Whether an agent receives her free part or not does not depend on her report,
    so the agents are essentially just facing a serial cost sharing mechanism, where
    the item being cost shared is $d$ portion of the public project.
    % The serial cost
    % sharing mechanism is strategy-proof and individually rational.
    % If $I(d\vec{v})=0$, then an agent receives neither the free nor the non-free part.
    % Every agent's release time equals $1$ and every agent pays $0$.
    % If $I(d\vec{v})=1$, then the total payment is exactly $1$. Therefore, the
    % single deadline mechanism is ex post budget balanced.
$\blacksquare$
\end{proof}

$$$$
\section{Max-Delay: Asymptotic Optimality}

In this section, we show that when the number of agents is large enough, the optimal single deadline mechanism is asymptotically optimal in terms of expected max-delay, regardless of the prior distribution.\\

\begin{theorem}\label{thm:maxdelay}
    The optimal single deadline mechanism's expected max-delay approaches $0$
    when the number of agents approaches infinity.
\end{theorem}

\begin{proof}
We consider a single deadline mechanism $M(d)$.  Every agent's valuation is drawn {\em i.i.d.} from a distribution with PDF $f$.
Let $V_i$ be the random variable representing agent $i$'s valuation.
Since $f$ is positive and Lipschitz continuous, we have the following:

\[\forall d, \exists k, P(dV_i\ge \frac{1}{k}) >0 \]
\\

That is, for any deadline $d$, there always exists an integer $k$, where the
probability that an agent is willing to pay $\frac{1}{k}$ for the non-free part is positive.
Let $p=P(dV_i\ge \frac{1}{k})$.
We define the following Bernoulli random variable:
    \[B_i=\begin{cases}
        1 & dV_i\ge \frac{1}{k}\\
        0 & \text{otherwise}
   \end{cases}
    \]
\\

$B_i$ equals $1$ with probability $p$. It equals $1$ if and only if agent $i$
can afford $\frac{1}{k}$ for the non-free part.
The total number of agents in $\vec{v}$ who can afford
$\frac{1}{k}$ for the non-free part then follows a Binomial distribution
$B(n,p)$. We use $B$ to denote this Binomial variable.
If $B\ge k+1$, then every agent receives the free part, because
agent $i$ receives the free part if excluding herself, there are at least $k$
agents who are willing to pay $\frac{1}{k}$ for the non-free part.
The probability that the max-delay is higher than $d$ is therefore bounded above
by $P(B\le k)$.
According to Hoeffding's inequality, when $k<np$,

\[P(B\le k)\le e^{-2n\left(p-\frac{k}{n}\right)^2}\]
\\
We immediately have that when $n$ approaches infinity, the probability
that the max-delay is higher than $d$ is approaching $0$. Since $d$
is arbitrary, we have that asymptotically, the single deadline mechanism's
expected max-delay is approaching $0$.
$\blacksquare$\\
\end{proof}
$$$$

Next, we use an example to show that when $n=500$,
the optimal single deadline mechanism's expected max-delay is close to $0.01$.
We reuse all notation defined in the proof of Theorem~\ref{thm:maxdelay}.
We make use of the Chernoff bound. When $k<np$, we have

\[P(B\le k)\le e^{-nD\left(\frac{k}{n} ||p \right)},\quad \text{where}\, D\left(a||p\right)=a\ln\frac{a}{p}+(1-a)\ln\frac{1-a}{1-p}\]
\\
When all agents receive the free part, the max-delay is at most $d$.
Otherwise, the max-delay is at most $1$.
The expected max-delay is at most
\[P(B\le k) + d(1 - P(B\le k)) \le P(B\le k)+d\]
\newpage
\begin{example}\label{ex:max}
    Let us consider a case where $n=500$. We set $d=0.01$ and $k=250$.

    \begin{itemize}

        \item $f$ is the uniform distribution $U(0,1)$:
            We have $p=0.6$ and $P(B\le 250)\le 3.69\mathrm{e}-5$.
            $M(0.01)$'s expected max-delay is then bounded above by $0.01 + 3.69\mathrm{e}-5$.

        \item $f$ is the normal distribution $N(0.5,0.1)$ restricted to $[0,1]$:
            We have $p=0.84$ and $P(B\le 250)\le 7.45\mathrm{e}-69$.
            $M(0.01)$'s expected max-delay is then bounded above by $0.01 + 7.45\mathrm{e}-69$.\\

    \end{itemize}

\end{example}

On the contrary, the expected max-delay of the serial cost sharing mechanism is not approaching
$0$ asymptotically. For example, when $n=500$, under the uniform distribution $U(0,1)$,
the expected max-delay of the serial cost sharing mechanism equals $0.632$ which is very close to $1-\frac{1}{e}$.\\

\begin{proposition}
    The expected max-delay of the serial cost sharing mechanism equals
    \[1-(\int_{\frac{1}{n}}^1f(x)dx)^n\]

    The above expression approaches $1-e^{-f(0)}$ asymptotically.
\end{proposition}

\newpage
\section{Sum-Delay: Asymptotic Optimality}

In this section, we show that when the number of agents is large enough, the
optimal single deadline mechanism's expected sum-delay approaches optimality,
regardless of the prior distribution.\\

\begin{theorem}\label{thm:sumdelay}
    When the number of agents approaches infinity, the optimal single deadline mechanism is optimal among all mechanisms in terms of expected sum-delay.
    % That is,
    % for any $\epsilon>0$, there exists $n_0$ so that when the number
    % of agents is above $n_0$, there exists a single deadline mechanism
    % whose expected sum-delay is within $\epsilon$ of the optimal value.
\end{theorem}
Theorem~\ref{thm:sumdelay} can be proved by combining
Proposition~\ref{prop:lowerboundsum} and
Proposition~\ref{prop:achievesumdelay}.\\

\begin{proposition}\label{prop:finite}
    The optimal expected sum-delay is finite regardless of the distribution.
\end{proposition}

\begin{proof}
    We consider the following mechanism: Pick an arbitrary integer $k>1$. We offer $\frac{1}{k}$ to the agents one by one. An agent gets the whole interval $[0,1]$ if she agrees
    to pay $\frac{1}{k}$ and if the project is built. Otherwise, she gets nothing. We build the project only when $k$ agents agree. Since we approach the agents one by one, after $k$ agents agree to pay $\frac{1}{k}$, all future agents receive the whole interval for free. This mechanism's expected sum-delay is bounded above by a constant. The constant
    only depends on the distribution.
    $\blacksquare$\\
\end{proof}

The following proposition follows from Proposition~\ref{prop:finite}.\\

\begin{proposition}\label{prop:fail}
    Given a mechanism $M$ and the number of agents $n$, let $Fail(n)$ be the probability of not building under $M$. We only need to consider $M$ that satisfies $Fail(n)= O(1/n)$.\\
\end{proposition}

% \begin{definition}[Single deadline with single offer (SDSO) mechanisms]
%     A single deadline with single offer (SDSO) mechanism, denoted by $M(d,o)$, is characterized
%     by two parameters: the deadline $d$ and an offer $o$ for the whole interval $[0,1]$.

%     Under $M(d,o)$, an agent receives the non-free part $[0,d]$ if she is willing
%     to pay $od$. An agent always receives the free part $[d,1]$.
% \end{definition}

% It should be noted that the SDSO mechanisms are not always feasible. A SDSO mechanism is strategy-proof, individually rational, but may not be ex post budget balanced. We only use SDSO in the proofs.

We then propose a relaxed version of the ex post budget balance constraint, and
use it to calculate the delay lower bound.
\newpage
\begin{definition}[Ex ante budget balance]
    Mechanism $M$ is ex ante budget balanced if and only if the expected total payment from the agents equals the probability of building (times project cost $1$).\\
\end{definition}

\begin{proposition}\label{prop:lowerboundsum}
    Let $Fail(n)$ be the probability of not building the project under the optimal mechanism when there are $n$ agents.

    We consider what happens when we offer $o$ for the whole interval $[0,1]$ to an individual
    agent. If the agent accepts $o$ then she pays $o$ and gets the whole interval. Otherwise,
    the agent pays $0$ and receives nothing.

    We define the {\em delay versus payment ratio} $r(o)$ as follows:

    \[r(o) = \frac{\int_{0}^{o}f(x)dx}{o\int_o^1f(x)dx}\]

    $r$ is continuous on $(0,1)$. Due to $f$ being Lipschitz continuous, we have $\lim_{o\rightarrow 0}r(o)=f(0)$
    and $\lim_{o\rightarrow 1}r(o)=\infty$.
    We could simply set $r(0)=f(0)$, then $r$ is continuous on $[0,1)$.

    We define the {\em optimal delay versus payment ratio} $r^*$ as follows:

    % \[o^* = \arg\min_{o\in [0,1]} r(o)\]
    \[r^* = \min_{o\in [0,1)}r(o)\]

    The expected sum-delay is bounded below by $r^*(1-Fail(n))$, which approaches
    $r^*$ asymptotically according to Proposition~\ref{prop:fail}.
\\

    % The above lower bound is calculated by relaxing the {\em ex post budget
    % balance} constraint by the {\em ex ante budget balance constraint}.  Under the
    % ex ante budget balance constraint, the optimal mechanism is

    % \[M(\frac{1-Fail(n)}{no^*\int_{o^*}^1f(x)dx}, o^*)\]

    % If we swap the ex post budget balance constraint by the ex ante budget balance constraint, then for any $\epsilon$ and any distribution, there exists $n_0$ so that when $n>n_0$, there exists a SDSO mechanism $M(d^*, o^*)$ that is within $\epsilon$ from optimality in terms
    % of expected sum-delay. We also have that $d^*<1$.
\end{proposition}

\begin{proof}[Outline]
    % We first show that $r$ is continuous.
    % \[r(o) = \frac{\int_{0}^{o}f(x)dx}{(o\int_o^1f(x)dx)}\]
    % \[r'(o) = \frac{f(o)(o\int_o^1 f(x)dx) -( \int_0^o f(x)dx)(\int_o^1 f(x)dx -of(o))}{(o\int_o^1f(x)dx)^2}\]
    % \[r'(o) = \frac{f(o)o -( \int_0^o f(x)dx)(\int_o^1 f(x)dx)}{(o\int_o^1f(x)dx)^2}\]
    % \[r'(o) = \frac{f(o)}{o(\int_o^1f(x)dx)^2}- \frac{\int_0^o f(x)dx}{o^2(\int_o^1f(x)dx)}\]
    If we switch to ex ante budget balance, then it is without loss of generality to
    focus on anonymous mechanisms. We then face a single agent mechanism design problem
    where an agent pays $\frac{1-Fail(n)}{n}$ in expectation and we want to minimize
    her expected delay. Based on Myerson's characterization for single-parameter settings, here every strategy-proof mechanism works as follows: for each infinitesimal time interval
    there is a price and the price increases as an agent's allocated interval increases in length. There is an optimal price that minimizes the ratio between the delay caused by the price and the payment.
    The total payment is $1-Fail(n)$, which means the total delay is at least $r^*(1-Fail(n))$.
    % An optimal mechanism simply uses this price and sets
    % the non-free part's length so that an agent's expected payment equals $\frac{1-Fail(n)}{n}$.
    $\blacksquare$
\end{proof}

\begin{proposition}\label{prop:achievesumdelay}
    Let $o^*$ be the optimal offer that leads to the optimal delay versus payment ratio
    $r^*$.\footnote{If $o^*=0$, then we replace it with an infinitesimally small $\gamma>0$. The achieved
    sum-delay is then approaching $r(\gamma)(1+\epsilon)$ asymptotically. When $\gamma$ approaches
    $0$, $r(\gamma)$ approaches $r^*$.}
    \[o^* = \arg\min_{o\in [0,1)} r(o)\]

    Let $\epsilon >0$ be an arbitrarily small constant.
    The following single deadline mechanism's expected sum delay approaches $r^*(1+\epsilon)$ asymptotically.

    \[M(\frac{1+\epsilon}{no^*\int_{o^*}^1f(x)dx})\]

\end{proposition}

\begin{proof}
    % Let $B_i$ be the following Bernoulli random variable:
    % \[B_i=\begin{cases}
    %     1 & (1+\epsilon)V_i \ge o^*\\
    %     0 & \text{otherwise}
   % \end{cases}
    % \]
Let $p=P(V_i(1+\epsilon)\ge o^*)$.
Let $k=n\int_{o^*}^1f(x)dx$.
$p$ is the probability that an agent is willing to pay $\frac{1}{k}$
for the non-free part whose length is $\frac{1+\epsilon}{no^*\int_{o^*}^1f(x)dx}$.
We use $B$ to denote the Binomial distribution $B(n,p)$.
If $B> k$, then every agent receives the free part, because
cost sharing is successful even if we remove one agent.
The probability that an agent does not receive the free part is then bounded above
by $P(B\le k)$.
According to Hoeffding's inequality, we have that when $k<np$, we have

\[P(B\le k)\le e^{-2n(p-\frac{k}{n})^2} =
e^{-2n(\int_{\frac{o^*}{1+\epsilon}}^1f(x)dx-\int_{o^*}^1f(x)dx)^2}
=e^{-2n(\int_{\frac{o^*}{1+\epsilon}}^{o^*}f(x)dx)^2}\]

Let $\beta=\int_{\frac{o^*}{1+\epsilon}}^{o^*}f(x)dx$.
The expected total delay when some agents do not receive the free part is then at most $ne^{-2n\beta^2}$, which approaches $0$ as $n$ goes to infinity.
Therefore, we only need to consider situations where all agents receive the free part and at least $k$ agents receive the non-free part.
The expected sum delay on the remaining $n-k$ agents is then at most

\[(n-k)\frac{1+\epsilon}{no^*\int_{o^*}^1f(x)dx}= (1+\epsilon)\frac{\int_0^{o^*}f(x)dx}{o^*\int_{o^*}^1f(x)dx}= (1+\epsilon)r^*\]
$\blacksquare$
\end{proof}

% \begin{proof}[Theorem~\ref{thm:sumdelay}]
    % The proof is based on the above propositions.
% When $n$ goes to infinity, the probability of building under the optimal mechanism goes to $1$.
    % If we relax the ex post budget balance constraint to the ex ante budget balance
    % constraint, then every agent's expected payment is $\frac{1}{n}$.\footnote{After
    % switching to ex ante budget balance, it is without loss of generality to focus
    % only on anonymous mechanisms, because with ex ante budget balance, we can take
    % the ``average'' of two mechanisms.}
    % The optimal mechanism is then to offer $o^*$ to every agent for the interval $[0,d^*]$ and give out the interval $[d^*,1]$ for free.
    % We could slightly reduce $d^*$ to $d^*-\epsilon$ so that the agents' expected payment becomes strictly above $\frac{1}{n}$.
    % When $n$ is large enough and when $\epsilon$ is small enough, $M(d^*-\epsilon)$ performs nearly the same as $M(d^*, o^*)$.
% \end{proof}

\newpage
We then use an example to show that when $n=500$, under different distributions,
the optimal single deadline mechanism's expected sum-delay is close to the optimal
value.

\begin{example}
    We consider $n=500$ which is the same as Example~\ref{ex:max}.
    Simulations are based on $100,000$ random draws.

    \begin{itemize}

        \item $f$ is the uniform distribution $U(0,1)$:
            The single deadline mechanism $M(1)$ (essentially the serial cost sharing mechanism)
            has an expected sum-delay of $1.006$, which is calculated via numerical simulation.
            $Fail(500)$ is then at most $0.002$.
            $r^*=1$. The lower bound is $0.998$, which
            is close to our achieved sum-delay $1.006$.\\

        \item $f$ is the normal distribution $N(0.5,0.1)$ restricted to $[0,1]$:
            The single deadline mechanism $M(1)$'s expected sum-delay equals $2.3\mathrm{e}-4$ in simulation, which is obviously close to optimality.\\

        \item $f$ is the beta distribution $Beta(0.5,0.5)$:
            The single deadline mechanism $M(0.01)$'s expected sum-delay equals $1.935$ in simulation.
            $Fail(500)$ is then at most $0.00387$.
            $r^*=1.927$. The lower bound equals $(1-0.00387)*r^*=1.920$, which
            is very close to the achieved sum-delay of $1.935$.
            The serial cost sharing mechanism $M(1)$ is far away from
            optimality in this example.
            The expected sum-delay of the serial cost sharing mechanism
            is much larger at $14.48$.\\
    \end{itemize}

\end{example}

% \section{Sequential Auction}
% The mechanisms are based on multi-round decisions and the number of rounds depends on the type profile.

% This mechanism has many sequential decisions rounds. In each sequential round, it has two vectors:$\{B,T\}$. Vector $B$ is the set of $b_i(i = 1...n)$, represents agent $i$'s bid in this round($\sum {b_i} = 1$  and $b_i \in [0,1]$). Vector $T$ is the set of $t_i(i = 1...n$  and $t_i \in [0,1])$, represents agent $i$'s part how much she or he can get. If $t_i = 1$, means agent $i$ can pay for $b_i$ and get whole part in this round if other agents all agree with their bids.

% The agents will be asked for a sequential bid (sequential $B$) in order. If all of agents agree their bids ($v_i \geq b_i (i = 1...n)$). then he or she will get his or her part ($t_i$). If any of them not agree with his/her bid, they will move into the next round.

% In the next round, they will have another gene $\{B,T\}$. If at the end of the sequential, they do not agree with any of the option levels. They will decide not to build the project.\\

% To make sure $\sum {b_i = 1}$, we use softmax to make sure it can pay 1 total cost (normalize the total cost) in every round or level.
% $$B = softmax (B')$$

% This sequential mechanism will include the largest unanimous mechanism, cost sharing mechanism.

\newpage
\section{Automated Mechanism Design for Smaller Number of Agents}

For smaller number of agents,
the single deadline mechanism family no longer contains a near optimal mechanism.
We propose two numerical methods for identifying better mechanisms for smaller
number of agents. One is by extending the single deadline mechanism family
and the other is via evolutionary computation.\\

\subsection{Multiple Deadline Mechanisms}

The first method is fairly straightforward. We could extend the single deadline
mechanism family as follows:\\

\begin{definition}[Multiple Deadline Mechanisms]
    A multiple deadline mechanism\\
    $M(d_1,d_2,\ldots,d_n)$ is characterized
    by $n$ different deadlines. Agent $i$'s non-free part is $[0,d_i]$ and
    her free part is $[d_i,1]$. The mechanism's rules are otherwise identical
    to the single deadline mechanisms.\\
    % : the highest $k$
    % agents who are willing to pay $\frac{1}{k}$ for their non-free parts
    % get their non-free parts. $k$ is set to be as large as possible.
    % An agent gets her free part if and only if the other agents (excluding her) can
    % still successfully cost share.
\end{definition}

We simply use exhaustive search to find the best set of deadlines. Obviously,
this approach only works when the number of agents is tiny.

\newpage
\subsection{Automated Mechanism Design via Evolutionary Computation}

Ohseto~ (\cite{Ohseto2000:Characterizations}) characterized all strategy-proof and
individually rational mechanisms for the binary public project model (under
several minor technical assumptions). We summarize the author's characterization as follows:

\begin{itemize}

    \item {\em Unanimous mechanisms} (characterization for the nonexcludable model): Under an unanimous mechanism, there is a cost share vector $(c_1,c_2,\ldots,c_n)$ with $c_i\ge 0$ and
        $\sum_{i}c_i=1$. The project is built if and only if all agents accept this cost share vector.\\

    \item {\em Largest unanimous mechanisms} (characterization for the
        excludable model): Under a largest unanimous
        mechanism, for every subset/coalition of the agents, there is a constant cost
        share vector.
        The agents initially face the cost share vector corresponding to the grand coalition. If some
        agents do not accept the current cost share vector, then they are forever excluded. The remaining
        agents face a different cost share vector based on who are left.
        If at some point, all remaining agents accept,
        then we build the project. Otherwise, the project is not built.\\
\end{itemize}

We extend the largest unanimous mechanisms by adding the {\em release time} element.\\

\begin{definition}[Sequential unanimous mechanisms]
    A cost share vector under a sequential unanimous mechanism includes
    both the payments and the release time:

    \[T_1,B_1,\quad T_2,B_2,\quad \ldots,\quad T_n,B_n\]
\\\\
    Agent $i$ accepts the above cost share vector if and only if her utility
    {\em based on her reported valuation} is nonnegative when paying $B_i$
    for the time interval $[T_i,1]$.  That is, agent $i$ accepts the above cost share
    vector if and only if {\em her reported valuation} is at least
    $\frac{B_i}{1-T_i}$.  $\frac{B_i}{1-T_i}$ is called the {\em
    unit price} agent $i$ faces. We require $B_i\ge 0$ and $\sum_{i}B_i=1$.
\\
    A sequential unanimous mechanism contains $m$ cost share vectors in a
    sequence.  The mechanism goes through the sequence and stops at the first
    vector that is accepted by all agents. The project is built and the agents'
    release time and payments are determined by the unanimously accepted cost
    share vector. If all cost share vectors in the sequence are rejected, then
    the decision is not to build.\\
\end{definition}

The largest unanimous mechanisms (can be interpreted as special cases with
binary $T_i$) form a subset of the sequential unanimous mechanisms.  The
sequential unanimous mechanisms' structure makes it suitable for genetic
algorithms --- we treat the cost share vectors as the {\em genes} and treat the
sequences of cost share vectors as the {\em gene sequences}. The sequential unanimous mechanisms are generally not strategy-proof. However, they
can be easily proved to be strategy-proof in two scenarios:

\begin{itemize}

    \item A sequential unanimous mechanism is strategy-proof when {\em the sequence
        contains only one cost share vector} (an agent faces a take-it-or-leave-it
        offer).
        This observation makes it easy to generate
        an initial population of strategy-proof mechanisms.\\

    \item If for every agent, as we go through the cost share vector sequence,
        the unit price an agent faces is {\em nondecreasing} and her release time is also
        {\em nondecreasing}, then the mechanism is strategy-proof.
        Essentially, when the above is satisfied,
        all agents prefer earlier cost share vectors.
        All agents are incentivized to report truthfully, as doing so enables them to secure the earliest possible cost share vector.\\
\end{itemize}

\newpage
The sequential unanimous mechanism family
{\em seems} to be quite expressive.\footnote{Let $M$ be a strategy-proof
mechanism. There exists a sequential unanimous mechanism $M'$ (with exponential
sequence length). $M'$ has an
approximate equilibrium where the equilibrium outcome is
arbitrarily close to $M$'s outcome.} Our experiments show that by optimizing
within the sequential unanimous mechanisms, we are able to identify mechanisms
that perform better than existing mechanisms.
Our approach is as follows:

\begin{itemize}

    \item Initial population contains $200$ strategy-proof mechanisms. Every initial mechanism is
        a sequential unanimous mechanism with only one cost share vector.
        The $B_i$ and the $T_i$ are randomly generated by sampling $U(0,1)$.\\

    \item We perform evolution for $200$ rounds. Before each round, we filter
        out mechanisms that are not truthful. We have two different filters:

        \begin{itemize}

            \item Strict filter: we enforce that every agent's unit price faced and
                release time must be nondecreasing. With this filter, the final
                mechanism produced must be strategy-proof. We call this variant
                the {\em Truthful Genetic Algorithm (TGA)}.\\

            \item Loose filter: we use simulation to check for
                strategy-proofness violations.  In every evolution round, we
                generate $200$ random type profiles.  For each type profile and
                each agent, we randomly draw one false report and we filter out
                a mechanism if any beneficial manipulation occurs.
                After finishing evolution, we use $10,000$ type profiles to
                filter out the untruthful mechanisms from the final population.
                It should be noted that, we can only claim that the
                remaining mechanisms are {\em probably} truthful.
                We call this variant
                the {\em Approximately Truthful Genetic Algorithm (ATGA)}.\\

        \end{itemize}

    \item We perform crossover and mutations as follows:

        \begin{itemize}

            \item Crossover (Figure~\ref{fig41}): We call the top $50\%$ of the
                population (in terms of fitness, {\em i.e.}, expected max-delay or sum-delay)
                the {\em elite population}. For every elite
                mechanism, we randomly pick another mechanism from the whole
                population, and perform a crossover by randomly swapping one gene segment.\\

\begin{figure}[h!]
\centering
\includegraphics[width=0.8\textwidth,height=2cm]{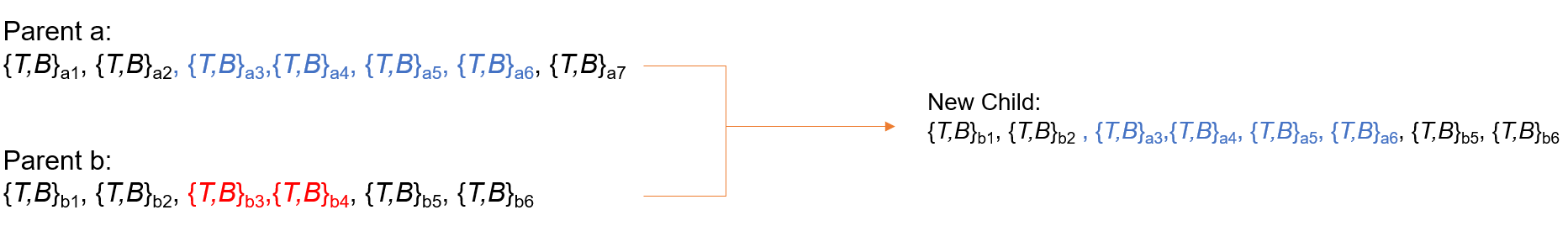}
\caption[Genetic algorithm: Crossover]{Genetic algorithm: Crossover}\label{fig41}
\end{figure}

            \item Mutation (Figure~\ref{fig2}): For every elite mechanism, with $20\%$ chance, we randomly select
                one gene, modify the offer of one agent.
                We insert that new cost share vector
                into a random position after the original position.

\begin{figure}[h!]
\centering
\includegraphics[width=0.8\textwidth,height=0.8cm]{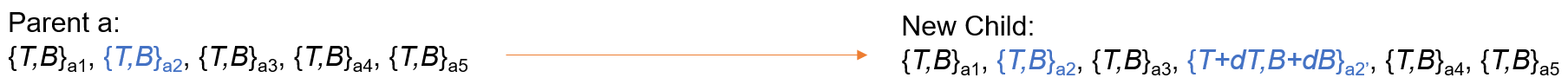}
\caption[Genetic algorithm: Mutation]{Genetic algorithm: Mutation}\label{fig2}
\end{figure}

            \item Neighbourhood Search (Figure~\ref{fig3}): For every elite mechanism, with $20\%$ chance,
                we randomly perturb one gene uniformly (from $-10\%$ to $+10\%$).

\begin{figure}[h!]
\centering
\includegraphics[width=0.8\textwidth,height=0.8cm]{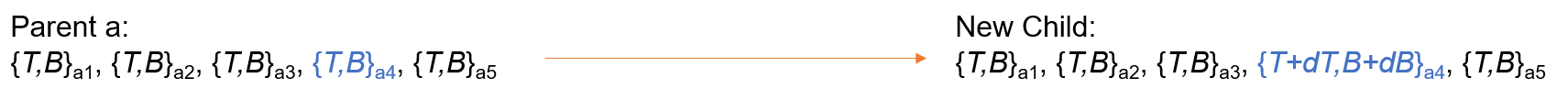}
\caption[Genetic algorithm: Neighborhood search]{Neighborhood Search}\label{fig3}
\end{figure}

        \end{itemize}

\item Abandon duplication and unused genes: In every evolution round,
    if a cost share vector is never unanimously accepted or if two cost share vectors
    are within $0.0001$ in terms of L1 distance.
    then we remove the duplication/unused genes.

\end{itemize}

% why 45% 45% 10%
% equal spilt other than 10%$

% neibhoughood search what is dt db
% shang xia 0.1/0.01 uniform pertub

% what is the percentage of three different crossover/mutations?
% uniform elite, uniform general pop, crossover
% every elite 0.2 mutate (not for every gene, random gene instead)
% every elite every gene 0.2 pertub

% abandon: two individuals same? or approximately same
% 10-4 difference

% how to initiaze? uniform distribution?
% 0.01 0.01 pool random draw

% elitesize?
% 0.5*popsize

% lie once in training?
% every pop, every agent, lie once

% population size 300 500?? different pop size for hte table?
% population 200, but keep best individual every iteration
% gmax 200

% 1000000 taoli? what does 1000000mean?
% same as training
\newpage
\subsection{Experiments}

We present the expected max-delay and sum-delay for $n=3,5$ and for different
distributions.  ATGA is only approximately truthful. We recall that in our
evolutionary process, in each round, we only use a very loose filter to filter
out the untruthful mechanisms. After evolution finishes, we run a more rigorous
filter on the final population (based on $10,000$ randomly generated type
profiles). The percentage in the parenthesis is the percentage of mechanisms
surviving the more rigorous test.  The other mechanisms (TGA and Multiple
deadlines) are strategy-proof. SCS is the serial cost sharing mechanism from
Moulin~ (\cite{Moulin1994:Serial}). According to Guo {\em et al.}~ (\cite{Guo2018:Cost})'s
experiments, SCS has the best known expected delays, so we use it as a benchmark.\\

\begin{table}[htp]
\centering{}
\small
\begin{tabular}{cccccc}
\multicolumn{1}{c|}{\textit{\textbf{n=3,sum-delay}}} & {\color[HTML]{303030} ATGA}                  & {\color[HTML]{303030} TGA}            & {\color[HTML]{303030} Single deadline} & {\color[HTML]{303030} Multiple deadline} & {\color[HTML]{303030} SCS}            \\ \hline
\multicolumn{1}{c|}{\textbf{Uniform(0,1)}}           & {\color[HTML]{303030} \textbf{1.605(95\%)}}  & {\color[HTML]{303030} \textbf{1.605}} & {\color[HTML]{303030} \textbf{1.605}}  & {\color[HTML]{303030} \textbf{1.605}}    & {\color[HTML]{303030} \textbf{1.605}} \\
\multicolumn{1}{c|}{\textbf{Beta(0.5,0.5)}}          & {\color[HTML]{303030} \textbf{1.756(89\%)}}  & {\color[HTML]{303030} \textbf{1.757}} & {\color[HTML]{303030} \textbf{1.757}}  & {\color[HTML]{303030} \textbf{1.757}}    & {\color[HTML]{303030} \textbf{1.757}} \\
\multicolumn{1}{c|}{\textbf{Bernoulli(0.5)}}         & {\color[HTML]{303030} \textbf{0.869(100\%)}} & {\color[HTML]{303030} \textbf{0.868}} & {\color[HTML]{303030} 1.499}           & {\color[HTML]{303030} 1.253}             & {\color[HTML]{303030} 1.498}          \\
\multicolumn{1}{c|}{\textbf{50\% 0, 50\% 0.8}}       & {\color[HTML]{303030} \textbf{1.699(98\%)}}  & {\color[HTML]{303030} 1.873}          & {\color[HTML]{303030} 1.873}           & {\color[HTML]{303030} 1.873}             & {\color[HTML]{303030} 1.873}          \\
\textbf{}                                            & {\color[HTML]{303030} \textbf{}}             & {\color[HTML]{303030} }               & {\color[HTML]{303030} }                & {\color[HTML]{303030} }                  & {\color[HTML]{303030} }               \\
\multicolumn{1}{c|}{\textit{\textbf{n=3,max-delay}}} & {\color[HTML]{303030} ATGA}                  & {\color[HTML]{303030} TGA}            & {\color[HTML]{303030} Single deadline} & {\color[HTML]{303030} Multiple deadline} & {\color[HTML]{303030} SCS}            \\ \hline
\multicolumn{1}{c|}{\textbf{Uniform(0,1)}}           & {\color[HTML]{303030} \textbf{0.705(97\%)}}  & {\color[HTML]{303030} \textbf{0.705}} & {\color[HTML]{303030} \textbf{0.705}}  & {\color[HTML]{303030} \textbf{0.705}}    & {\color[HTML]{303030} \textbf{0.705}} \\
\multicolumn{1}{c|}{\textbf{Beta(0.5,0.5)}}          & {\color[HTML]{303030} \textbf{0.754(87\%)}}  & {\color[HTML]{303030} 0.757}          & {\color[HTML]{303030} 0.782}           & {\color[HTML]{303030} 0.757}             & {\color[HTML]{303030} 0.782}          \\
\multicolumn{1}{c|}{\textbf{Bernoulli(0.5)}}         & {\color[HTML]{303030} \textbf{0.5(100\%)}}   & {\color[HTML]{303030} \textbf{0.498}} & {\color[HTML]{303030} 0.687}           & {\color[HTML]{303030} \textbf{0.50}}     & {\color[HTML]{303030} 0.877}          \\
\multicolumn{1}{c|}{\textbf{50\% 0, 50\% 0.8}}       & {\color[HTML]{303030} \textbf{0.676(94\%)}}  & {\color[HTML]{303030} 0.753}          & {\color[HTML]{303030} 0.749}           & {\color[HTML]{303030} 0.749}             & {\color[HTML]{303030} 0.877}          \\
\multicolumn{1}{l}{}                                 & \multicolumn{1}{l}{}                         & \multicolumn{1}{l}{}                  & \multicolumn{1}{l}{}                   & \multicolumn{1}{l}{}                     & \multicolumn{1}{l}{}                  \\
\multicolumn{1}{c|}{\textit{\textbf{n=5,sum-delay}}} & {\color[HTML]{303030} ATGA}                  & {\color[HTML]{303030} TGA}            & {\color[HTML]{303030} Single deadline} & {\color[HTML]{303030} Multiple deadline} & {\color[HTML]{303030} SCS}            \\ \hline
\multicolumn{1}{c|}{\textbf{Uniform(0,1)}}           & {\color[HTML]{303030} 1.462(95\%)}           & {\color[HTML]{303030} 1.503}          & {\color[HTML]{303030} \textbf{1.415}}  & {\color[HTML]{303030} \textbf{1.415}}    & {\color[HTML]{303030} \textbf{1.415}} \\
\multicolumn{1}{c|}{\textbf{Beta(0.5,0.5)}}          & {\color[HTML]{303030} 2.279(92\%)}           & {\color[HTML]{303030} 2.12}           & {\color[HTML]{303030} \textbf{1.955}}  & {\color[HTML]{303030} \textbf{1.955}}    & {\color[HTML]{303030} \textbf{1.955}} \\
\multicolumn{1}{c|}{\textbf{Bernoulli(0.5)}}         & {\color[HTML]{303030} \textbf{1.146(100\%)}} & {\color[HTML]{303030} 1.867}          & {\color[HTML]{303030} 2.106}           & {\color[HTML]{303030} 1.711}             & {\color[HTML]{303030} 2.523}          \\
\multicolumn{1}{c|}{\textbf{50\% 0, 50\% 0.8}}       & {\color[HTML]{303030} 2.432(94\%)}           & {\color[HTML]{303030} 2.845}          & {\color[HTML]{303030} 2.323}           & {\color[HTML]{303030} \textbf{2.248}}    & {\color[HTML]{303030} 2.667}          \\
\textbf{}                                            & {\color[HTML]{303030} \textbf{}}             & {\color[HTML]{303030} }               & {\color[HTML]{303030} }                & {\color[HTML]{303030} }                  & {\color[HTML]{303030} }               \\
\multicolumn{1}{c|}{\textit{\textbf{n=5,max-delay}}} & {\color[HTML]{303030} ATGA}                  & {\color[HTML]{303030} TGA}            & {\color[HTML]{303030} Single deadline} & {\color[HTML]{303030} Multiple deadline} & {\color[HTML]{303030} SCS}            \\ \hline
\multicolumn{1}{c|}{\textbf{Uniform(0,1)}}           & {\color[HTML]{303030} 0.677(91\%)}           & {\color[HTML]{303030} 0.677}          & {\color[HTML]{303030} \textbf{0.662}}  & {\color[HTML]{303030} \textbf{0.662}}    & {\color[HTML]{303030} 0.678}          \\
\multicolumn{1}{c|}{\textbf{Beta(0.5,0.5)}}          & {\color[HTML]{303030} 0.754(79\%)}           & {\color[HTML]{303030} 0.75}           & {\color[HTML]{303030} \textbf{0.73}}   & {\color[HTML]{303030} \textbf{0.73}}     & {\color[HTML]{303030} 0.827}          \\
\multicolumn{1}{c|}{\textbf{Bernoulli(0.5)}}         & {\color[HTML]{303030} 0.506(100\%)}          & {\color[HTML]{303030} \textbf{0.50}}  & {\color[HTML]{303030} 0.577}           & {\color[HTML]{303030} \textbf{0.50}}     & {\color[HTML]{303030} 0.971}          \\
\multicolumn{1}{c|}{\textbf{50\% 0, 50\% 0.8}}       & {\color[HTML]{303030} \textbf{0.666(80\%)}}  & {\color[HTML]{303030} 0.751}          & {\color[HTML]{303030} 0.736}           & {\color[HTML]{303030} 0.679}             & {\color[HTML]{303030} 0.968}
\end{tabular}
\caption[Experiment result: Our methods' sum-delay and max-delay vs state of the art]{Experiment result: Our methods' sum-delay and max-delay vs state of the art.}
\textbf{We see that ATGA performs well in many settings. If we focus on {\em provable} strategy-proof mechanisms, then TGA and the optimal multiple deadline mechanism also often perform better than the serial cost sharing mechanism.}
\end{table}

\newpage
\section{Chapter Summary}
In this chapter, we study the excludable public project model where the decision is binary (build or not build).  In a classic excludable and binary public project model, an agent either consumes the project in its whole or is completely excluded. We study a setting where the mechanism can set different project release times for different agents, in the sense that high-paying agents can consume the project earlier than the low-paying agents.  The mechanism design objective is to minimize the expected maximum release delay and the expected total release delay.  We propose the single deadline mechanisms.
We show that the optimal single deadline mechanism is asymptotically optimal for both objectives, regardless of the prior distributions.  For a small number of agents, we propose the sequential unanimous mechanisms by extending the largest unanimous mechanisms from Ohseto~ (\cite{Ohseto2000:Characterizations}). We propose an automated mechanism design approach via evolutionary computation to optimize within the sequential unanimous mechanisms.

%% file: Chapters/Chapter5.tex
\chapter{Redistribution in Public Project Problems via Neural Networks} % Main chapter title

\label{Redistribution in Public Project Problems via Neural Networks} % Change X to a consecutive number; for referencing this chapter elsewhere, use \ref{ChapterX}

In this chapter, we discuss VCG redistribution mechanisms (variants of the VCG mechanism) for the public project problems. We design mechanisms via neural networks with two welfare-maximizing objectives: optimal in the worst case and optimal in expectation.\\

We combine generative adversarial networks and multi-layer perceptions (GAN + MLP) to find the optimal worst-case VCG redistribution mechanisms for the public project problem. We use multi-layer perceptions (MLP) combined with a cost function that takes into consideration the agents' prior distributions to find the optimal-in-expectation VCG redistribution mechanisms for the public project problem. \\
$$$$

\section{Introduction}

\subsection{VCG Redistribution Mechanisms}

Many important problems in multiagent systems are related to resource allocations. The problem of allocating one or more resources among a group of competing agents can be solved through economic allocation mechanisms that take the agents’ reported valuations for the resources as input, and produce an allocation of the resources to the agents, as well as payments to be made by the agents. As a central research branch in economics and game theory, mechanism design concerns designing collective decision-making rules for multiple agents, to achieve desirable objectives, such as maximizing the social welfare, while each agent pursues her own utility. A mechanism is efficient if the agents who value the resource the most will get it. A mechanism is strategy-proof if the agents have the incentives to report their valuations truthfully, which is to say, an agent's utility is maximized when reporting her true valuation, no matter how the other agents report.\\

The Vickrey-Clarke-Groves (VCG) mechanism is a celebrated efficient and strategy-proof mechanism. Under the VCG mechanism, each agent $i$ reports her private type $\theta_i$. The outcome that maximizes the agents' total valuations is chosen. Every agent is required to make a VCG payment $t(\theta_{-i})$, which is determined by the other agents' types. An agent's VCG payment is often described as how much this agent's presence hurts the other agents, in terms of the other agents' total valuations. The total VCG payment may be quite large, leading to decreased welfare for the agents. In particular, in the context of the public project problem, where the goal is often to maximize the social welfare (the agents' total utility considering payments), having large VCG payments are undesirable.\\

%mgnew: Was this idea proposed by these two? I typically cite Cavallo 2006 as the one who proposed this idea
%wz: Updated, I researched the two papers, and reckon that Cavallo 2006 should be identified as the one who proposed this idea.

To address the welfare loss due to the VCG payments, Cavallo (\cite{Cavallo2006:Optimal}) suggested that we first execute the VCG mechanism and then redistribute as much of the payments back to the agents, without violating the efficiency and strategy-proofness of the VCG mechanism, and in a weakly budget-balanced way. This is referred to as the VCG redistribution mechanism. The amount that every agent receives (or pays additionally) is called the redistribution payment. To maintain efficiency and strategy-proofness of VCG, the redistribution payment of an agent is required to be independent of her own valuation. To maintain weakly budget-balance, the total amount redistributed should never exceed the total VCG payment. The redistribution payment is characterized by a redistribution function $h$, where
$h(\theta_{-i})$ represents agent $i$'s redistribution payment.\\

There have been many successes on designing redistribution mechanisms for various multi-unit/combinatorial auction settings (\cite{Cavallo2006:Optimal, Clippel2014:Destroy, Faltings2005:Budget-Balanced, Guo2009:Worst, Moulin2009:Almost, Gujar2011:Redistribution, Guo2011:VCG, Guo2012:Worst, Guo2014:Better,tsuruta2014optimal,Guo2011:Budget}), including a long list of optimal/near-optimal mechanisms.
On the other hand, there hasn't been comparable success in solving for optimal redistribution mechanisms for the public project problem, despite multiple attempts (\cite{Naroditskiy2012:Redistribution, Guo2016:Competitive, Guo2017:Speed, Guo2019:Asymptotically}). In terms of optimal results, Naroditskiy {\em et al.} (\cite{Naroditskiy2012:Redistribution}) solved for the worst-case optimal mechanism for
three agents. Unfortunately, the authors' technique does not generalize to more than three agents. Guo (\cite{Guo2019:Asymptotically}) proposed a mechanism that is worst-case optimal when the number of agents approaches infinity, but for small number of agents, the mechanism is not optimal. For maximizing expected welfare, there are no existing results, because it is difficult for traditional mathematical analysis(eg: mixed integer programming) to maximize the expectation of welfare.\\ %Therefore, we use neural network to do two separate object. One is maximize the expectation of welfare, and the other is worst-case performance.

%In this paper, using neural networks, we successfully obtain better mechanisms for the worst-case objective and near-optimal mechanisms for the optimal-in-expectation objective.

\subsection{Designing VCG Redistribution Mechanisms via Neural Networks}

A recent emerging topic in mechanism design is to bring tools such as neural networks from machine learning to design mechanisms (\cite{Duetting2019:Optimal, Golowich2018:Deep, Manisha2018:Learning, Shen2019:Automated, Wang2021:Mechanism}). Duetting et al. (\cite{Duetting2019:Optimal}) proposed a neural network approach for the automated design of optimal auctions. They model an auction as a multi-layer neural network and frame optimal auction design as a constrained learning problem which can be solved using standard machine learning pipelines. {\em The training and testing type profiles are generated based on the prior distribution. The cost function involves the mechanism objective and the penalty for property violation.}
\\
Essentially, neural networks were used as tools for functional optimisation. Shen et al. (\cite{Shen2019:Automated}) proposed a neural network based framework to automatically design revenue optimal mechanisms. This framework consists of a seller’s network, which provides a menu of options to the buyers, and a buyer’s network, which outputs an action that maximizes her utility. Wang et al. (\cite{Wang2021:Mechanism}) studied mechanism design for the public project problem and proposed several technical innovations that can be applied to mechanism design in general to improve the performance of mechanism design via neural networks.

\begin{figure}[H]
  \centering
  \includegraphics[width=0.7\linewidth,height=15cm]{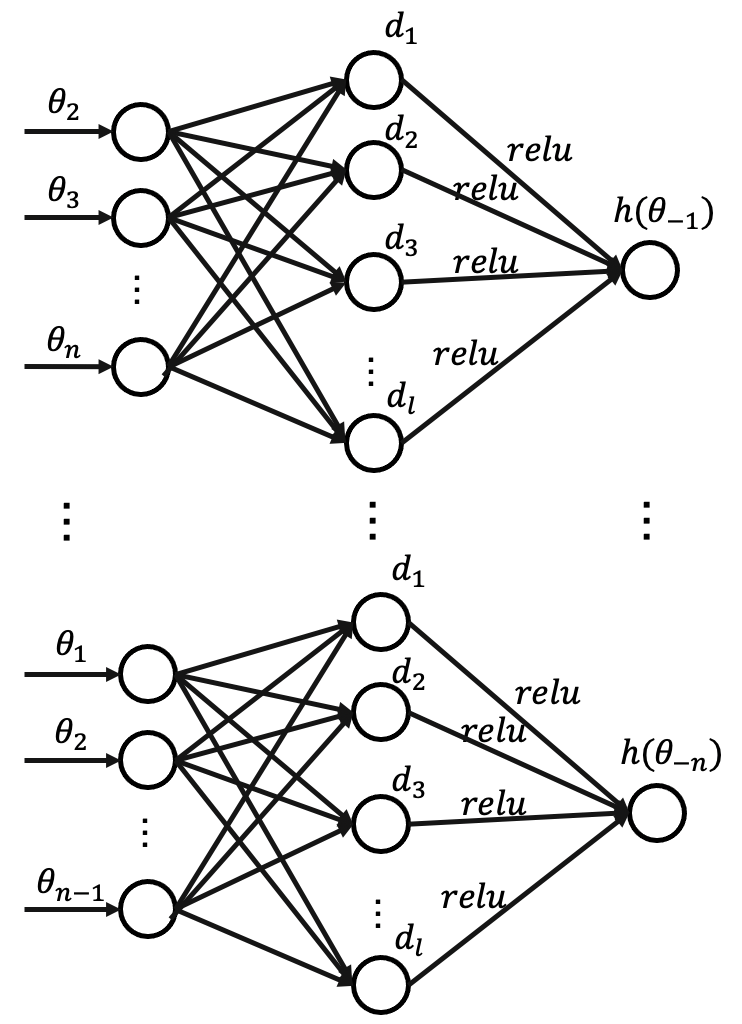}
  \caption[Manisha's neural network structure (\cite{Manisha2018:Learning})]{Neural network structure reported by Manisha et al. (\cite{Manisha2018:Learning})}
  \label{fig:ManishaNet}
\end{figure}

The work by Manisha et al. (\cite{Manisha2018:Learning}) is the first and only attempt to design VCG redistribution mechanisms using neural networks. They focused on multi-unit auctions with unit demand and studied both worst-case and optimal-in-expectation objectives. By randomly generating a large number of bid profiles, they train a neural network to maximize the total redistribution payment, while enforcing that the total redistribution should not exceed the total VCG payment. They modelled the redistribution function as a simple network outlined in Figure \ref{fig:ManishaNet}. It is a fully connected network with one hidden layer using ReLU activation. It takes the valuations of all the agents other than agent $i$ herself as input, and outputs the predicted redistribution payment for agent $i$. Their data is sampled randomly from uniform distribution ($\theta_i \in Uniform(0, 1)$).\\

%mg: using hi to represent hidden nodes looks bad, because you also have the h function, why don't we use some other characters for hidden nodes
%wz: updated

\subsection{Improved Neural Networks for Designing VCG Redistribution Mechanisms for the Public Project Problem}

In this chapter, we train neural networks to design VCG redistribution functions for the public project problem, which turns out to be a more challenging setting compared to multi-unit auctions studied by Manisha et al. (\cite{Manisha2018:Learning}). The public project problem is a classic mechanism design problem that has been studied extensively in economics and computer science (\cite{Mas-Colell1995:Microeconomic, Moore2006:General, Moulin1988:Axioms}). In this problem, $n$ agents decide whether or not to build a non-excludable public project, for example, a public bridge that can be accessed by everyone once built. Without loss of generality, we assume that the cost of the project is $1$, and $\theta_{i} (0\leq\theta_{i}\leq1)$ is agent $i$’s valuation for the project if it is built. If the decision is not to build, every agent retains her share of the cost, which is $1/n$.\\

%We consider to achieve either of the two design objectives: optimal in the expectation of the total utility and optimal in the worst-case efficiency ratio.

%For the worst case optimal objective, Guo \cite{Guo2019:Asymptotically} identifies an asymptotically optimal approach and proves that it has a worst-case efficiency ratio of $(n+1)/2n$.

%For the objective to optimize the expectation, it is hard for traditional methods (e.g. linear programming) to get the expectation of different prior distributions. Existing studies \cite{Naroditskiy2012:Redistribution,Guo2016:Competitive,Guo2017:Speed,Guo2019:Asymptotically} are not focused on this objective. Manisha et al. \cite{Manisha2018:Learning} use neural networks to design optimal mechanisms with the expectation objective, but for a different model.

We first evaluate the simple multilayer perceptron (MLP) model
proposed by Manisha et al. (\cite{Manisha2018:Learning}).
That is,  for each agent $i$, we train a neural network that maps $\theta_{-i}$ to agent $i$'s redistribution. The training and testing samples are randomly generated based on the prior distribution.
The cost function maximizes the mechanism design objective, as well as enforces mechanism design constraints via penalty.
We find that such a simple MLP is not effective enough for the public project problem for the following reasons:
\newpage
\begin{enumerate}
  \item From our experiments, by randomly generating the type profiles,
  we are not getting the true worst-case type profiles for the public project problem (it is a coincidence that for multi-unit auctions with unit demand, it is a lot easier to hit a worst case).
  \item Another challenge is the high input dimension when the number of agents is large. For 100 agents, the neural network has to take a 99-dimensional input, which is computationally unrealistic.
  \item In the public project problem, the agents' collective payments differ significantly between cases where the decision is to build the public project and cases
  where the decision is not to build.
  The number of ``build'' samples in a training batch significantly impacts
  the parameter gradient, which results in a wild loss fluctuation during the training process.\\
\end{enumerate}

To solve the aforementioned problems, we propose a novel neural network approach to design redistribution mechanisms for the public project problems.
Our approach involves the following technical innovations.\\

\paragraph{GAN Network} For the worst-case objective, we introduce a generative adversarial network (GAN) to generate worst-case type profiles, and then use these type profiles to train the redistribution function. Our experiment shows that a mechanism trained only using randomly generated data fails when facing
type profiles generated by GAN, so GAN is necessary and effective to derive the worst case.\\

\paragraph{Dimension Reduction} Instead of feeding $\theta_{-i}$ as input to train the mechanism, which has $n-1$ dimensions, we extract a few expressive features from $\theta_{-i}$ (e.g. the maximum of types $\theta_{-i}$, the sum of $\theta_{-i}$ excluding the maximum, etc.). This reduces the input dimension to 3. This helps the neural network loss converge faster and still retain good performance.\\

\paragraph{Supervised Learning} Wang et al. (\cite{Wang2021:Mechanism}) suggested that supervision to manual mechanisms often outperforms random initialization in terms of training speed by pushing the performance to a state that is already somewhat close to optimality. In addition, unlike many other deep learning problems, for mechanism design, there often exist simple and well-performing mechanisms that can be used as starting points. In this particular problem, we first conduct supervised learning to let the network mimic the state of art manual mechanism (\cite{Guo2019:Asymptotically}), and then leave it to gradient descent. This approach saves time for larger $n$ in our experiments.\\

\paragraph{Feeding Prior Distribution into Loss Function} We use probability density function (PDF) of the prior distribution to provide quality gradients. In specific, for each valuation profile $\theta = \{\theta_i\} (i=1..n)$ generated from a distribution $D$, we randomly choose $\theta_i$ to be replaced by a randomly generated $\theta'_i$ from $Uniform(0,1)$ and update $\theta$ to be $\theta'$. This sample is then assigned a weight based on the PDF.
In experiments, we see that this approach significantly reduces the loss fluctuation during training. One potential explanation (or observation) is that this approach reduces the fluctuation in the proportion of ``build'' cases among a batch.

With a more sophisticated network architecture due to the above technical adjustments, we get better results for the worst-case than the state of the art (\cite{Guo2019:Asymptotically}). For the optimal-in-expectation objective, our results are close to the theoretical optimal values.\\

% \subsection{Organization}

% The rest of this chapter is organized as followings: Section 2 gives a formal description of the VCG redistribution mechanism for the public project problem. Section 3 and 4 present our neural network approaches for the worst-case and expectation optimal objectives respectively, with the explanation of the network architecture, the technical adjustments, and loss function design. Section 5 discusses the training of neural networks and the experimental analysis. Section 6 concludes.

\newpage
\section{Model Description}

For the public project problem, VCG redistribution mechanisms have the following form (\cite{Naroditskiy2012:Redistribution}):
\begin{itemize}
  \item	Build the public project if and only if $\sum_i\theta_{i}\geq1$.
  \item	If the decision is to build, agent i receives $\sum_{j\neq i}\theta_{j}-h(\theta_{-i})$.
  \item	If the decision is not to build, then agent i receives $(n-1)/n-h(\theta_{-i})$.
  \item	$h$ is an arbitrary function and $\theta_{-i}$ refers to the types from the agents other than $i$ herself.\\
\end{itemize}

A VCG redistribution mechanism is characterized by the function $h$.

Due to Guo (\cite{Guo2019:Asymptotically}),

\begin{itemize}
  \item $S(\theta)=max\{\sum_i \theta_i, 1\}$ is exactly the first-best total utility. (I.e., if the sum of types is higher than $1$, then the efficient decision is to build. Otherwise, the efficient decision is not to build.)
  \item The agents' welfare (total utility considering payments) under type profile $\theta$ is $nS(\theta)-\sum_i h(\theta_{-i})$, which is obtained via simple algebraic simplification based on the definition of VCG redistribution mechanisms.\\
\end{itemize}

We consider two objectives. One is to find a mechanism that maximizes the worst-case efficiency ratio, and the other is to maximize the expected efficiency ratio.\\

\subsection{Worst-case Optimal Mechanism}

The efficiency ratio $r$ is defined as the ratio between the achieved total utility and the first best total utility:
$$r = \frac{nS(\theta)-\sum_i{h(\theta_{-i})}}{S(\theta)} = n - \frac{\sum_i{h(\theta_{-i})}}{S(\theta)}$$
\\
The worst-case efficiency ratio is the worst case ratio between the achieved total utility and the first best total utility, namely, the minimum of $r$ over all type profiles.
Due to Guo (\cite{Guo2019:Asymptotically}), the mechanism has a worst-case efficiency ratio $\alpha$ if and only if:
    \begin{equation}\label{equ:worst-const}
      \forall\theta, (n-1)\leq\sum_i h(\theta_{-i})/S(\theta)\leq(n-\alpha)
    \end{equation}
\\
In Inequality \ref{equ:worst-const}, the left side is the constraint for weakly budget-balance, and the right side corresponds to the definition of $\alpha$.

Therefore, taking the worst-case ratio as the objective, we need to design an $h$ function that:

\begin{maxi}[3]
{}{\alpha}{}{}
\addConstraint{\forall{\theta},(n-1)}{ \leq \sum_i{h(\theta_{-i})}/S(\theta)}{\leq (n-\alpha)}
\end{maxi}

\subsection{Optimal-in-Expectation Mechanism}

For this objective, we maximize the expected efficiency ratio $r$ to 1, which is equivalent to minimize $\sum_i{h(\theta_{-i})}/S(\theta)$ from above to $n-1$, with the consideration of the weakly budget-balance constraint.

We are designing an $h$ function that:
%mg your definition of expectation is wrong, but I guess we have to keep it as it is.

\begin{mini}[3]
{}{\overline {\sum_i{h(\theta_{-i})}/S(\theta)}}{}{}
\addConstraint{\forall{\theta},(n-1)}{\leq \sum_i{h(\theta_{-i})}/S(\theta)}
\end{mini}

\newpage
\section{Worst-case Optimal Mechanism}

In this section, we focus on the worst-case optimal mechanism. We first describe our neural network approach by explaining the network architecture and details of the relating techniques. Then we define the loss function.  \\

\subsection{Network Architecture}

\begin{figure*}[htbp]
\centering
\includegraphics[width=\textwidth,height=12cm]{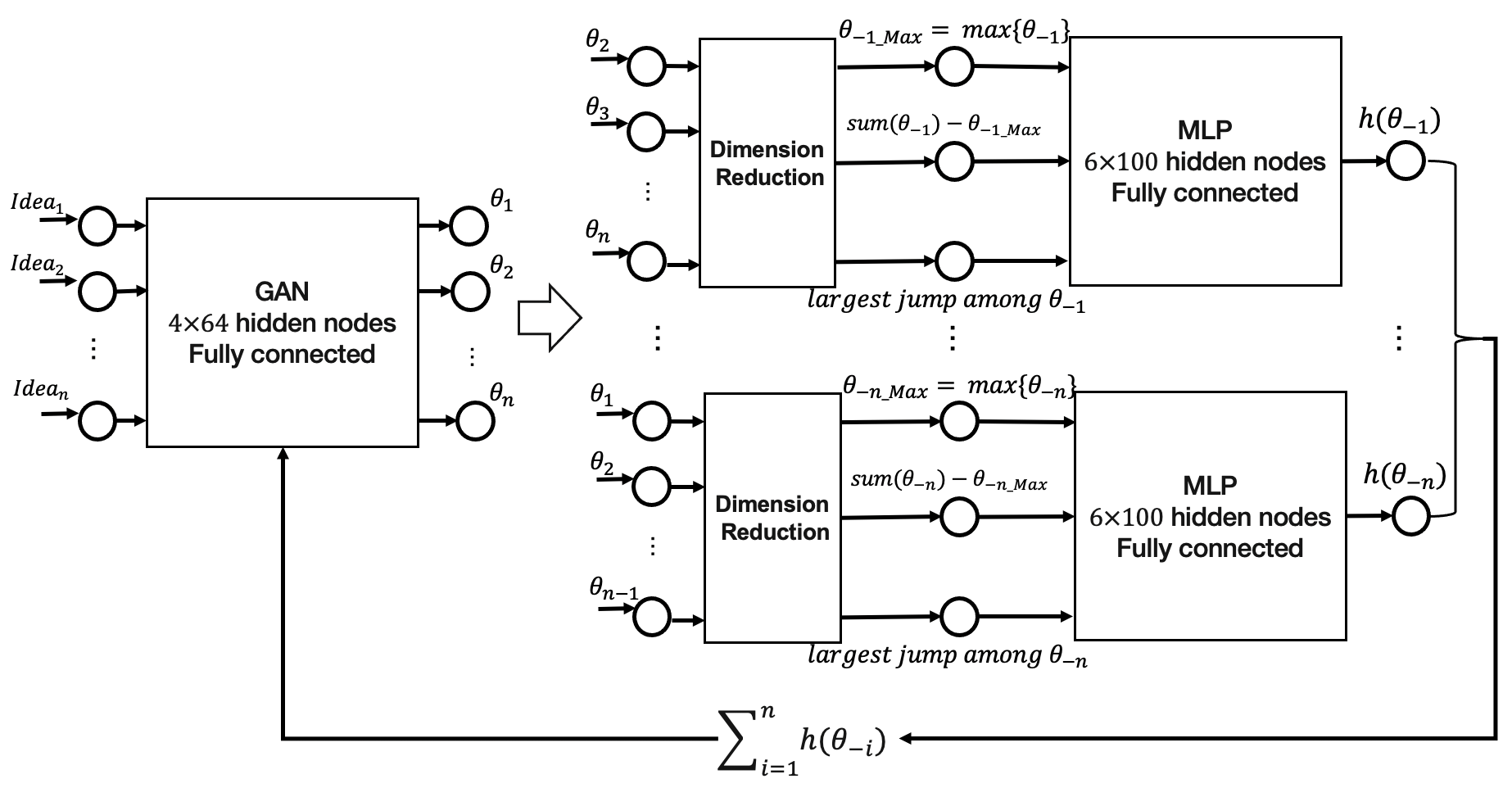}
\caption[Our neural network architecture for worst-case optimal scenarios]{Our neural network architecture for worst-case optimal scenarios}
\label{fig:GAN-MLP}
\end{figure*}

We construct a neural network to determine the $h$ function. As illustrated in Figure \ref{fig:GAN-MLP}, it is a network system in which a GAN and an MLP interacting with each other, we call it GAN+MLP. \\

The GAN works as a Generative Model and is used to generate special samples (type profiles). It takes $n$ randomly generated values as input ideas, and the output is $\theta=\{\theta_i\} (i=1...n)$. It is a fully connected network with 4 hidden layers, and each hidden layer contains 64 nodes. For a given batch ($batch\_size=b$), the GAN generates $b$ type profiles: $\theta^{(j)}\in batch = \{\theta^{(1)}, \theta^{(2)},...,\theta^{(b)}\}$, with the aim to maximize the difference between the maximum and the minimum of $\sum_i{h(\theta^{(j)}_{-i})}/S(\theta^{(j)})$. It means to:

$$maximize\ (\sum_i {h(\theta^{(j_1)}_{-i})}/S(\theta^{(j_1)}) - \sum_i{h(\theta^{(j_2)}_{-i})}/S(\theta^{(j_2)}))$$
\\
where $\theta^{(j1)}, \theta^{(j2)}\in batch$ is the sample that gives the maximum and minimum of $\sum_i{h(\theta^{(j)}_{-i})}/S(\theta^{(j)})$, respectively.\\

The MLP works as a Discriminative Model that learns the samples generated by the GAN. The MLP is a fully connected neural network. For each agent $i$, the network takes  $\theta_{-i}$ as the input, and outputs the value of $h(\theta_{-i})$. In the MLP, there are 6 hidden layers, each of which contains 100 nodes and with ReLU as the activation function. We first train the MLP under supervision to the best-performing manual mechanism and then leave it to unsupervised learning. For unsupervised learning, our cost function is the combination of design objective and also penalty due to constraint violation.\\

\subsection{Details of the Networks and Evaluations}

To improve the result for the neural network, we adopt some technical adjustments. We use a GAN instead of uniform to generate special cases in order to find out the worst case. For cases with a greater number of agents ($n\geq 5$), we adopt two technical tricks: Dimension Reduction and Supervised Learning. \\

\newpage
\subsubsection{GAN Network}

In previous studies, the authors used random generation or fixed data to find the worst case (\cite{Manisha2018:Learning}). As mentioned in Section 3.1, We propose a new GAN approach to find out the worst case. We conduct a contrast experiment to verify the validation of the GAN. We use only data generated from uniform distribution to train a network, and test this network with two sets of data.

\begin{itemize}
    \item Test Set A: 20000 data drawn from $Uniform(0, 1)$
    \item Test Set B: 10000 data generated from a trained GAN network and 10000 data drawn from $Uniform(0, 1)$\\
\end{itemize}

Figure \ref{fig:GAN-comp} outlines the experimental results for $n=10$ showing the difference of the the network performance on Test Set A and B. In the left figure, the network is tested by randomly generated data set A. It gives $\alpha=0.896$, and $\sum_{i}{h(\theta_{-i})}/ S(\theta)$ is between 9 and 9.104. However, in the right figure, the network performs poorly with significant violation of the weakly budget-balance constraint ($\sum_{i}{h(\theta_{-i})} / S(\theta)$ is from 3 to 9.2). \\

\begin{figure}[ht]
\begin{minipage}{0.5\textwidth}
 \centering
 \includegraphics[width=\linewidth]{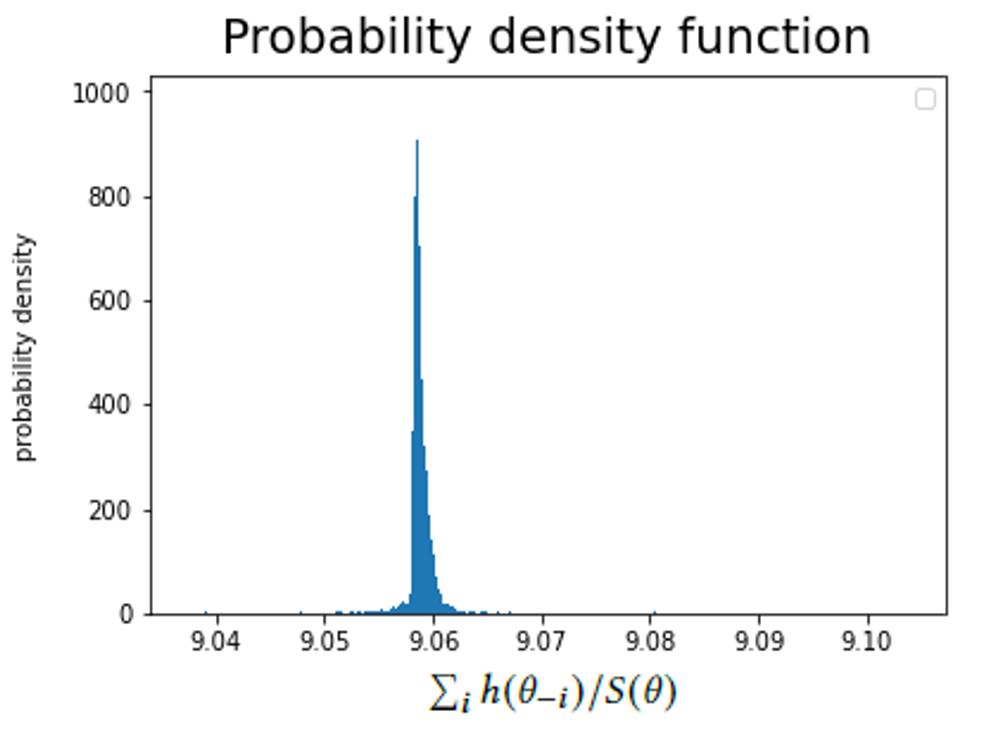}
\end{minipage}\hfill
\begin {minipage}{0.5\textwidth}
 \centering
 \includegraphics[width=\linewidth]{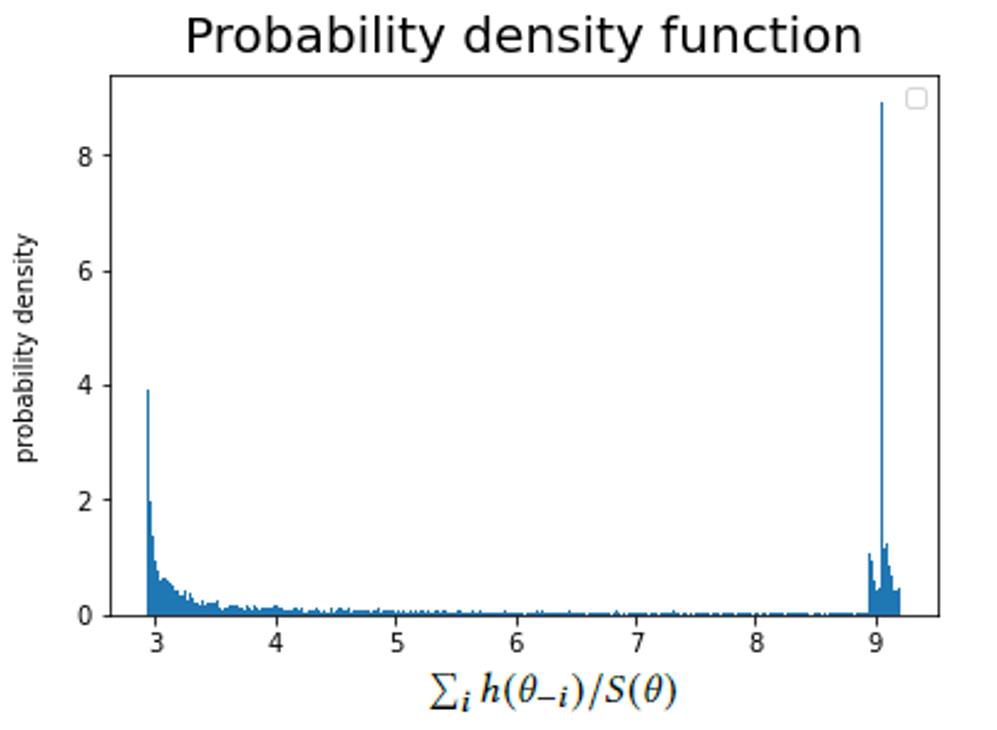}
\end{minipage}
\caption[Experiment result: Use vs not use GAN]{Experiment result: Use vs not use GAN. } \label{fig:GAN-comp}
\textbf{Spread of $\sum_{i}{h(\theta_{-i})} / S(\theta)$ evaluated using random type profiles and GAN generated type profiles. Random generation fails to generate true worst cases.}
\end{figure}

\newpage
Therefore, random type profile generation as used in Manisha et al. (\cite{Manisha2018:Learning}) does not work for this problem. To get the worst-case performance, we need a GAN network to generate higher quality worst-case profiles, and then let the MLP learn these profiles.

\subsubsection{Dimension Reduction}

The MLP takes a $(n-1)$-dimensional input for $n$ agents, resulting in expensive computation when $n$ is large. This motivates us to look for an effective dimension-reducing technique.\\

We first manually design a list of features that describe $\theta_{-i}$, and then experimentally search for a good combination of three features to be used for dimension reduction purposes. (We essentially reduce $\theta_{-i}$ to three dimensions this way.)\\

%group these features in different ways and use the combination as the input of a simple network. By comparing the performance of different inputs, we distinguish those expressive features from the good-performing combinations. We use the combination of the expressive features as our dimension-reducing mechanism, and apply it to the neural network to handle a relatively larger number of agents. It is verified that by adopting the dimension-reducing mechanism, the loss converges faster.

The features we consider include:
\begin{itemize}
  \item The highest type(s) from $\theta_{-i}$
  \item The the lowest type(s) from $\theta_{-i}$
  \item The sum of some types from $\theta_{-i}$
  \item The standard deviation of some types from $\theta_{-i}$
  \item The largest jump of adjacent types from $\theta_{-i}$\\
  Here $jump\ of\ adjacent\ types$ is defined as: for a sorted list ${\theta_{-i}}$, there is $jump_j$ between $\theta_j$ and $\theta_{j+1}$:
  $$jump_j = \theta_{j+1}-\theta_j$$\\
\end{itemize}

%Using different combinations to generate the input, we train the network for $n=5,...,10$, and compare the worst-case ratio $\alpha$ of the networks.
\newpage
We first experimentally derive that the following features are more important than the rest (i.e., removing them results in significant performance loss):

\begin{itemize}
  \item The highest type from $\theta_{-i}$
  \item The the lowest type from $\theta_{-i}$
  \item The sum of some type from $\theta_{-i}$\\
\end{itemize}

We then experimentally evaluate different combinations of the above features:

\begin{enumerate}
\item   Combination 1: the highest type(s) from $\theta_{-i}$ \& the sum of all the other types
\item   Combination 2: the highest type(s) from $\theta_{-i}$ \& the difference between the highest and the lowest type from $\theta_{-i}$
\item   Combination 3: the highest type(s) from $\theta_{-i}$ \& the standard deviation of all types from $\theta_{-i}$
\item   Combination 4: the highest type(s) from $\theta_{-i}$ \& the standard deviation of all the other types
\item   Combination 5: the highest type(s) from $\theta_{-i}$ \& the largest jump of adjacent types
\item   Combination 6: the highest type(s) from $\theta_{-i}$, the lowest type from $\theta_{-i}$, \& the largest jump of adjacent types
\item   Combination 7: the highest type(s) from $\theta_{-i}$, the sum of all the other types \& the largest jump of adjacent types
\item   Combination 8: the highest type(s) from $\theta_{-i}$, the lowest type from $\theta_{-i}$, \& the sum of all the other types\\
\end{enumerate}

\newpage
Figure \ref{fig:CombiComp-alpha} shows $\alpha$ of the networks trained with the different input combinations against the number of agents $n$. It is found that Combination 1, 7 and 8 performs better than the other combinations.

\begin{figure}[H]
  \centering
  \includegraphics[width=0.9\linewidth]{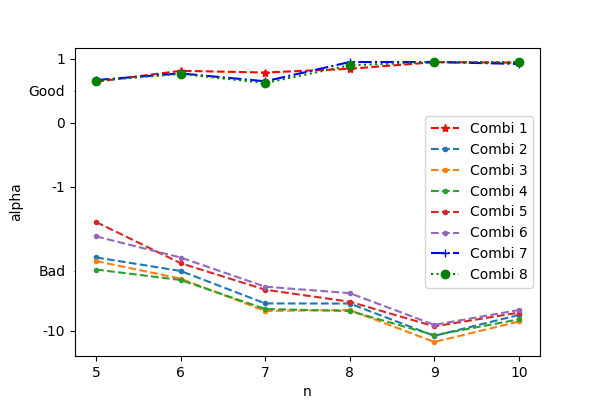}
  \caption[Experiment result: Different dimension reduction methods]{Experiment result: Effects of different dimension reduction methods.
}
  \label{fig:CombiComp-alpha}
  \textbf{The worst-case ratio $\alpha$ of the models trained with the input generated by using different feature-combinations against the number of agents $n=5, ..., 10$.}
\end{figure}

The above dimension-reducing mechanism improves both the training speed and sometimes the performance. We can infer that by using this dimension-reducing mechanism, the training speed would have a more significant improvement with the increase of the agent number $n$.\\

\newpage
Figure \ref{fig:comp-loss} shows the difference of the loss of the network between using and not using dimension-reducing mechanism for $n=10$. In the left figure, the training and test loss is stabilized within 1000 iterations with the application of the dimension-reducing mechanism, while in the right figure, the losses still vibrate in a wider range till 2000 iterations.
\\
\begin{figure}[H]
  \centering
  \includegraphics[width=\linewidth,height=6cm]{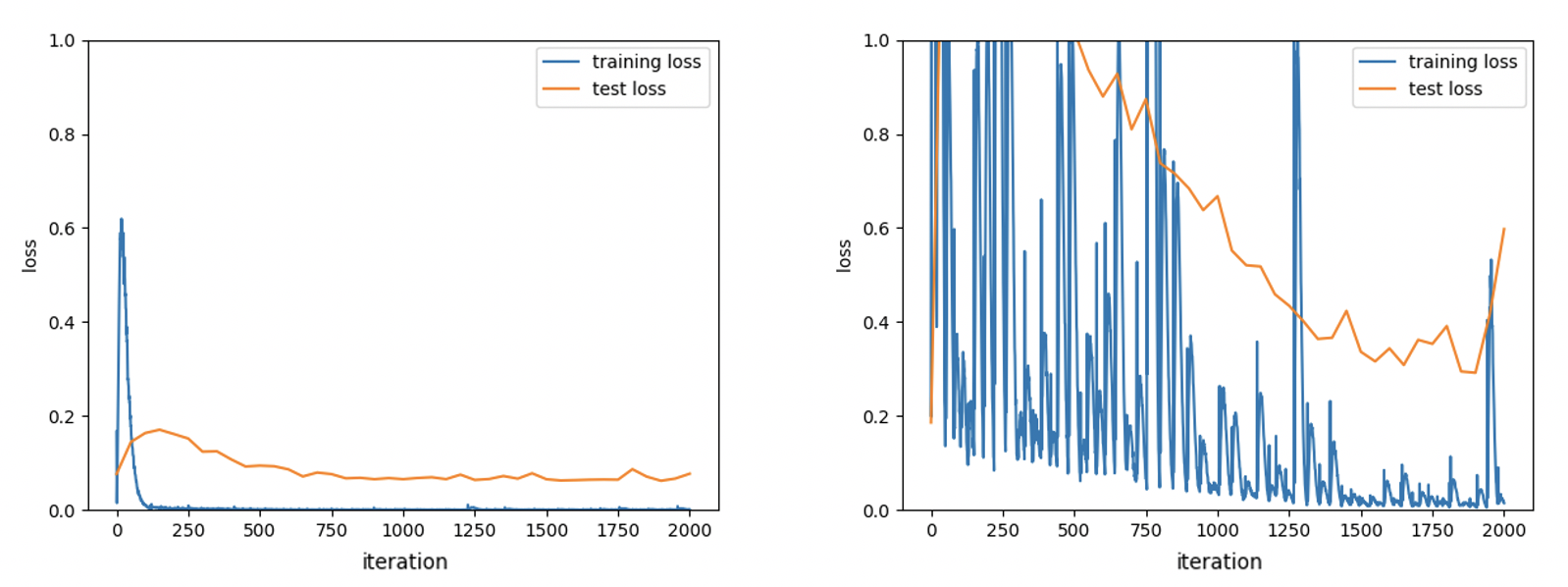}
  \caption[Experiment result: Speed of dimension reduction]{Experiment result: Speed of dimension reduction.}
  \label{fig:comp-loss}
  \textbf{ The loss during the training when using (left) and not using (right) the dimension-reducing mechanism for $n=10\ (\theta_i \in Uniform(0,1))$. Dimension reduction leads to faster convergence.}
\end{figure}

%mg you should describe the prior distribution for these figures
\newpage
Figure \ref{fig:comp-sumh-S} shows the difference on the spreading of $\sum_i{ h(\theta_{-i})}/S(\theta)$ between using and not using dimension-reducing mechanism for $n=10$. In the left figure, when using the dimension-reducing mechanism, we get an expected performance of $\sum_i{h(\theta_{-i})}/S(\theta) = 9.003$, which is very close to $n-1=9$. This indicates the corresponding mechanism is close to being optimal. (We recall that for the expected performance, we want the expected ratio to be as close to $n-1$ as possible.)
When not using the dimension-reducing mechanism, our performance is $9.377$, which is further away from $n-1=9$.

\begin{figure}[H]
  \centering
  \includegraphics[width=\linewidth,height=6cm]{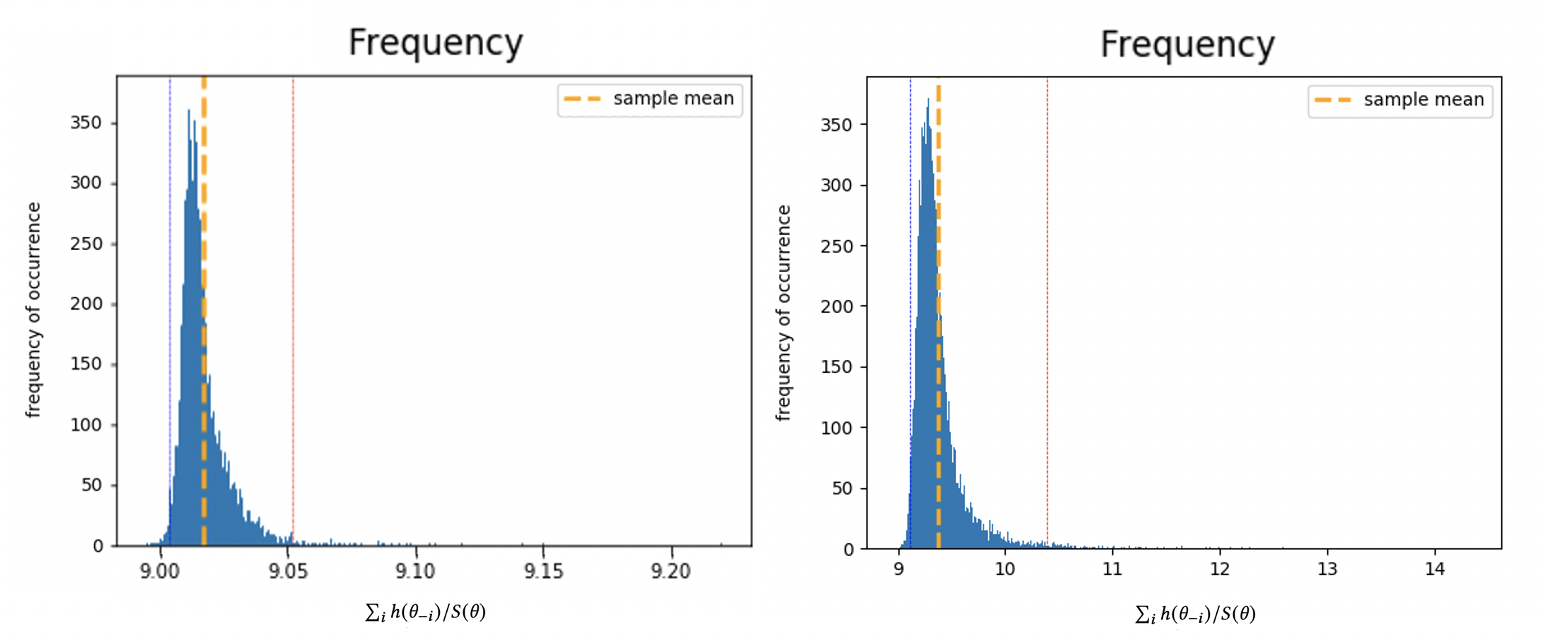}
  \caption[Experiment result: Use vs not use  dimension reduction]{Experiment result: Use vs not use  dimension reduction.}
  \label{fig:comp-sumh-S}
  \textbf{ Spreading of $\sum_i h(\theta_{-i})/S(\theta)$ when using (left) and not using (right) the dimension-reducing mechanism for $n=10\ (\theta_i \in Uniform(0,1))$, data size = 10000, bins = 500. Dimension reduction leads to better-performing mechanism.}
\end{figure}

\subsubsection{Supervised Learning}

For a greater number of agents, it takes gradient descent a long time to fix the constraint violations. Wang et al. (\cite{Wang2021:Mechanism}) suggested that we first supervise the neural network into the existing best manual mechanism, and then leave it to unsupervised learning. With the best manual mechanism as the starting point, better mechanisms can usually be found.
The existing best manual mechanism for worst-case optimal objective is reported by Guo (\cite{Guo2019:Asymptotically}).

%\begin{align*}
%h\_manual =  h^{3*''}(a,b,t)&=a-min\{a,t\}+b-min\{b,t\}\\
%& +man\{min\{a,t\}+min\{b,t\}, \frac{2t}{3}\}\\
%& +\frac{1}{2}max\{min\{a,t\}+min\{b,t\}, t\}\\
%& -\frac{1}{2}max\{min\{a,t\}, min\{b,t\}, \frac{2t}{3}\}-\frac{1}{6}
%\end{align*}

%mg the above is not the full description, so you either include the full description or not describe it at all
\newpage
\subsection{Loss Function}

\subsubsection{GAN Network Loss}
The GAN network has the following loss:
\begin{align*}
   loss_{GAN} = & min\{\sum_i { h(\theta^{(j_1)}_{-i})}/S(\theta^{(j_1)})\}\\
   & - max\{\sum_i{h(\theta^{(j_2)}_{-i})}/S(\theta^{(j_2)})\},
   \ & \theta^{(j_1)},\theta^{(j_2)} \in batch
\end{align*}

\subsubsection{Supervised Loss}

In supervised learning, we want the predicted $h$ to be as close as possible to the best manual value $h\_manual$ (\cite{Guo2019:Asymptotically}), so the loss is:
$$loss_{supervised} = (h - h\_manual)^2$$

\subsubsection{Unsupervised Loss}

In the unsupervised learning stage, we want all the $\sum_{i}{h(\theta_{-i})}/S(\theta)$ to maintain the weakly budget-balance constraint, and let $\alpha$ to be close to 1, which means to make the upper bound of $\sum_{i}{h(\theta_{-i})}/S(\theta)$ as small as possible.

The loss function includes two parts:

\begin{itemize}
\item $objective\_loss$ = $(relu(\sum_{i}{h(\theta_{-i})} - (n-up\_bound) S(\theta)))^2 $
\item $constraint\_loss$ = $(relu((n-1)S(\theta)- \sum_{i}{h(\theta_{-i})}))^2$
\end{itemize}

%mg it looks like the above contains typo, is ( outside of relu or inside, should be inside right
% updated, square loss, relu should be inside of the bracket.

Since the weakly budget-balance is a strict constraint, while the objective is soft, we add a multiplier $\epsilon$ to weaken the effect of $objective\_loss$. For the worst-case optimal network, we get the best $\epsilon=0.01$ through experiments.
\begin{align*}
loss_{unsupervised} &= \epsilon \cdot objective\_loss + constraint\_loss\\
&= objective\_loss/100 + constraint\_loss
\end{align*}

\newpage
\section{Optimal-in-Expectation Mechanism}

We design the optimal-in-expectation mechanism with slight modifications based on the worst-case mechanism.\\

The architecture of the MLP stays the same. We do not need a GAN for generating the worst-case since the worst-case does not matter. For a large agent number $n$, we also adopt dimension reduction and supervised learning as we do for the worst-case mechanism. In addition, we feed the prior distribution into the loss function to achieve a high-quality gradient. This network is called MLP+FEED.\\

\subsection{Feed Prior Distribution into Loss Function}

In this problem, the decision to build or not to build significantly affects the expectation of a training batch. For batches with different amounts of ``build'' cases, the gradients fluctuate significantly, causing worse training results and speed.\\

Wang et al. (\cite{Wang2021:Mechanism}) discovered a way to insert the cumulative distribution function (CDF) of the prior distribution into the cost function to get more accurate loss function. The approach was shown to be effective for optimal-in-expecation mechanism design. We adopt a similar idea, but use probability density function (PDF) from the prior distribution to provide quality gradients for our training process. For each valuation profile $\theta = \{\theta_i\} (i=1..n)$ generated from a distribution $D$, we randomly choose $\theta_i$ to be replaced by $\theta'_i$ which is regenerated from $Uniform(0,1)$ and update $\theta$ to be $\theta'$. \\

\newpage
The probability of the profile $\theta'$ is proportional to $PDF_D(\theta'_i)$, so the loss should also be multiplied by $PDF_D(\theta'_i$).\\

$$loss_{unsupervised\_feeding} = loss_{unsupervised} \cdot PDF_D(\theta'_i)$$

Here,
$$PDF_D(\theta'_i) = 10^{log\_prob(\theta'_i)}$$

$log\_prob$ is provided by PyTorch to calculate PDF. PyTorch distributions package is based on Schulman (\cite{schulman2015gradient}).
Our experiments show the difference between feeding and not feeding the distribution into loss function. Figure \ref{fig:Feed} shows that for normal distribution ($Normal(0.5,0.1),\ n=3$), feeding distribution into the loss function helps get a better gradient and thus dramatically improves the test result.
$$$$
% We feed with ... it shows no significant difference if we feed using...

\begin{figure}[H]
  \centering
  \includegraphics[width=0.7\linewidth]{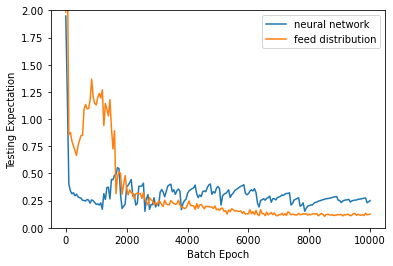}
  \caption[Experiment result: Feed vs not feed distribution]{Experiment result: Feed vs not feed distribution.} \label{fig:Feed}
  \textbf{Test loss (the distance from $\sum_{i}{h(\theta_{-i})} / S(\theta)$ to $(n-1)$) during the training for distribution: $Normal(0.5,0.1),\ n=3)$. Feeding prior distribution leads to faster convergence and better mechanism.}
\end{figure}

\newpage
\subsection{Loss Function}

The loss for the optimal-in-expectation network is similar to that of the worst-case MLP.

\subsubsection{Supervised Loss}

$$loss_{supervised} = (h - h\_manual)^2 $$

\subsubsection{Unsupervised Loss}

We use the square loss to approximate the loss for the expectation, since the derivative of square loss ($f(x) =x^2$) is a linear function, and can represent the strength of the gradient decedent, which is also linear ($f'(x)=ax+b$).

Loss function consists of $objective\_loss$ and $constraint\_loss$, where

\begin{itemize}
\item $objective\_loss = (relu(\sum_{i}{h(\theta_{-i})} - (n-1)S(\theta)))^2$
\item $constraint\_loss = (relu((n-1) S(\theta)- \sum_{i}{h(\theta_{-i})}))^2$
\end{itemize}

$objective\_loss$ aims to push $\sum_{i}{h(\theta_{-i})}/S(\theta)$ close to $(n-1)$, so that the mechanism redistributes as much of the collected VCG payment back to the agents. $constraint\_loss$ is for weakly budget-balance, to make the total redistributed amount less than the total VCG payment.

Similar with the worst-case optimal network, we use a multiplier $\epsilon$ to soften $objective\_loss$, and the experimentally best $\epsilon=10^{-4}$. So,
\begin{align*}
loss_{unsupervised} &= \epsilon \cdot objective\_loss + constraint\_loss\\
&= objective\_loss/10000 + constraint\_loss
\end{align*}

With the feeding of prior distribution to loss function as described in Section 4.1, we have:
\begin{align*}
&loss_{unsupervised\_feeding}\\
&= (objective\_loss/10000 + constraint\_loss)*PDF_D(\theta'_i)
\end{align*}

\newpage
\section{Experiments and Results}

We program with python by using the third-party library {\em PyTorch}. The experiments are conducted on a computer with an i5-8300H CPU and an Nvidia 1060 GPU. The experiment running time varies from a few minutes up to about 1 hour, depending on different agent numbers and data sizes.\\

\subsection{Experiment settings}

\subsubsection{Generate Data}

We randomly generate training and test data from prior distributions, and also use the GAN network. Half of the test data is generated from the GAN and the other half is randomly generated data. The test data size is from 10000 to 100000, and increases with $n$, as shown in Table \ref{tab52:data-size}. We generate new test data for each test.\\

\begin{table}[ht]
    \centering
    \small
    \begin{tabular}{cccccccc}
    \hline
        n & 4 & 5 & 6 & 7 & 8 & 9 & 10 \\ \hline
        Data size & 10000 & 20000 & 20000 & 20000 & 50000 & 50000 & 100000 \\
        \hline
    \end{tabular}
    \caption[Experiment setting: Test data size for different $n$]{Experiment setting: Test data size for different $n$.}
    \label{tab52:data-size}
\end{table}

\subsubsection{Batch Size}
Keskar et al. (\cite{keskar2017large}) show that a larger batch leads to a dramatic degradation in the quality of the model. They investigate the cause for this generalization drop in the large-batch regime and present numerical evidence that supports the view that large-batch methods tend to converge to sharp minimizers of the training and testing functions. It shows that the large batch size converges to the sharp minimum, while the small batch size converges to the flat.

Our experiments support Keskar's view and we find that the loss does not reduce with a big batch size ($\geq 1024$). We set the batch size to 64 through comparative experiments.\\

\subsubsection{Order input}

Arora et al. (\cite{arora2016understanding}) shows that it is hard for ReLU to simulate the function $max\{a,b\}$. A single $max\{a,b\}$ needs two layers and 5 nodes with exact weights and biases. So inputting sorted $\theta_{-i}$ values to the MLP is a necessary and important step to get good results. We order the valuations such that $\theta_1 \geq \theta_2 \geq ... \geq \theta_n$. For redistribution problems, sorted input values will not influence strategy-proofness.\\

%In our experiments, rather than directly using the valuation of the agents, we extract features, which also require sorted data.

\subsubsection{Initialization and Optimizer}

We use Xavier normal initialization for the weights and $Normal(0,0.01)$ for the bias, and use Adam optimizer with the $learning\_rate=0.001$ initially, and decays by 0.98 every 100 steps by Pytorch Scheduler.\\

\subsection{Results}

\subsubsection{Worst-case results}

We compare our result for the worst-case optimal mechanism(GAN+MLP) with the previously proposed mechanisms:
\begin{itemize}
\item SBR: heuristic-based SBR mechanism (\cite{Naroditskiy2012:Redistribution})
\item ABR: heuristic-based ABR mechanism (\cite{Guo2016:Competitive})
\item AMD: mechanisms derived via Automated Mechanism Design (AMD) (\cite{Guo2017:Speed})
\item AO: asymptotically optimal (AO) VCG Redistribution (\cite{Guo2019:Asymptotically})
\item UB: the {\bf conjectured} upper bounds (UB) on
the efficiency ratios (\cite{Naroditskiy2012:Redistribution})\\
\end{itemize}

\newpage
The result in Table \ref{tab:Worst-case} shows that our mechanism achieves better worst-case efficiency ratios than all previous results.\\

\begin{table}[H]
\centering
\begin{tabular}{ccccccc}
\hline
n & SBR & ABR   & AMD   & AO    & GAN + MLP      & UB    \\ \hline
4  &  0.354      & 0.459 & 0.600 & 0.625 & \textbf{0.634} & 0.666 \\
5  &  0.360       & 0.402 & 0.545 & 0.600 & \textbf{0.622} & 0.714 \\
6  &  0.394       & 0.386 & 0.497 & 0.583 & \textbf{0.592} & 0.868 \\
7  &  n too large       & 0.360 & 0.465 & 0.571 & \textbf{0.626} & 0.748 \\
8  &  n too large         & 0.352 & 0.444 & 0.563 & \textbf{0.654} & 0.755 \\
9   &  n too large        & 0.339 & 0.422 & 0.556 & \textbf{0.682} & 0.772 \\
10  &  n too large        & 0.336 & 0.405 & 0.550 & \textbf{0.623} & 0.882 \\ \hline
\end{tabular}
\caption[Experiment result: GAN+MLP Compare with state of the art for worst case]{Experiment result: GAN+MLP Compare with state of the art for worst case.}
\label{tab:Worst-case}
\end{table}

The main downside of our results is that our worst-case is calculated numerically by trying a large number of type profiles. This is a limitation
due to our neural network based approach. (Manisha et al. (\cite{Manisha2018:Learning})'s neural network approach also evaluated the worst-case by randomly generating a large number of type profiles. We showed that our GAN approach is a lot more rigorous compared to simply random profile generation.)
It should be noted that SBR, ABR and AMD's worst-cases were also calculated numerically. (AO's worst-case was derived analytically.)
To test the stability of the worst-case ratios, we experimented with different test sizes. Table \ref{tab51:data-size} shows that for $n=10$ with different test sizes of 10000, 20000 and 100000, $\alpha$ is stable at $0.623$.
\\
\begin{table}[ht]
    \centering
    \begin{tabular}{cccc}
    \hline
        Data size & 10000 & 20000 & 100000 \\\hline
        $\alpha$ & 0.623 & 0.623 & 0.623 \\
        \hline
    \end{tabular}
    \caption[Experiment setting: $\alpha$ for Different test size ($n=10$)]{Experiment setting: $\alpha$ for Different test size ($n=10$).}
    \label{tab51:data-size}
\end{table}

%For a large number of agents, our approach may not hit the absolute worst case. However, practically speaking, we ca
\newpage
\subsubsection{Optimal-in-expectation results}

We evaluate the expectation of $\sum_{i}{h(\theta_{-i})} / S(\theta)$ under our mechanism obtained via MLP+FEED. As defined in our model, for a case with $n$ agents, the theoretical optimal value for $\sum_{i}{h(\theta_{-i})} / S(\theta)$ is $n-1$, so we want the expectation to be as close to $n-1$ as possible from above (to maintain weakly budget-balance). Our results show that the obtained mechanisms are near optimal. For example, we get $E=4.061$ vs $n-1=4$ for $n=5$, $E=5.034$ vs $n-1=5$ for $n=6$.\\

\begin{table}[H]
\centering
\begin{tabular}{cccc}
\hline
MLP+FEED & Uniform(0,1) & Normal (0.5,0.1) & n-1 \\ \hline
n=3      & 2.079        & 2.101    & 2        \\
n=4      & 3.071        & 3.111    & 3       \\
n=5      & 4.061        & 4.142    & 4      \\
n=6      & 5.027        & 5.034    & 5       \\
n=7      & 6.009        & 6.067    & 6       \\
n=8      & 7.008        & 7.023    & 7       \\
n=9      & 8.002        & 8.008    & 8       \\
n=10     & 9.003        & 9.023    & 9       \\ \hline
\end{tabular}
\caption[Experiment result: Our result for optimal-in-expectation scenario ]{Experiment result: Our result for optimal-in-expectation scenario.}
\textbf{$\sum_{i}{h(\theta_{-i})} / S(\theta)$ in Expectation for Different Distributions}
\label{tab:Expectation}
\end{table}

Table \ref{tab:Expectation} shows that for both data generated from uniform distribution and normal distribution, the average $\sum_{i}{h(\theta_{-i})}/S(\theta)$ of our MLP+FEED network is very close to the theoretical optimal value $n-1$, which means that the redistribution function will return the vast majority of the total VCG payment to the agents (particularly for a great number of agents).\\

%If there exist some special cases such that $(n-1)S(\theta) \geq \sum_i {h(\theta_{-i})}$, we can just abandon them or use other methods\footnote{Cavallo \cite{Cavallo2006:Optimal} solve this by increasing $h(\theta_{-i})$ by $\delta/n$, where $\delta$ is calculated from the maximum violation of ($\theta_i, \theta_{-i}$). This method may hurt the optimal value a lot.} to redistribute the payment. These special circumstances hardly affect expectations.
\newpage
\section{Chapter Summary}

In this chapter, we consider designing optimal redistribution mechanisms for the public project problem under two objectives: worst-case optimal and optimal-in-expectation. With effective technical improvements on existing networks, we train a neural network to design good redistribution functions. We use a GAN network to generate valuation profiles to find the worst case, and feed prior distribution into loss function to get quality gradients for the optimal-in-expectation objective. To deal with large numbers of agents, we study different dimension-reducing methods and supervise the network into the existing manual mechanism as initialization.
Our experiments show that for the worst case, we could find better worst-case mechanisms compared to existing mechanisms, and for expectation, the neural networks can derive near-optimal redistribution mechanisms.

%% file: Chapters/Chapter6.tex
\chapter{Revenue-Maximizing Markets for Zero-Day Exploits} % Main chapter title

\label{Revenue-Maximizing Markets for Zero-Day Exploits} % Change X to a consecutive number; for referencing this chapter elsewhere, use \ref{ChapterX}

In this chapter, we study a mechanism design model called zero-day exploit market.  In such a market, one zero-day exploit (i.e., an exploit that allows cyber attackers to hack into iOS systems) is sold to multiple offender and defenders. In our model, for the defensive side, as long as any defender gains access to the exploit, the exploit is assumed to be immediately fixed, which benefits all defenders. The defensive side of the our model corresponds to a non-excludable public project problem. Otherwise, the model studied in this paper is only very loosely related to the public project model. The main goal of this chapter is to maximize revenue. We propose two numerical solution techniques for tuning the well-studied AMA mechanisms for revenue maximization, one is based on neural networks and the other one is based on evolutionary computation.

\section{Introduction}

Revenue-maximizing markets for zero-day exploits were defined in many research papers (\cite{Guo2016:Revenue,Egelman2013:Markets,Hata2017:Understanding,kanda2017towards}). Guo et al. (\cite{Guo2016:Revenue}) gave clear definitions for zero-day exploit markets. The authors proposed a Linear Programming (LP) based approach for tuning the parameters of AMA mechanisms for the purpose of revenue maximization. However, the authors' approaches had limitations. For example, it cannot handle too many constraints, and it failed to reach optimality in some cases.

Our contributions consist of two machine learning methods (neural networks and evolutionary computation) for optimizing within the AMA mechanism family for revenue maximization. The experiments show that our mechanisms based on these two methods are better than the existing LP mechanisms.  \\

\subsection{Zero-day Exploit Markets}

A zero-day exploit refers to a software bug which has not been disclosed to the public, and is also unknown to the software vendor. The zero-day exploit market has a long history and has been accepted by the security community  (\cite{Egelman2013:Markets}). The market for zero-day exploits is not necessarily a black market. The buyers like software vendors, police or national agencies typically purchase bugs through internal or community-run bug bounty reward programs. It has been widely reported that government agencies use zero-day vulnerabilities to track criminals or for other national security reasons. Some organisations buy exploits to ensure safety for themselves.\\

%These markets benefit all participants including: offenders, defenders and auctioneers.

% We focus on market issues in zero-day exploits. Zero-day exploits could be regarded as a public project problem. We propose two numerical solution techniques, one is based on neural networks and the other one is based on evolutionary computation. We use neural networks to automatically get the optimal curve for Affine Maximizer Auctions (AMA) mechanisms. We also use evolutionary computation based on Fourier series to optimize AMA mechanisms. The experiments show that our mechanisms based on neural networks and evolutionary computation are near-optimal and get better results compared to existing linear programming (LP) approaches.

\subsection{Problem Description}
\label{Model Description6.1}
\cite{Guo2016:Revenue} formally described the zero-day exploit market model. In this model, one exploit can be sold to one or more buyers. The seller is the mechanism designer, who wants to sell the exploit to maximize revenue. For example, the seller can be a cyber security company that sells bugs for profit.\\

\begin{assumption} One exploit is sold over a time frame from 0 to 1 ($[0,1]$). The exploit is available to be traded from time $0$ (zero-day), and $1$ is the moment that the exploit's life ends ({\em e.g.}, due to the end of life of the affected software, or the update of a major service pack).
    (\cite{Guo2016:Revenue})
\label{assumption 1}
\end{assumption}
According to Assumption \ref{assumption 1}, for each buyer $i$, we use $t_i \in [0,1]$ to denote the time agent $i$ receives the information regarding the exploit.\\

\begin{assumption}
    "There are two types of agents (buyers): {\em defenders} and {\em offenders}." (\cite{Guo2016:Revenue})
    \begin{itemize}
        \item A defender is a buyer who would like to fix the exploit.
        \item An offender is a buyer who buys the exploit in order to to utilize it (or attack it).\\
    \end{itemize}
\end{assumption}

For a given exploit, we assume that it can be fixed by any defender. More specifically, once an exploit is received by any defender, it is immediately fixed, rendering it worthless to all offenders. All defenders can enjoy a shared "protected" time interval from the moment the exploit is fixed ($t_{end}$) to $1$.
$t_{end}$ is the earliest time any defender obtains the exploit.  We say $t_{end}$ is the ending time of the exploit.

Buyer $i$'s type is a non-negative value. Function $v_i(t)$ is buyer $i$'s instantaneous valuation at time $t$.

\begin{itemize}
\item If $i$ is an offender, and he/she receives the exploit at $t_i$ ($t_i \in [0,1]$). Recall that the exploit gets fixed at $t_{end}$.  The valuation of the buyer $i$ (offender) is a integral, which equals
\begin{equation}
\int_{t_i}^{t_{end}}v_i(t)dt
\label{equation6.1}
\end{equation}
\\
\item If $i$ is a defender, then her/his value is determined by the interval between $t_{end}$ and the end of the exploit's life cycle, which equals

\begin{equation}
\int_{t_{end}}^{1}v_i(t)dt
\label{equation6.2}
\end{equation}
\\
\end{itemize}

The mechanism should satisfy a few desired properties: strategy-proofness, individual rationality and straight-forwardness. These properties are defined below for zero-day exploit markets:

\begin{definition} {\em Strategy-proofness (SP)}: For any buyer $i$, his/her utility is maximized when revealing $v_i(t)$ truthfully.\end{definition}

\begin{definition} {\em Individual rationality (IR)}: For any buyer $i$,  his/her utility is nonnegative when revealing $v_i(t)$ truthfully.\end{definition}

\begin{definition}
{\em Straight-forwardness (SF)}: A  straight-forward mechanism is defined as follow: before asking for offenders' valuation functions, the mechanism reveals the full details of the exploit to all offenders.\\
\label{SF}
\end{definition}

Here, SF is a property that is specifically introduced for zero-day exploit markets and only for this chapter. An exploit can be regarded as a piece of one-time information. As a result, if the auctioneer discloses the details of the exploit to the buyers before the auction, the buyers may immediately walk away with the information for free.  If the auctioneer does not describe what is being sold, it is hard for the buyers to come up with their valuation functions.

\begin{assumption}
We assume that there are two ways for the seller to describe an exploit: either describe the full details, or describe what can be achieved with the exploit ({\em e.g.}, with this exploit, anyone can seize full control of a Windows 10 system remotely). (\cite{Guo2016:Revenue})
    \begin{itemize}
        \item We assume that it is safe for the seller to disclose what can be achieved with the exploit. That is, the buyers will not be able to derive ``how it is done'' based on ``what can be achieved''. (\cite{Guo2016:Revenue})
        \item If the seller only discloses what can be achieved, then it is difficult for an offender to determine whether the exploit is new, or something she already knows, and thus difficult to come up with their valuation. (\cite{Guo2016:Revenue})
        \item We assume that the defenders are able to come up with valuation functions just based on what can be achieved. This is because all zero-day exploits are by definition unknown to the defenders. (\cite{Guo2016:Revenue})
        \\
    \end{itemize}

\end{assumption}

The above assumption leads to the SF property. It is important to note that SF does not require the disclosure of exploitation details to the defenders prior to their bidding.  If the seller does so, then the defenders can simply fix the exploit and bid $v_i(t)\equiv 0$. Due to IR, the defenders can go away without paying. Offenders are given the details before they bid, but they cannot simply bid $v_i(t)\equiv 0$ to go away without paying, which is due to the following reasoning.

Guo et al. (\cite{Guo2016:Revenue}) used the defenders as a ''THREAT''. That is, if offenders bid low, the auctioneer will disclose the exploit to the defenders early ($t_{end}$ would be smaller). The exploit then becomes less valuable for the offenders according to Equation \ref{equation6.1}. Essentially, the offenders are encouraged to bid/pay more to keep the exploit alive.  The higher they bid, the longer the exploit remains alive.\\

\subsection{Affine Maximizer Auctions Model Description}

For revenue maximization, many researchers developed well performing mechanisms. Myerson (\cite{myerson1981optimal})'s optimal auction is optimal for selling a  single item. For combinatorial auctions, Myerson's technique does not generalize beyond single-parameter settings. Revenue maximizing mechanism design remains an open problem for general combinatorial auctions. Many revenue-boosting techniques were proposed by researchers (\cite{Guo2014:Increasing,Guo2013:Revenue,Guo2015:Social}).
One particular revenue-boosting technique is the Affine Maximizer Auctions (AMA) mechanisms (\cite{Likhodedov2005:Approximating}). The AMA mechanism family is a rich family of mechanisms. AMA mechanisms are all strategy-proof and they are characterized by a set of parameters.  By focusing on the AMA mechanisms, the original zero-day market design problem is transformed into a value optimization problem where the mechanism designer only needs to adjust the AMA parameters.\\

The family of AMA mechanisms is formally defined by Guo et al. (\cite{Guo2016:Revenue}) as follows:
\\
\begin{tcolorbox}
    \begin{center}
        AMA Mechanisms
    \end{center}
\begin{itemize}
    \item Given a type profile $\theta$, the outcome picked is the following:
        \[o^*=\arg\max_{o\in O}\left(\sum_{i=1}^nu_iv_i(\theta_i,o)+a_o\right)\]

    \item Agent $i$'s payment equals:
        \[\frac{\max_{o\in O}\left(\sum_{j\neq i}u_jv_j(\theta_j,o)+a_o\right)
        - \sum_{j\neq i}u_jv_j(\theta_j,o^*)-a_{o^*}}{u_i}\]
        \\
    \end{itemize}
\end{tcolorbox}

Here, $O$ represents the outcome space, $\Theta_i$ represents agent $i$'s type space, and $v_i(\theta_i,o)$ represents agent $i$'s valuation for
outcome $o\in O$ when her type is $\theta_i\in \Theta_i$.
Under the model described in \ref{Model Description6.1}, the outcome space is $[0,1]$. To be more specific, an outcome $o \in [0,1]$ represents when the exploit ends (revealed to the defenders).
Our techniques require that the outcome
space be finite, so the outcome space is discretized into a set ($\{0,\frac{1}{k},\frac{2}{k},\ldots,1\}$). The outcome space size is $| O | = k + 1 $.\\

By focusing on AMA mechanisms defined like the above, we only need to adjust the $u_i$ and the $a_o$, which are the AMA mechanism parameters.

\newpage
\section{Optimizing Affine Maximizer Auctions via Neural Networks}

% An AMA mechanism $M(u_1,u_2,\ldots,u_n,a_0,a_1,\ldots,a_k)$
% is characterized by $n+|O|$ parameters.
% Given a specific AMA mechanism, to evaluate its expected revenue, we resort to generating a large set of sample
% profiles $S$ based on the prior distribution. For each sample profile $\theta\in S$, we evaluate the corresponding revenue and then take the average over all samples in $S$. Due to this, we require $n$ to be tiny so that we can accurately evaluate the expected revenue using $S$. For large $n$, we cannot afford to evaluate a mechanism's expected revenue, not to mention carrying out mechanism optimization.
% Realistically, the largest $S$ we can work with is in the magnitude of ``hundreds''.\\

Mechanism design via neural networks has recently drawn significant attention
  in the algorithmic game theory
  community~ (\cite{Duetting2019:Optimal,Shen2019:Automated,Manisha2018:Learning,Wang2021:Mechanism}).
In the context of our model, the high-level approach of tuning AMA parameters using neural networks is as follows:

\begin{itemize}
    \item Treat the $u_i$ and the $a_i$ as model parameters.
    \item Initialize the model parameters randomly or start from a known mechanism such as the VCG mechanism.
    \item In each learning step, we generate a batch of type profiles based on the prior distribution. We evaluate the current AMA mechanism's average revenue on this batch.
    Parameter gradient is calculated based on this average. Model parameters are adjusted via gradient descent.\\
\end{itemize}

\newpage
The {\em training data} are randomly sampled based on the prior distribution.
In each learning step, we generate a fresh batch of type profiles.
After we finish training, we generate another (much larger) fresh batch of type profiles to be the {\em testing data}.
It should be noted that we never need to reuse any type profile, so there is a separation between training and testing data.\\

One limitation of the neural network approach is that the batch size in
each learning step needs to small ({\em e.g.}, we set the batch size to be $16$).
If the batch size is very large, the time consumption of each learning step gets too long, and it
actually hurts the learning performance.
Obviously, it is insufficient to use only $16$ type profiles to {\em accurately} estimate
a mechanism's expected revenue, but generally speaking, we do not
need every learning step to move toward the correct direction. We only need
that in {\em most} learning steps we are moving toward the correct direction.\\

Due to the batch size limitation, for our neural network based approach, we focus on the case with only two agents (one defender and one offender). A batch of $16$ type profiles can be generated
by drawing $4$ sample types for each agent.\\

% We note that the $a_i$ can naturally
% be interpreted as a function $a(t)$ with $t \in [0,1]$.

With two agents, an AMA mechanism can be expressed as

\[M(u_{offender},u_{defender},a_0,a_1,\ldots,a_k)\]
\\
It is without loss of generality to set $u_{offender}=1$.
Besides the model parameter $u_{defender}$, the only other mechanism parameters
are the $a_i$.
We recall that $a_i$ represents the constant term in the AMA mechanism for outcome $\frac{i}{k}$. That is, the $a_i$ can naturally be represented using
a curve $a(t)$ with $t\in [0,1]$, where $a(t)$ is the constant term for outcome $t$.
The function $a(t)$ can be expressed using the following neural network:
\\
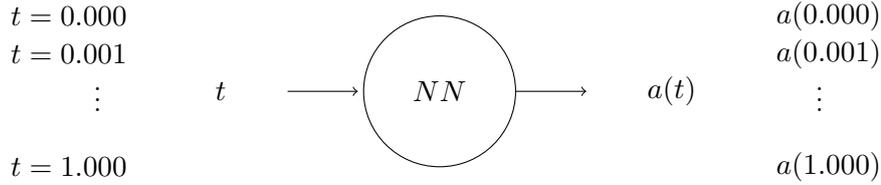
\begin{figure}[h]
\centering
\begin{tikzpicture}
    \centering
	\tikzstyle{unit}=[draw,shape=circle,minimum size=2cm]

	\node[unit](p)[shorten >=1pt,->] at (4,1){$NN$};
	\node(dots) at (-0.5,1){\vdots};

	\draw (0,2)     node[xshift=-25]  {$t=0.000$} ;
	\draw (0,1.5)   node[xshift=-25]  {$t=0.001$} ;

	\draw (0,0)     node[xshift=-25]{$t=1.000$} ;
	\draw (2,1)     node[xshift=-25]  {$t$}[shorten >=2pt,->] -- (p);
	\draw (p)[shorten >=2pt,->] --  (6,1) node[xshift=30] {$a(t)$};

	\draw (10,2)     node[xshift=-25]  {$a(0.000)$} ;
	\draw (10,1.5)   node[xshift=-25]  {$a(0.001)$} ;
	\node(dots) at (9,1){\vdots};
	\draw (10,0)     node[xshift=-25]  {$a(1.000)$} ;

\end{tikzpicture}
\caption{Neural network representation of the $a_i$ when $k=1000$} \label{fig:M1}
\end{figure}
\\
% We use a simple A fully connected network with four hidden layers to find the curve of the a(t), with $t \in [0,1]$. We implement an unsupervised neural network to optimizes maximise the revenue for the AMA mechanism. The linger programming(LP) can approximately solve the problem. LP can deal with the hundred of the parameters. the input T space to $\left\{0.0,\frac{1}{k} ,\frac{2}{k} ,...,1.0 \right\}$. However, the k can not set over hundreds partitions in LP to mark sure accuracy. The neural network can easily deal with thousands of parameters. In this neural network, we divide time frame T into a thousand parts.

In the context of the above representation, the AMA mechanism's allocation (the ending time) equals

$$t^* = \arg\max\limits_{t \in [0,1]} (\sum_{i\in N}u_i v_i( \theta_i,t)+a(t)  ) $$
$$ N=\left\{offender,defender\right\}$$
\\
Agent i's payment equals $(i\in \left\{offender,defender\right\})$:
$$p_i=\frac{\max \limits_{t \in [0,1]} \left\{\sum_{j \neq i} u_j v_j(\theta_j, t) + a(t) \right\} -  \left\{\sum_{j \neq i} u_j v_j(\theta_j, t^*) + a(t^*) \right\} }{u_i}$$
\\

% Explicitly, the payments are

% $$p_{offender}=\frac{\max \limits_{t \in [0,1]} \left\{ u_{defender} v_{defender}(t) + a(t) \right\} -  \left\{ u_{defender} v_{defender}(t^*) + a(t^*) \right\} }{u_{offender}}$$

% $$p_{defender}=\frac{\max \limits_{t \in [0,1]} \left\{ u_{offender} v_{offender}(t) + a(t) \right\} -  \left\{ u_{offender} v_{offender}(t^*) + a(t^*) \right\} }{u_{defender}}$$

To maximise the total revenue, the loss function is set to be:
  $$minimise: loss = - (p_{offender} + p_{defender})$$
\\
In training, we set $u_{defender}$ as an {\em autograd}~\footnote{For an overview of automatic differentiation in PyTorch, please refer to (\cite{Paszke2017:Automatic}).} parameter, and we use
a fully connected network to represent the function $a(t)$.
We assume that we have the analytical form of the agents' valuation function $v_i(\theta_i,t)$, which is to facilitate automatic differentiation needed by gradient descent.
We will present the experimental results in Section~\ref{sec:experiments}.

\newpage
\section{Optimizing Affine Maximizer Auctions via Evolutionary Computation}

As discussed in the above section, an AMA mechanism is characterized by
a curve $a(t)$ with $t\in [0,1]$.\footnote{For two agents, besides the curve
$a(t)$, we also have another model parameter $u_{defender}$, which is a single
parameter that can be dealt with separately (using a naive for loop).}
In the previous section on neural networks, we used neural networks to model
the curve $a(t)$. A natural idea is to consider other methods for expressing
curves, such as:

\begin{itemize}

    \item Piece-wise linear segments: We may join $k$ straight-line segments to
        form a curve.  The coordinates for the end points are $(\frac{j}{k},c_j)$ for
        $j=0,1,\ldots,k$.  To optimize for the best piece-wise linear segments,
        we just need to adjust the $c_j$.\\

    \item Polynomial:
        \[a(t)=c_kt^k+c_{k-1}t^{k-1}+\ldots+c_1t+c_0\]
        \\
        To optimize for the best polynomial representation of $a(t)$, again, we just
        need to adjust the $c_j$.\\

    \item Fourier series:
        \[a(t) = \frac{c_0}{2} + \sum_{j=1}^{N} (c_jcos (\frac{2\pi}{p}jt) + c_j' sin (\frac{2\pi}{p}jt)) \]
        To optimize for the best Fourier series representation of $a(t)$, we
        need to adjust the $c_j$ and the $c_j'$.\\

\end{itemize}

For all the above representation models, we use the following genetic algorithm to adjust
the parameters:

\begin{algorithm}[H]
Step 1: Create an initial population of $60$ curves (PopSize = $60$). All initial curves are randomly generated.
    For piece-wise linear segments, the initial random curves are straight lines $a(t)=st+b$ where $s,b$ are drawn
    randomly from $U(-100,100)$. We choose a random slope from $U(-100,100)$ so that the resulting straight lines have similar {\em vertical swings} to the curves obtained via linear
    programming and neural networks.
        For polynomials and Fourier series, all parameters are drawn randomly from $U(-10,10)$.

Step 2:  Evolution:\newline
 \While{Have not reached the $100$-th round}{
 \For {Each Individual} {
        Randomly choose $200$ type profiles to test fitness (revenue)\;
 }

 Sort individuals according to fitness\;
 Selection:\newline
 Keep the top $20$ individuals in terms of fitness (EliteSize = $20$)\;

 Add new individuals through \textbf{Crossover} and \textbf{Mutation} until we reach Popsize\;
\textbf{1. Crossover:} Generate $20$ new curves via standard two-point crossover.\\
\textbf{2. Mutation:} Generate $20$ new curves via perturbing existing curves (with $0.1$ probability, a parameter is increased or decreased by $\delta=0.5$).
}
Step 3: Testing:\newline
 Average over $10000$ random type profiles to choose the best individual.\newline
 \caption{Genetic Algorithm}

\end{algorithm}

We will present the experimental results in Section~\ref{sec:experiments}.

\newpage
\section{Experiments}
\label{sec:experiments}

According to~ (\cite{Greenberg2012:Shopping}), an exploit that
attacks the Chrome browser sells for at most 200k for offensive clients
(USD). According to Google's official bug bounty reward program for the Chrome
browser~ (\cite{Projects2015:Severity}), a serious exploit is priced for at most 15k.
We adopt an experimental setting based on the
numbers above.

There are two agents.
The offender's valuation function is \[v(\theta_O,t) =
\int_0^t\theta_O(1-x)dx\]
\\
$\theta_O$ is drawn uniformly at random from $U(0,
400)$. That is, the offender's valuation for the whole time
interval $[0,1]$ is at most $200$.
The offender gets less and less interested in the exploit as time goes on (the instantaneous valuation gets to $0$ as $1-x$ approaches $0$ when $x$ approaches $1$).

The defender's valuation function is
\[v(\theta_D,t) = \int_t^1\theta_Dxdx\]
\\
$\theta_D$ is drawn uniformly at random from $U(0, 15)$. That is, the
defender's valuation for the whole time
interval $[0,1]$ is at most $15$.
The defender's instantaneous valuation in the exploit does not change over time.\\

% \textbf{Upper bound on the revenue: $53.75$}

% The above valuation functions satisfy all the conditions needed for the
% single-parameter model introduced in \cite{Guo2016:Revenue}.
% Based on the single-parameter optimal mechanism proposed in \cite{Guo2016:Revenue},
% we have that $MO^*$ simply sells the whole interval to the
% offender for a fixed price $p_O$, and $MD^*$ simply sells the whole interval to
% the defender for a fixed price $p_D$.

% \newpage
% \[ p_O = \arg\max_{p\le 200}pP(v(\theta_O,1)\ge p) = \arg\max_{p\le 200}p\frac{200-p}{200} = 100 \]
% \[ EPO(MO^*) = 50\]
% \[ p_D = \arg\max_{p\le 15}pP(v(\theta_D,0)\ge p) = \arg\max_{p\le 15}p\frac{15-p}{15} = 7.5\]
% \[ EPD(MD^*) = 3.75\]
% Therefore, the revenue upper bound according to \cite{Guo2016:Revenue} is $53.75$.\\

\textbf{Optimal revenue: $50.55$}

Since the valuation functions satisfy all the conditions needed for the
single-parameter model introduced in \cite{Guo2016:Revenue}, we are able to derive the revenue-maximizing
mechanism. The ending time is based on the following rule:
$$t^* = \arg\max\limits_{t \in [0,1]}  \{ (2 \times \theta_O - 400)(t-t*t/2) + (2 \times \theta_D - 15)(1-t)  \}$$

\begin{figure}[H]
\centering
\includegraphics[width=0.7\linewidth]{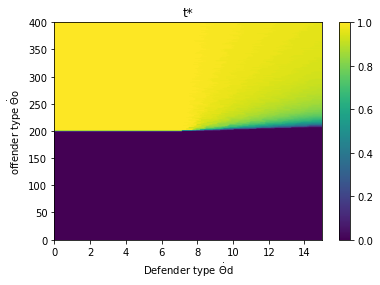}
    \caption[Experiment result: Ending time $t^*$ function]{Experiment result: Ending time $t^*$ function.}
\label{fig:map1}
\textbf{Ending time $t^*$ as a function of type profiles for the optimal mechanism}
\end{figure}

The optimal revenue equals $50.55$.

% \newpage
% \textbf{Revenue via iterative linear programming: $50.14$}

% We pick $k=10$ and $\epsilon=0.01$.
% We use the VCG mechanism as the initial solution. The VCG mechanism's expected
% revenue is $6.9$, which is very far away from the upper bound. Our technique
% starts from the VCG mechanism, and at the end produces a mechanism whose
% expected revenue equals $50.14$.

% \begin{figure}[H]
% \centering
% \includegraphics[width=.49\linewidth]{Figures/AMA2021/LP_payment.png}
% \includegraphics[width=.49\linewidth]{Figures/AMA2021/LP_t.png}
% \includegraphics[width=.6\linewidth]{Figures/AMA2021/LP_t_curve.png}
%     \caption[Experiment result: Total payment, ending time $t^*$, and $a(t)$ via linear programming]{Experiment result: Total payment, ending time $t^*$, and $a(t)$ via linear programming.}
% \textbf{Total payment, ending time $t^*$, and $a(t)$ as functions of type profiles
%     for the AMA mechanism derived via linear programming}
% \label{fig:map2}
% \end{figure}

% $\theta_O$(for offender) is drawn uniformly at random from $U(0, 400)$.
% \newline
% $\theta_D$(for defender) is drawn uniformly at random from $U(0, 15)$.

% an offender's valuation is (due to uniform distribution)
% $$v_{offender}(t)=v(\theta_O, t)= \int_{0}^{t} \theta_O(1-x)dx=\theta_O \times (t - t^{2} / 2)$$
% an defender's valuation is
% $$v_{defender}(t)=v(\theta_D, t)= \int_{t}^{1} \theta_Ddx=\theta_D \times (1 - t)$$
$$$$
\textbf{Revenue via neural networks: $50.31$}

% We use neural network to general carve a(t) with time frame t.
% Initially, we set constant value $u_{offender}=1.0$, and set up $u_{defender}=1.0$ with an auto grad parameter. Thus, we can get $t^*$ with:

We use the following setup:

\begin{itemize}
    \item We use a fully connected network with three hidden layers ($200$ nodes per layer) to represent $a(t)$.
    \item We use the Adam optimizer with a learning rate of $0.0001$.
    \item The batch size is set to $16$.
    \item The training set consists of $80000$ randomly generated type profiles.
    \item The testing set consists of $20000$ randomly generated type profiles.
\end{itemize}

The achieved expected revenue equals $50.31$.

\begin{figure}[H]
\centering
\includegraphics[width=.49\linewidth]{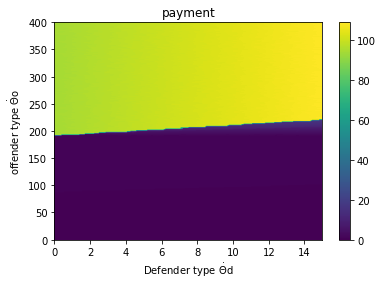}
\includegraphics[width=.49\linewidth]{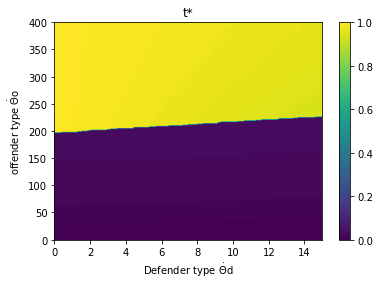}
\includegraphics[width=.6\linewidth]{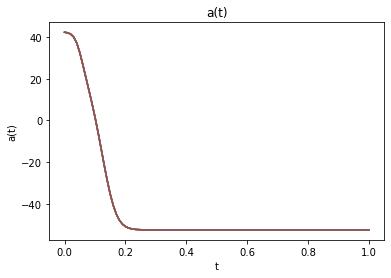}
    \caption[Experiment result: Total payment, ending time $t^*$, and $a(t)$ via neural network]{Experiment result: Total payment, ending time $t^*$, and $a(t)$ via neural network.}
\textbf{Total payment, ending time $t^*$, and $a(t)$ as functions of type profiles
    for the AMA mechanism derived via neural network}
\label{fig:map3}
\end{figure}

\newpage
\subsubsection{Main computational challenges for the neural networks}
Setting aside $u_{defender}$, every AMA mechanism is characterized by a {\bf
curve} $a(t)$, where $t\in [0,1]$. We note that
the main computational bottleneck is due to the learning batch size.  Our
fully-connected network structure with 3 layers and $200$ nodes per layer
is more than {\em
expressive} enough to represent curves.
In our
experiments, we generate $20,000$ type profiles to evaluate an AMA mechanism's
expected revenue. We cannot afford to do this in every learning iteration.
Instead, we use a batch size of $16$ during learning. That is, we use the
average revenue of $16$ randomly generated type profiles to estimate
the parameter gradient.
Certainly, the gradient direction obtained this way is not very accurate.
Fortunately, as
long as the learning rate is small, and as long as most of the time, the
gradient direction is generally correct, then we still are able to successfully
train the neural network. Figure~\ref{fig61} is an illustration of the neural
network training process. As shown in Figure~\ref{fig61}, the average revenue of
$16$ randomly generated type profiles perturbs wildly\footnote{In
Figure~\ref{fig61}, the presented data points for the batch data revenue are the
average revenue for every $100$ batches.}, but the mechanism is still
improving steadily during the training, which shows that a small batch is
good enough for training.

\begin{figure}[H]
\includegraphics[width=0.7\textwidth]{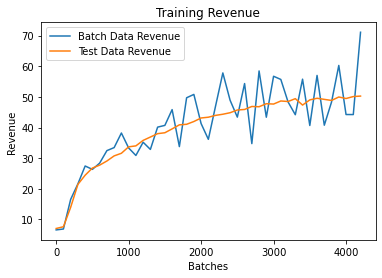} \centering
    \caption[Experiment result: Neural network training process]{Experiment result: Neural network training process.} \centering
\label{fig61}
\end{figure}

\newpage
\textbf{Revenue via evolutionary computation: $47.67$}

Our last numerical technique approximates $a(t)$ using segmented straight lines,
polynomials, and Fourier series. The parameters ({\em i.e.}, polynomial coefficients)
are adjusted according to an evolutionary algorithm (Algorithm~1).

\begin{table}[H]
    \caption[Experiment result: Evolutionary computation's results]{Experiment result: Evolutionary computation's results.}
\begin{tabular}{llllc}
% \textit{\textbf{Segmented straight line (30 segments): }} & 48.69             \\
\textit{\textbf{Segmented straight line (50 segments): }} & 47.67             \\
\textit{\textbf{Quartic polynomial: }} & 38.80             \\
\textit{\textbf{Sextic polynomial: }} & 37.20             \\
\textit{\textbf{Fourier series (N=5, p=2) : }} & 44.54             \\
\textit{\textbf{Fourier series (N=30, p=2) : }} & 46.88            \\

    \textbf{Revenue via evolutionary computation: different representations of $a(t)$}
\end{tabular}
\end{table}

We present all the $a(t)$ functions in the figure that follows.

\begin{figure}[!ht]
\centering  %图片全局居中
    \subfigure[Segmented straight lines (50 segments)]{
\label{Fig.sub.1}
\includegraphics[width=0.4\textwidth]{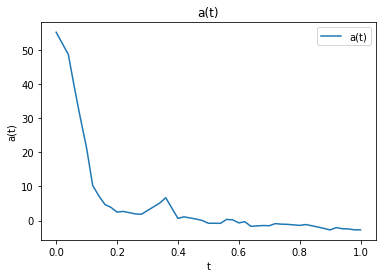}}
\subfigure[Quartic]{
\label{Fig.sub.2}
\includegraphics[width=0.4\textwidth]{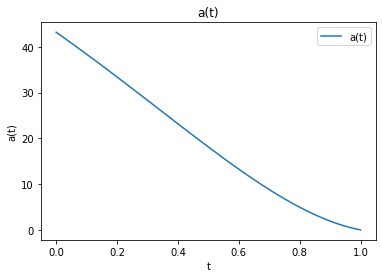}}
\subfigure[Sextic]{
\label{Fig.sub.3}
\includegraphics[width=0.4\textwidth]{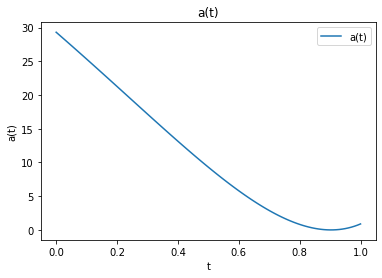}}
\subfigure[Fourier series N=5]{
\label{Fig.sub.4}
\includegraphics[width=0.4\textwidth]{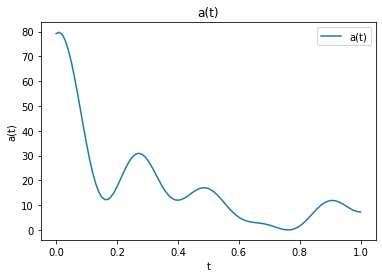}}
\subfigure[Fourier series N=30]{
\label{Fig.sub.5}
\includegraphics[width=0.4\textwidth]{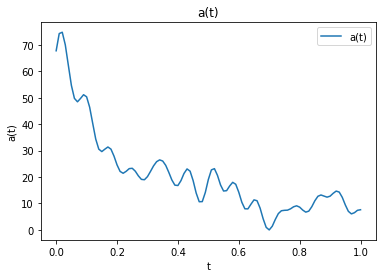}}
% \caption{$a(t)$ curve from different kinds of evolutionary algorithm}
\label{Fig.main}
\end{figure}

For Fourier series (N=30), the payment and ending time are as follows:

\begin{figure}[H]
\centering
\includegraphics[width=.49\linewidth]{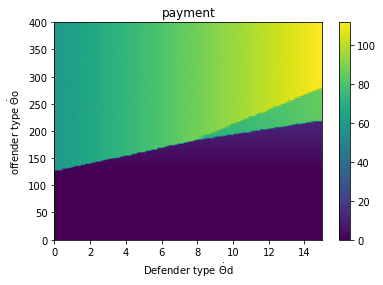}
\includegraphics[width=.49\linewidth]{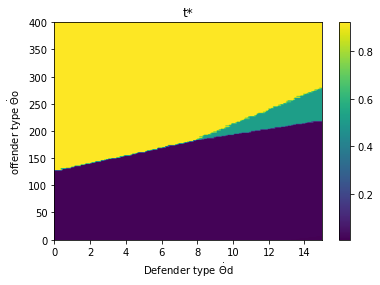}
    \caption[Experiment result: Total payment and ending time $t^*$ via evolutionary computation]{Experiment result: Total payment and ending time $t^*$ via evolutionary computation.}
\textbf{Total payment and ending time $t^*$ as functions of type profiles
    for the AMA mechanism derived via evolutionary computation (Fourier series N=30)}
\label{fig:map4}
\end{figure}

\subsection{Comparison of different AMA solution techniques}

The neural network based approach produces slightly better
result than the linear programming based approach.  On the other hand, the
neural network based approach realistically only works for two agents and
cannot deal with black-box valuation functions, as it requires the valuation
functions' analytical forms for auto differentiation.  The evolutionary
computation technique scales the best if we adopt a representation with small
number of parameters ({\em i.e.}, if we approximate $a(t)$ using a quadratic
polynomial such as $a(t)=c_2t^2+c_1t+c_0$, then we only need to adjust three parameters). The down-side of the evolutionary computation technique is that all the
representations we have tried (polynomials or Fourier series) are not as expressive
as neural networks, so the achieved revenue using the evolutionary computation
technique is slightly worse.

% \subsection{Compare}
% We run two algorithms in the same computer (i5-8300H, GTX1060). We find even we choose only 200 in each generation to test fitness, the genetic algorithm will also spend near 40 minutes. And the neural network only takes 15 minutes. On the other hand, And the NN' accuracy is good than the GA' accuracy. Because NN is improved by gradient descent. The GA is mainly formed by many parts of straight lines, and does some neighborhood mutation. However, GA revenue is greater than 45 in few generations.
% Thus, GA is hard to find the optimal curve, but it is easy to find a near optimal curve.

\newpage
\section{Chapter Summary}

In this chapter, we study markets for zero-day exploits from a revenue-maximizing mechanism design perspective. For the defenders, as long as any defender receives the exploit, all defenders are protected. So from the perspective of the defenders, the model studied in this chapter can be viewed as a public project model that is not excludable. Nevertheless, otherwise the model studied in this chapter is not related to a typical public project model.

We adopted the computationally feasible automated mechanism design approach.  We focused on the AMA mechanisms.  To identify an AMA mechanism with high revenue, we need to design an allocation ``curve''.  We propose two numerical solution techniques, one is based on neural networks and the other one is based on evolutionary computation.  All techniques are able to produce near-optimal mechanisms.

%% file: Chapters/Chapter7.tex
\chapter{Conclusion} % Main chapter title

\label{Conclusion} % Change X to a consecutive number; for referencing this chapter elsewhere, use \ref{ChapterX}

\section{Summary}
In this thesis, we used different machine learning methods to design optimal or near-optimal strategy-proof and individual rational mechanisms for public project problems.
Specifically, we addressed the following problems:

\begin{itemize}

\item In Chapter \ref{Mechanism Design for Public Projects via Neural Networks}, we study the public project that is indivisible and binary (e.g., a bridge). Indivisible means that an agent either consumes the project in its whole or is completely excluded. Binary means that this project is built or not built.
We identified a sufficient condition on the prior distribution for the conservative equal costs mechanism to be the optimal strategy-proof and individually rational mechanism.
For the non-excludable model, we designed novel mechanisms, such as dynamic programming, to get optimal results. For the excludable indivisible public project models, we involved several technical innovations that can be applied to mechanism design in general. We interpreted the mechanisms as price-oriented rationing-free (PORF) mechanisms. The experiments showed that our mechanisms are better than previous results and more close to the theoretical upper bound.\\

\item  In Chapter \ref{Public Project with Minimum Expected Release Delay}, we focused on the divisible public projects. In a classic excludable and binary public project model, we study a setting where the mechanism can set different project release times for different agents, which means that for a certain agent, the higher he/she pays, the earlier he/she can use the project. For a small number of agents, we proposed the sequential unanimous mechanisms by extending the existing mechanisms and used evolutionary computation to optimize them. We proposed the single deadline mechanisms which are shown to be asymptotically optimal. The experiments showed that our mechanisms are better than existing mechanisms.\\

\item In Chapter \ref{Redistribution in Public Project Problems via Neural Networks}, we studied the VCG redistribution mechanisms for the classic public project problem. Multi-layer perceptrons (MLP) were used in combination with a carefully designed cost function that takes into consideration of the agents' prior distributions for the optimal-in-expectation objective. We designed generative adversarial networks (GAN + MLP) to find the optimal worst-case VCG redistribution mechanisms. The experiments showed that our mechanisms are very close to the theoretical upper bounds and are better than existing mechanisms.\\

\item In Chapter \ref{Revenue-Maximizing Markets for Zero-Day Exploits},  we studied markets for zero-day exploits from a revenue-maximizing mechanism design perspective. We used a neural network to get the optimal curve that characterizes the optimal Affine Maximizer Auctions (AMA) mechanism. A second technique used evolutionary computation to evolve mathematical expressions for representing the optimal AMA curve. The experiments showed that our neural networks and evolutionary computation based techniques both produce near-optimal revenue.
\end{itemize}

% Here is a high-level summary of our automated mechanism design approaches in general.
% For a small number of agents, neural networks could be the best option to find the optimal or near-optimal mechanism. The neural network structures can be designed to ensure strategy-proofness and individual rationality. We can use model specific loss functions to achieve optimal or near-optimal results. Evolutionary computation uses every mechanism as an individual, ensures strategy-proofness and individual rationality during the genetic operator stages, and uses the objective functions of public project problems to achieve optimal or near-optimal results. Evolutionary computation can also make use of the prior distribution or history information. However, there are some drawbacks. Evolutionary computation is easy to get local minimums and the best result may not be as accurate as that of the neural networks'. For the evaluation of the mechanisms, we can compare our mechanisms to the state of the art ones.
% For a large number of agents, we can extend those mechanisms which we find for a small number of agents, and theoretically prove by mathematics that they are optimal or near-optimal. \\

$$$$
\section{Future Work}
In addition to the public project problems addressed in this thesis, we point out the following open problems that we expect to explore in the future:\\
\begin{itemize}

\item Apply our methods to other models: The techniques proposed in this thesis can not only be used in public project problems, but also have the potential to be applied to other economic models. One future direction is to extend our methods to other models such as cake cutting, facility location, and auctions.\\

\item Design better neural network structures for mechanism design: In this thesis, we carefully designed our neural networks to ensure strategy-proofness and individual rationality. However, most of our design focused on the mechanism design aspect, not on the neural networks themselves. That is, neural network is used as an assisting tool and it is heavily guided by manual human inputs. Most of our networks are straight-forward feed-forward fully-connected networks. One future direction is to consider more sophisticated neural network structures such as pointer networks, LSTM, permutation invariant networks, etc.

\item Combinatorial public project model: Through out this thesis, we assumed that there is only one public project. However, in practise, it is common that we face many public projects and the agents face combinatorial decision making. One future research direction is to see to what extent our techniques and results generalise beyond a single public project.

\end{itemize}